\documentclass[runningheads]{llncs}

\usepackage{graphicx}
\usepackage{cutelmet}
\usepackage{proof}
\usepackage{latexsym}
\usepackage{bussproofs}
\usepackage{enumitem}
\usepackage{framed}
\usepackage[utf8]{inputenc}
\usepackage{lmodern}
\usepackage{hyperref}
\usepackage{qtree}
\usepackage{url}
\usepackage{soul}
\usepackage{color}
\usepackage{xspace}
\usepackage{verbatim}
\usepackage{tikz}
\usetikzlibrary{arrows,automata}
\usepackage{mathtools}
\usepackage{stmaryrd}
\usepackage{amsfonts}
\usepackage{amssymb}
\usepackage{amsmath}
\usepackage{listings}
\usepackage{verbatim}
\usepackage{float}

\newcommand{\Apar}{{\cal S}}
\newcommand{\regFlow}[1]{\ensuremath{\mathcal{P}^{\star}(#1)}}

\newcommand{\eqpred}[2]{\ensuremath{ f(#1) = #2}}
\newcommand{\neqpred}[2]{\ensuremath{ f(#1) \not = #2}}
\newcommand{\leqpred}[2]{\ensuremath{ f(#1) < #2}}
\newcommand{\nleqpred}[2]{\ensuremath{ f(#1) \not < #2}}
\newcommand{\orra}[1]{\overrightarrow{#1}}
\newcommand{\Nat}{\mathbb{N}}
\begin{document}
\title{Schematic Refutations of Formula Schemata}
\author{David Cerna\inst{1} \and
Alexander Leitsch\inst{2} \and
Anela Lolic\inst{3}}
\authorrunning{D. Cerna et al.}
%
\institute{Institute for Formal Methods and Verification, JKU, Linz, Austria\\
Research Institute for Symbolic Computation, JKU, Hagenberg, Austria\\
\email{david.cerna@jku.at,david.cerna@risc.jku.at}
\and
Institute of Logic and Computation, TU Wien, Vienna, Austria\\
\email{leitsch@logic.at} \and
Institute of Logic and Computation, TU Wien, Vienna, Austria\\
\email{anela@logic.at}}
\maketitle              
\begin{abstract}
Proof schemata are infinite sequences of proofs which are defined inductively. In this paper we present a general framework for schemata of terms, formulas and unifiers and define a resolution calculus for schemata of quantifier-free formulas. The new calculus generalizes and improves former approaches to schematic deduction. As an application of the method we present a schematic refutation formalizing a proof of a weak form of the pigeon hole principle.

\end{abstract}
\section{Introduction}

Recursive definitions of functions play a central role in computer science, particularly in functional programming. While recursive definitions of proofs are less common they are of increasing importance in automated proof analysis. 
Proof schemata, i.e. recursively defined infinite sequences of proofs, serve as an alternative formulation of induction. Prior to the formalization of the concept, an analysis of F\"{u}rstenberg's proof of the infinitude of primes~\cite{MBaaz2008a} suggested the need for a formalism quite close to the type of proof schemata we will discuss in this paper. The underlying method for this analysis was CERES~\cite{MBaaz2000} (cut-elimination by resolution) which, unlike reductive cut-elimination, can be applied to recursively defined proofs by extracting a schematic unsatisfiable formula and constructing a recursively defined refutation. Moreover, Herbrand's theorem can be extended to an expressive fragment of proof schemata, that is those formalizing $k$-induction~\cite{CDunchev2014,ALeitsch2017}. Unfortunately, the construction of recursively defined refutations is a highly complex task. In previous work~\cite{ALeitsch2017} a superposition calculus for certain types of formulas was used for the construction of refutation schemata, but only works for a weak fragment of arithmetic and is hard to use interactively. 

The key to proof analysis using CERES in a first-order setting is not the particularities of the method itself, but the fact that it provides a bridge between automated deduction and proof theory. In the schematic setting, where the proofs are recursively defined, a bridge over the chasm has been provided ~\cite{CDunchev2014,ALeitsch2017}, but there has not been much development on the other side to reap the benefits of. The few existing results about automated deduction for recursively defined formulas barely provide the necessary expressive power to analyse significant mathematical argumentation. Applying the earlier constructions to a weak mathematical statement such as the {\em eventually constant schema} required much more work than the value of the provided insights~\cite{DCerna2016}. The resolution calculus we introduce in this work generalizes resolution and the first-order language in such a way that it provides an excellent environment for carrying out investigations into decidable fragments of schematic propositional formulas beyond those that are known. Furthermore, concerning the general unsatisfiability problem for schematic formulas, our formalism provides a perfect setting for interactive proof construction. 

Proof schema is not the first alternative formalization of induction with respect to Peano arithmetic~\cite{GTakeuti1975}. However, all other existing examples~\cite{JBrotherston2005,JBrotherston2010,RMcdowell1997} that provide calculi for induction together with a cut-elimination procedure do not allow the extraction of Herbrand sequents\footnote{Herbrand sequents allow the representation of the propositional content of first-order proofs.}~\cite{SHetzl2008,GTakeuti1975} and thus Herbrand's theorem cannot be realized. In contrast, in \cite{ALeitsch2017} finite representations of infinite sequences of Herbrand sequents are constructed, so-called {\em Herbrand systems}. Of course, such objects do not describe finite sets of ground instances, though instantiating the free parameters of Herbrand systems does result in sequents derivable from a finite set of ground instances. 

The formalism developed in this paper extends and improves the formal framework for refuting formula schemata in~\cite{CDunchev2014,ALeitsch2017} in several ways: 1. The new calculus can deal with arbitrary quantifier-free formula schemata (not only with clause schemata), 
2. we extend the schematic formalism to  multiple parameters (in~\cite{CDunchev2014} and in~\cite{ALeitsch2017} only schemata defined via one parameter were admitted); 3. we strongly extend the recursive proof specifications by allowing mutual recursion (formalizable by so-called call graphs). Note that  in~\cite{CDunchev2014} a complicated schematic clause definition was used, while the schematic refutations in~\cite{ALeitsch2017} were based on negation normal forms and on a complicated translation to the $n$-clause calculus. Moreover, the new method presented in this paper provides a simple, powerful and elegant formalism for interactive use. The expressivity of the method is illustrated by an application to a (weak) version of the pigeon hole principle. 

\section{A Motivational Example}
\label{motivation}
In~\cite{DCerna2016}, proof analysis of a mathematically simple statement, the {\em Eventually Constant Schema}, was performed using an early formalism developed for schematic proof analysis~\cite{CDunchev2014}. The Eventually Constant Schema states that any monotonically decreasing function with a finite range is eventually constant. The property of being eventually constant may be formally written as follows: 
\begin{equation}
\label{endconstate}
\exists x\forall y ( x\leq y \rightarrow f(x) = f(y)),
\end{equation}
where $f$ is an uninterpreted function symbol with the following property $$ \forall x\left( \bigvee_{i=0}^{n} f(x)= i\right) $$
for some $n\in\mathbb{N}$. The method defined in~\cite{CDunchev2014} requires a strong quantifier-free end sequent, thus implying the proof must be skolemized. The skolemized formulation of the eventually constant property is $\exists x( x\leq g(x) \rightarrow f(x) = f(g(x)))$ where $g$ is the introduced Skolem function. The proof presented in~\cite{DCerna2016} used a sequence of $\Sigma_2$-cuts $$\exists x \forall y (((x \leq y) \Rightarrow n + 1 = f(y)) \vee f(y) < n + 1).$$ Also, the Skolem function was left uninterpreted for the proof analysis. The resulting cut-structure, when extracted as an unsatisfiable clause set, has a fairly simple refutation. Thus, with the aid of automated theorem provers, a schema of refutations was constructed. 

The use of an uninterpreted Skolem function greatly simplified the construction presented in~\cite{DCerna2016}. In this paper we will interpret the function $g$ as the successor function. Note that using the axioms presented in~\cite{DCerna2016} the following statement 
$$ \forall x\left( \bigvee_{i=0}^{n} f(x)= i\right)  \vdash \exists x(  f(x) = f(\mathit{suc}(x)))$$ is not provable. Note that we drop the implication of Equation~\ref{endconstate} and the antecedent of the implication given that $x\leq \mathit{suc}(x)$ is a trivial property of the successor function.  However, using an alternative set of axioms  and a weaker cut we can prove this statement.  The additional axioms are as follows:
 $$f(x)=i \vdash  f(x) < s(k) \ , \ \mbox{for $0\leq i \leq k < n$ }$$
 $$f(\mathit{suc}(x))=i \vdash  f(x) < s(k) \ ,\ \mbox{for $0\leq i \leq k <n$ }$$
 $$ f(x)=k , f(\mathit{suc}(x)) =k \vdash f(x)  = f(\mathit{suc}(x)) \ ,\ \mbox{for $0\leq k \leq n$ } $$
$$f(0) < 0 \vdash $$
$$f(\mathit{suc}(x)) < s(k)  \vdash f(\mathit{suc}(x))=k , f(x) < k \ ,\ \mbox{for $0\leq k \leq n$ } $$
$$f(x) < s(k)  \vdash  f(x)=k , f(x) < k \ ,\ \mbox{for $0\leq k < n$ } $$
For the most part these axioms are harmless, however $f(\mathit{suc}(x)) < s(k)  \vdash f(\mathit{suc}(x))=k , f(x) < k$ implies that $f$ has some  monotonicity properties similar to the eventually constant schema. For all values $y$ in the range larger than $x$, either $f(y) = n$ or $f(y) < n$. Our axioms imply that once $f(y) <n$ all values $z>y$ are mapped to values of the range less than $n$. Being that our proof enforces this property using the following $\Delta_2$-cut formula we are guaranteed to reach a value in the domain above which $f$ is constant. The cut formula is as follows: 
$$ \exists x( f(x) = k\wedge k = f(\mathit{suc}(x))) \vee \forall x (f(x)<k) \ ,\ \mbox{for } 0\leq k \leq n.$$
One additional point which the reader might notice is that we use what seems to be the less than relation and equality relation of the natural numbers, but do not concern ourselves with substitutivity of equality nor transitivity of $<$. While including these properties will change the formal proof we present below, the argument will still require a free numeric parameter denoting the size of the range of $f$ and the number of positions we require to map to the same value.

We will refer to this version of the eventually constant schema as the {\em successor eventually constant schema}. While this results in a new formulation of the eventually constant schema under an interpretation of the Skolem function as the successor function, we have not taken complete advantage of this new interpretation yet in that this re-formulation is actually of lower complexity than the eventually constant schema. For example in Figure~\ref{simpleproof} we provide the output of Peltier's superposition induction prover\cite{VAravantinos2013} when ran on the clausified form of thee cut structure of the successor eventually constant schema. The existence of this derivation implies that the proof analysis method of~\cite{ALeitsch2017} may be applied to the successor eventually constant schema. Unfortunately, the prover does not find the invariant discovered in~\cite{DCerna2016}, but this may have more to do with the choice of axioms rather than the statement being beyond the capabilities of the prover. 

\begin{figure}

\begin{verbatim}
============================== PROOF =================================

% Proof 1 at 0.017 (+ 0.000) seconds.
% Given clauses 73.

% number of calls to fixpoint : 3 
 S_init  : 
 (51:  [ EQ(v0,f(h(v1))) | LE(f(v1),v0) if  n = s(v0) ].
50:  [ EQ(v0,f(v1)) | LE(f(v1),v0) if  n = s(v0) ].
33:  [ PHI(v0,v1) if  n = s(v0) ].
) 
 S_loop  : 
 (82:  [ EQ(v0,f(h(v1))) | LE(f(v1),v0) if  n = s(s(v0)) ].
80:  [ EQ(v0,f(v1)) | LE(f(v1),v0) if  n = s(s(v0)) ].
53:  [ PHI(v0,v1) if  n = s(s(v0)) ].
) 
 The empty clauses  : 
 (45:  [  n = 0 ].
81:  [  n = s(0) ].
112:  [  n = s(s(0)) ].
) max_rank 2 

============================== end of proof ==========================
\end{verbatim}

\caption{output of Peltier \textit{et al.}'s Prover9 extension~\cite{DBLP:phd/hal/Kersani14}.}
\label{simpleproof}
\end{figure}
We can strengthen the successor eventually constant schema beyond the capabilities of~\cite{DBLP:phd/hal/Kersani14} easily by adding a second parameter as follows:
$$ \forall x\left( \bigvee_{i=0}^{n} f(x)= i\right)  \vdash \exists x( \bigwedge_{i=0}^{m} f(x) = f(\mathit{suc}^i(x))).$$ 
We refer to this problem as the $m$-successor eventually constant schema. Applying this transformation to the the eventually constant schema of~\cite{DCerna2016} is not so trivial being that the axioms used to construct the proof do not easy generalize. However, for the successor eventually constant schema the generalization is trivial and is provided below: 
$$f(\mathit{suc}^{r}(x))=i \vdash  f(x) < s(k) \ ,\ \mbox{for $0\leq i \leq k \leq n$ and $0\leq r\leq m$}$$
$$ f(x)=k , f(\mathit{suc}^{r}(x)) =k \vdash f(x)  = f(\mathit{suc}^{r}(x)) \ ,\ \mbox{for $0\leq k \leq n$ and $0<r\leq m$} $$
$$f(0) < 0 \vdash $$
$$f(\mathit{suc}^{r}(x)) < s(k)  \vdash f(\mathit{suc}^{r}(x))=k , f(x) < k \ ,\ \mbox{for $0\leq k \leq n$ and $0\leq r\leq m$},$$
where $\mathit{suc}^{0}(x) = x$. Furthermore, the cut formula can be trivial extended as follows: 
$$ \exists x( \bigwedge_{i=0}^{m} f(\mathit{suc}^{i}(x)) = k) \vee \forall x (f(x)<k) \ ,\ \mbox{for } 0\leq k \leq n.$$
Given that the $m$-successor eventually constant schema contains two parameters it is beyond the capabilities of~\cite{DBLP:phd/hal/Kersani14}. Interesting enough, the prover can find invariants for each value of $m$ in terms of $n$, though, these invariant get impressively large quite quickly. The cut structure of the $m$-successor eventually constant schema by be extracted as an inductive definition of an unsatisfiable negation normal form formula. We provide this definition by rewrite rules below: 
\begin{align*}
O(n,m) \equiv &  D(n,m)  \wedge\  P(n,m)\\
C(y,n,0)  \equiv&  \neqpred{S(0,y)}{n}\\ 
C(y,n,s(m))  \equiv&  \neqpred{S(s(m),y)}{n} \ \vee \ C(y,n,m)\\
T(n,0) \equiv&  \forall x(\nleqpred{S(0,x)}{s(n)}  \ \vee \ \eqpred{S(0,x)}{n} \ \vee \ \leqpred{x}{n})\\
T(n,s(m)) \equiv&   \forall x (\nleqpred{S(s(m),x)}{s(n)}  \ \vee \ \eqpred{S(s(m),x)}{n} \vee \leqpred{x}{n}) \\ & \wedge \ T(n,m)\\
P(0,m) \equiv&   \forall x(C(x,0,m)) \ \wedge\ \nleqpred{a}{0} \\
P(s(n),m) \equiv&  ( \forall x(C(x,s(n),m))  \ \wedge  (T(n,m)) \ \wedge\ P(n,m) \\
D(n,0) \equiv&  \forall x(\eqpred{S(0,x)}{n} \ \vee \  \leqpred{x}{n}) \\
D(n,s(m)) \equiv&  \forall x( \eqpred{S(s(m),x)}{n} \ \vee \  \leqpred{x}{n} ) \wedge \  D(n,m)\\
S(0,y)\equiv&  y\\
S(s(n),y) \equiv& \mathit{suc}(S(s(n),y))\\
\end{align*}
where $a$ is some arbitrary constant.  We will show how our new formalism can provide a finite representation of the refutations of inductive definition even though the refutation requires mutual recursion as well as multiple parameters (six in total) to state finitely.

\section{Schematic Language}\label{sec.schematicproofs}

We work in a two-sorted version of classical first-order logic. The first sort we consider is $\omega$, in which every ground term normalizes to a {\em numeral}, i.e. a term inductively constructable by  $N \Rightarrow s(N) \ | \ 0$, such that $s(N) \not= 0$ and $s(N) = s(N') \to N = N'$. Natural numbers ($\mathbb{N}$)  will be denoted by lower-case Greek letters ($\alpha$, $\beta$, $\gamma$, etc); The numeral $s^\alpha 0$, $\alpha \in \mathbb{N}$, will be written as $\bar{\alpha}$.  The set of numerals is denoted by $\Num$. Furthermore, the $\omega$ sort includes a countable set of variables $\mathcal{N}$ called parameters. We denote parameters by $n,m,n_1,n_2,\ldots,m_1,m_2,\ldots$. The set of parameters occurring in an expression $E$ is denoted by $\Ncal(E)$.

The second sort, the $\iota$-sort (individuals), is a standard first-order term language extended by {\em defined function symbols}. Defined function symbols, i.e. primitive recursively defined functions,  are differentiated from standard function symbols by adding a $\hat{\cdot}$.

Variables of the $\iota$-sort are what we refer to as {\em global} variables.  Global variables are indexed by terms of the $\omega$-sort and are of type $\omega^{\alpha} \to \iota$ for $\alpha \geq 0$.  The set of global variables will be denoted by $V^G$. Given a sequence of $\omega$ terms $\overrightarrow{t}$ of length $\alpha$ and a global variable X of type $\omega^{\alpha} \to \iota$, we will refer to the $\iota$ terms $X(\overrightarrow{t})$ as {\em V-terms over  $X$}. The set of {\em individual variables} $V^{\iota}$ consists of objects of the form $X(\overrightarrow{t})$ where $\overrightarrow{t}$ is a tuple of numerals. Additionally we introduce $V^F$, the set of formula variables of type $o$ which aid the construction {\em defined predicate symbols}.

For terms we consider the set of function symbols of type $\tau$, $\Fcal^\tau$.  The set of {\em defined function symbols} of type $\tau$ is denoted 
by $\Fcalhat^\tau$.  The types $\tau$ are either of the form $\omega^\alpha \to \omega$ (for $\alpha \in \omega$)  which we call {\em numeric types} 
or of type $\iota^\alpha \to \iota$ for $\alpha \geq 0$ which we call {\em individual types}. We distinguish $\Fcalhat_\omega$ - 
the set of all defined function symbols of numeric type and $\Fcalhat_\iota$ - the set of all defined function symbols of individual type. We define 
$\Fcal^\omega= \{\bar{0}\}$, $\Fcal^{\omega \to \omega}= \{s\}$, and $ \Fcal^\tau = \emptyset$ for  $\alpha>1$. We leave the symbol sets of complex $
\omega$ types empty so that normalization of a ground term always results in a numeral. For the same reason the only constants of type $\omega$ are $0$ and $s$. Basic 
functions such as {\em projections} will be introduced as defined function symbols. For all other types $\tau$ the set $\Fcal^\tau$ is infinite. This 
extends to the  defined symbols of any type, i.e.  $\Fcalhat^{\tau}$  is infinite. For defined function symbols we make use of global variables: the types $\tau$ are $(\omega^{\gamma_1} \to  \iota) \times \ldots \times (\omega^{\gamma_\alpha} \to  \iota) \times \omega^\beta \to \iota$ for $\alpha \geq 0$ and $\gamma_i> 0$ for $i \in \{1, \ldots , \alpha\}$. The symbols in $\Fcalhat_\omega$ and $\Fcalhat_\iota$ are partially ordered by $\orderf$ where $\orderf$ is irreflexive, transitive and Noetherian.

We define a similar signature for predicate symbols of type $\tau$, where $\Pcal^\tau_i$ is the (infinite) set of predicate symbols of type $\tau$; 
the set of {\em defined predicate symbols} of type $\tau$ is denoted by $\Pcalhat^\tau $. For ordinary ($\alpha$-ary) predicate symbols the types are $\iota^\alpha \to o$, for $\alpha \geq 0$ as usual. Like for defined function symbols we make use of global and context variables for defined predicate symbols: the types $\tau$ are $(\omega^{\gamma_1} \to  \iota) \times \ldots \times (\omega^{\gamma_\alpha} \to  \iota) \times \omega^\beta \to o$ for $\alpha \geq 0$ and $\gamma_i> 0$ for $i \in \{1, \ldots , \alpha\}$. The symbols in $\Pcalhat$ are partially ordered by $\orderp$ where $\orderp$ is irreflexive, transitive and Noetherian.

For the term language we consider {\em $\omega$-terms} of type $\omega$ and {\em $\iota$-terms} of type $\iota$. Both term sets are defined via function symbols and defined function symbols.
\begin{definition}[$\omega$-terms $T^\omega$ ]\label{def.omegaterms}\mbox{ }
\begin{itemize}
\item[(1)] $\bar{0} \in T^\omega,\ \mathcal{N} \IN T^\omega$, and  if $t \in T^\omega$ then $s(t) \in T^\omega$,
\item[(2)] if $\fhat \in \Fcalhat_{\omega}^\tau$ for $\tau = \omega^\alpha \to \omega$  and $t_1,\ldots,t_\alpha \in T^\omega$ then $\fhat(t_1,\ldots,t_\alpha) \in T^\omega$.
\end{itemize}
The set $T^\omega_0$ denotes terms constructed using (1). Note that the set  of parameter-free terms in $T^\omega_0$ is $\Num$, the set of numerals.
\end{definition}
For every defined function symbol $\fhat \in \Fcalhat_\omega$ there exists a set of defining equations $D(\fhat)$ which expresses a primitive recursive definition of $\fhat$.
\begin{definition}[defining equations for numeric function symbols]\label{def.def-equ-f}
For every $\fhat \in \Fcalhat_\omega$, $\fhat\colon \omega^{\alpha+1} \to \omega$ we define a set $D(\fhat)$ consisting of two equations:
$$\fhat(n_1,\ldots,n_\alpha,\overline{0}) = t_B,  \quad \fhat(n_1,\ldots,n_\alpha,s(m)) = t_S\{k \ass \fhat(n_1,\ldots,n_\alpha,m)\},$$
where for minimal $\fhat$ $t_B,t_S \in T^\omega_0$; for nonminimal $\fhat$ 
$t_B,t_S \in T^\omega$ where $t_B,t_S$ may contain only defined function symbols smaller than $\fhat$ in $\orderf$. Furthermore $\Ncal(t_B) \subseteq \{n_1,\ldots,n_\alpha\}$, and $\Ncal(t_S) \subseteq \{n_1,\ldots,n_\alpha\} \union \{m,k\}$.

We define $D(\Fcalhat_\omega) = \Union\{D(\fhat) \mid \fhat \in \Fcalhat_\omega\}$, which is the set of all defining equations in the numeric types.
\end{definition}
\begin{example}\label{ex.Tomega}
For $\widehat{p}\in \Fcalhat^{\omega \to \omega}$, $D(\widehat{p}) = \left\lbrace \widehat{p}(\bar{0}) = \bar{0},\
\widehat{p}(s(m)) = m\right\rbrace $,  $t_B = \bar{0}$,  $t_s = m$.\\[1ex]
Let $\fhat,\ghat \in \Fcalhat^\tau$ for $\tau = \omega \times \omega \to \omega$, $\fhat$ be minimal and $\fhat \orderf \ghat$. We define $D(\fhat)$ as 
$$\fhat(n,\bar{0}) = t_B,\quad \fhat(n,s(m)) = t_S\{k \ass \fhat(n,m)\}$$
for $t_B = n$ and $t_S = s(k)$. Then, obviously, $\fhat$ defines $+$. \\[1ex]
Now we define $D(\ghat)$ as 
$$\ghat(n,\bar{0}) = t'_B,\quad \ghat(n,s(m)) = t'_S\{k \ass \ghat(n,m)\}$$
where $t'_B = \bar{0}$ and $t'_S = \fhat(n,k)$. Then $\ghat$ defines $*$.
\end{example}
It is easy to see that, given any parameter assignment, all terms in $T^\omega$ evaluate to numerals. 
\begin{definition}[parameter assignment]
A function $\sigma\colon \Ncal \to \Num$ is called a {\em parameter assigment}.
$\sigma$ is extended to terms homomorphically:
\begin{itemize}
\item $\sigma(\bar{\beta}) = \bar{\beta}$ for numerals $\bar{\beta}$.
\item $\sigma(\fhat(t_1,\ldots,t_\alpha)) = \fhat(\sigma(t_1),\ldots,\sigma(t_\alpha))$ 
for $\fhat\colon \omega^\alpha \to \omega$ and $t_1,\ldots,t_\alpha \in T^\omega$.
\end{itemize}
The set of all parameter assigments is denoted by $\Apar$.
\end{definition}
To simplify notation we use the following convention: if $\sigma \in \Apar$ and $\vec{n} = (n_1,\ldots,$ $n_\alpha)$ we write  $\sigma(\vec{n})$ for $(\sigma(n_1),\ldots,$ $\sigma(n_\alpha))$.
\begin{definition}[rewrite system $R(\Fcalhat_\omega)$]\label{def.RFcalhat}
Let $\fhat \in \Fcalhat_\omega$. Then $R(\fhat)$ is the set of the following rewrite rules obtained from $D(\fhat)$: 
$$\fhat(n_1,\ldots,n_\alpha,\hat{0}) \to t_B,  \quad \fhat(n_1,\ldots,n_\alpha,s(m)) \to t_S\{k \ass \fhat(n_1,\ldots,n_\alpha,m)\}$$
$R(\Fcalhat_\omega) = \Union\{R(\fhat) \mid \fhat \in \Fcalhat_\omega\}$. When a numeric term $s \in T^\omega$ rewrites to $t$ under $R(\Fcalhat_\omega)$ we write $s \toomega t$.
\end{definition}
\begin{proposition}\label{prop.evalTomega}
\mbox{ }.
\begin{itemize}
\item $R(\Fcalhat_\omega)$ is a canonical rewrite system.
\item  Let $t \in T^\omega$ and $\sigma \in \Apar$. Then the (unique) normal form of $\sigma(t)$ under $R(\Fcalhat_\omega)$  (denoted by $\sigma(t)\Eval_\omega$) is a numeral .
\end{itemize}
\end{proposition}
\begin{proof}
Straightforward: termination and confluence of $R(\Fcalhat_\omega$) are well known. In particular $\bar{0},s$ and $R(\Fcalhat)$ define a language for computing the set of primitive recursive functions; in particular the recursions are well founded. A formal proof of termination requires double induction on $\orderf$ and the value of the recursion parameter.
\end{proof}
An important remark concerning global variables and normalization is as follows: Let $t_1,\cdots t_{\alpha},s_1,\cdots s_{\alpha} \in \omega$, $X,Y \in V^G$ of type $\omega^{\alpha}\to \iota$,  then we say $X(t_1,\cdots t_{\alpha}) = Y(s_1,\cdots s_{\alpha})$ iff $X=Y$ and for any parameter assignment $\sigma$ we have $\sigma(t_i)\Eval_\omega = \sigma(s_i)\Eval_\omega$ for $1\leq i \leq \alpha$.
\begin{definition}[the $\iota$-terms $T^\iota$]\label{def.Tiota}
The set $T^\iota$ is defined inductively as follows:
\begin{itemize}
\item all constants of type $\iota$ are in $T^\iota$,
\item for all $X:\omega^{\alpha}\to \iota \in V^{\iota}$, for $\alpha \geq 0$ and $t_1,\cdots, t_{\alpha} \in T^\omega$, $X(t_1,\cdots, t_{\alpha}) \in T^\iota$, i.e.  $X(t_1,\cdots, t_{\alpha})$ is a $V$-term over $X$. The set of all $V$-terms is denoted by $T^\iota_V$.
\item if $f \in \Fcal$, $f\colon \iota^\alpha  \to \iota$, $s_1,\ldots,s_\alpha \in T^\iota$ then $f(s_1,\ldots,s_\alpha) \in T^\iota$, 
\item Let $\fhat \in \Fcalhat$, $\fhat\colon \tau(\gamma(1),\ldots,\gamma(\alpha))\times \omega^{\beta+1} \to \iota$, where $\tau(\gamma(1),\ldots,\gamma(\alpha)) = (\omega^{\gamma(1)} \to \iota)\times \ldots \times (\omega^{\gamma(\alpha)} \to \iota)$. If $X_i \in V^G$ and $X_i\colon \omega^{\gamma(i)} \to \iota$, and $t_1,\ldots,t_{\beta+1} \in T^\omega$ then 
$\fhat(X_1,\ldots,X_{\alpha},t_1, \ldots,t_{\beta+1}) \in T^\iota$.
\end{itemize}
The set of all terms in $T^\iota$ which contain no defined symbols and no parameters is denoted by $T^\iota_0$. $T^\iota_0$ is a set of ``ordinary'' first-order terms.

In the following definitions we will abbreviate sequences of $\omega$-terms of length $\beta$ by $\overrightarrow{t}_{\beta}$ and sequences of instantiated global variablesf $X_1(\overrightarrow{t_1}_{\beta_1}), \cdots X_{\alpha}(\overrightarrow{t_{\alpha}}_{\beta_{\alpha}})$ as $\overrightarrow{X(\overrightarrow{t}_{\beta_i})}^{\alpha}_{i}$. When dealing with uninstantiated global variables we just write $\overrightarrow{X}_{\alpha}$. 
\end{definition}
\begin{definition}[defining equations for $\iota$-symbols]\label{def-eq-iota}
Let $\fhat \in \Fcalhat_{\iota}^\tau$ for $\tau = \tau(\gamma(1),\ldots,\gamma(\alpha)) \times \omega^{\beta+1} \to \iota$  where $\alpha\geq 0$ and $1\leq i\leq \alpha$, $\gamma_i >0$. The defining equations $D(\fhat)$ are  defined below.
\begin{eqnarray*}
\fhat(\overrightarrow{X}_{\alpha},\overrightarrow{n}_\beta,\bar{0}) &=& t_B,\\
\fhat(\overrightarrow{X}_{\alpha},\overrightarrow{n}_\beta,s(m)) &=& t_S\{Y \ass \fhat(\overrightarrow{X}_{\alpha},\overrightarrow{n}_\beta,m)\},
\end{eqnarray*}
where  $Y \colon \iota$. For minimal $\fhat$ $t_B$ is a term of type $\iota$ such that V-terms are over the variables $\overrightarrow{X}_{\alpha},Y$, $\Ncal(t_B) \IN \{n_1,\ldots,n_\beta$\} and $t_B$ contains no defined symbols from $\Fcalhat^\tau$ for nonnumeric types $\tau$. For nonminimal $\fhat$, $t_B$ may contain defined symbols $\ghat$ of type $\tau(\gamma'(1),\ldots,\gamma'(\alpha')) \times \omega^{\beta'+1} \to \iota$  with $\ghat \orderf \fhat$.\\[1ex]
$t_S$ is a term of $T^\iota$ such that  $T^\iota_V(t_B)$ only contains V-terms over the variables $\overrightarrow{X}_{\alpha},Z$ and $\Ncal(t_B) \IN \{n_1,\ldots,n_\beta\} \union \{m\}$. For all defined symbols $\ghat$  of type $\tau(\gamma'(1),\ldots,\gamma'(\alpha'_1))\times \omega^{\beta'+1} \to \iota$  occurring in $t_S$ we must have $\ghat \orderf \fhat$.\\[1ex]
Like for the numeric terms we define $D(\Fcalhat_\iota) = \Union \{D(\fhat) \mid \fhat \in \Fcalhat_\iota\}$.
\end{definition}

\begin{example}\label{ex.Tiota}
Let $g \in \Fcal^{\iota\to (\iota\to \iota)}$ and $\fhat \in \Fcalhat^{(\omega\to\iota) \times \omega \to \iota}$. We define $D(\fhat)$  as 
$$ \fhat(X,0)  = X(0),\ \ \fhat(X,m+1) = g(X(m+1),\fhat(X,m)).$$
Here, $t_B = X(0), t_S = g(X(m+1),Y)$.
\end{example}
While numeric terms evaluate to numerals under parameter assignments, terms in $T^\iota$ evaluate to terms in $T^\iota_0$, i.e. to ordinary first-order terms. Like for the terms in $T^\omega$ the evaluation is defined via a rewrite system. 

\begin{definition}[rewrite system $R(\Fcalhat_\iota)$]\label{def.rewrite-iota}
Let $\fhat \in \Fcalhat_\iota$. Then $R(\fhat)$ is the set of the following rewrite rules obtained from $D(\fhat)$: 
\begin{eqnarray*}
\fhat(\overrightarrow{X}_{\alpha_1},\overrightarrow{n}_\beta,\bar{0}) &\to& t_B,\\
\fhat(\overrightarrow{X}_{\alpha_1},\overrightarrow{n}_\beta,s(m)) &\to& t_S\{Y \ass \fhat(\overrightarrow{X}_{\alpha_1},\overrightarrow{n}_\beta,m)\}
\end{eqnarray*}
$R(\Fcalhat_\iota) = \Union\{R(\fhat) \mid \fhat \in \Fcalhat_\iota\}$. \\[1ex]
If a term $s$ rewrites to $t$ under $R(\Fcalhat_\iota)$ we write $s \toiota t$.
\end{definition}
\begin{proposition}
$R(\Fcalhat_\iota)$ is a canonical rewrite system.
\end{proposition}
\begin{proof}
That $R(\Fcalhat_\iota)$ is strongly normalizing and locally confluent can be shown in the same way as for $R(\Fcalhat_\omega)$.
\end{proof}
To evaluate a term $t \in T^\iota$ under $\sigma \in \Apar$ to a numeral we have to combine $\toomega$ and $\toiota$.

\begin{definition}[evaluation of $T^\iota$]\label{def.evalTiota}
Let $\sigma \in \Apar$ and $t \in T^\iota$. We define $\sigma(t)\Eval_\iota$: 
\begin{itemize}
\item if $c$ is a constant of type $\iota$ then $\sigma(c)\Eval_\iota = c$.
\item If $X(\overrightarrow{t}) \in T^\iota_V$ then $\sigma(X(t))\Eval_\iota = X(\sigma(t)\Eval_\omega)$.
\item if $f \in \Fcal$, $f\colon \iota^\alpha \to \iota$, $\overrightarrow{s}_\alpha \in (T^\iota)^\alpha$ then 
$\sigma(f(s_1,\ldots,s_\alpha))\Eval_\iota  = f(\sigma(\overrightarrow{s}_\alpha)\Eval_\iota).$
\item if $\fhat \in \Fcalhat$, $\fhat\colon \tau(\gamma(1),\ldots,\gamma(\alpha)) \times \omega^{\beta+1} \to \iota$, s.t.  $\alpha\geq 0$, $1\leq i\leq \alpha$, and $\gamma_i >0$,  $\overrightarrow{t}_{\beta+1} \in T^\omega$ then 
$$\sigma(\fhat(\overrightarrow{X}_{\alpha},\overrightarrow{t}_{\beta+1}))\Eval_\iota = \fhat(\overrightarrow{X}_{\alpha}, \sigma(\overrightarrow{t}_{\beta+1})\Eval_\omega)\Eval_\iota.$$
\end{itemize}
\end{definition}
Under a parameter assignment every term in $T^\iota$ evaluates to a first-order term:
\begin{proposition}
Let $t \in T^\iota$ and $\sigma \in \Apar$ then $\sigma(t)\Eval_\iota \in T^\iota_0$.
\end{proposition}
\begin{proof} 
We proceed according to Definition~\ref{def.evalTiota}.
\begin{itemize}
\item If $c$ is a constant in $T^\iota$  then $c \in T^\iota_0$ and $\sigma(c)\Eval_\iota = c$.
\item If $t=X(s_1,\ldots,s_\alpha)$ then $\sigma((s_1,\ldots,s_\alpha))\Eval_\omega = (\bar{\gamma_1},\ldots,\bar{\gamma_\alpha})$ for $\gamma_i \in \Nat$. Therefore $\sigma(X(s_1,\ldots,s_\alpha))\Eval_\iota = X(\bar{\gamma_1},\ldots,\bar{\gamma_\alpha}) \in V^\iota \IN T^\iota_0$.
\item If $t = f(s_1,\ldots,s_\alpha)$ for $f \in \Fcal, f\colon \iota^\alpha \to \iota$ then 
$$\sigma (f(s_1,\ldots,s_\alpha))\Eval_\iota = f(\sigma(s_1)\Eval_\iota,\ldots,\sigma(s_\alpha)\Eval_\iota)$$
By induction we may assume that $s'_i\colon \sigma(s_i)\Eval_\iota \in T^\iota_0$ for $i \in \{1,\ldots,\alpha\}$  (the base cases are shown above). Thus $f(s'_1,\ldots,s'_\alpha) \in T^\iota_0$. 
\item Now consider $\fhat(\orra{X}_{\alpha},t_1,\ldots,t_\beta,t_{\beta+1})$ for $\fhat \in \Fcalhat$ and $\fhat$ is minimal in $<_{\Fcalhat}$.  we distinguish two cases 
	\begin{enumerate}
		\item $\sigma(t_{\beta+1})\Eval_\omega = \bar{0}$.
		Then, by definition, 
		$$\sigma(\fhat(\orra{X}_{\alpha},t_1,	
               \ldots,t_\beta,t_{\beta+1}))\Eval_\iota)) = \fhat(\orra{X}_{\alpha},\bar{\gamma}_1,\ldots,
               \bar{\gamma}_\beta,\bar{0})$$
          for $\gamma_i \in \Nat$. According to Definition~\ref{def.rewrite-iota} 
        $\fhat(\orra{X}_{\alpha}, 
        \bar{\gamma}_1,\ldots,\bar{\gamma}_\beta,\bar{0})$ rewrites to $t'_B$  
       where $t'_B$ is a term of type $\iota$ containing neither defined symbols nor  
       parameters; so $t'_B \in T^\iota_0$.
      \item $\sigma(t_{\beta+1})\Eval_\omega = \bar{p}$ for $p>0$. Then 
      $\fhat(\orra{X}_{\alpha},
         \bar{\gamma}_1,\ldots,\bar{\gamma}_\beta,\bar{p})$ rewrites to a term 
    $$t_S\{Y \ass \fhat(\orra{X}_{\alpha},
         \bar{\gamma}_1,\ldots,\bar{\gamma}_\beta,\bar{p-1})$$
    where $t_S \in T^\iota_0$. By induction on $p$ we infer that \\
   $\fhat(\orra{X}_{\alpha},
         \bar{\gamma}_1,\ldots,\bar{\gamma}_\beta,\bar{p-1})\Eval_\iota \in 
         T^\iota_0$ and so \\ 
         $\fhat(\orra{X}_{\alpha}, \bar{\gamma}_1,\ldots 
            \bar{\gamma}_\beta,\bar{p}) \in T^\iota_0$.
        \end{enumerate}
\item It $t= \fhat(\orra{X}_{\alpha},t_1,\ldots,t_\beta,t_{\beta+1})$ for $\fhat \in \Fcalhat$ and $\fhat$ is not minimal in $<_{\Fcalhat}$ we have to add induction on $<_{\Fcalhat}$ with the base cases shown above.
\end{itemize} 
\end{proof}
\begin{example}\label{ex.evalTiota}
\label{ex.Tiota}
Let $X \in V^G, X\colon \omega^2 \to \iota$, $Y\colon \iota$, $g \in \Fcal, g\colon \iota \times \iota \to \iota$ and $n,m$ parameters. Assume $\fhat$ is defined as 
\begin{eqnarray*}
\fhat(X,Y,n,\bar{0}) &=& Y,\\
\fhat(X,Y,n,s(m)) &=& g(X(n,m),\fhat(X,Y,n,m)).
\end{eqnarray*}
We evaluate $\fhat(X,Y,n,m)$ under $\sigma$, where $\sigma(n) = \bar{1}, \sigma(m) = \bar{2}$ and $\sigma(k) = \bar{0}$ for $k \not \in \{n,m\}$.
\[
\begin{array}{l}
\sigma(\fhat(X,Y,n,m))\Eval_\iota = \fhat(X,Y,\bar{1},\bar{2})\Eval_\iota = 
g(X(\bar{1},\bar{1}),\fhat(X,Y,\bar{1},\bar{1})\Eval_\iota) =\\
g(X(\bar{1},\bar{1}),g(X(\bar{1},\bar{0}),\fhat(X,Y,\bar{1},\bar{0})\Eval_\iota) = 
g(X(\bar{1},\bar{1}),g(X(\bar{1},\bar{0}),Y)).
\end{array}
\]
When we write $x_1$ for $X(\bar{1},\bar{1})$ and $x_2$ for $X(\bar{1},\bar{0})$ and $y$ for $Y$ we get the term in the common form $g(x_1,g(x_2,y))$.
\end{example}
Substitutions on term schemata need to be schematic as well, particularly when we are interested in unification. We develop some formal tools below to describe such schemata.
\begin{definition}\label{def.essdis}
Let $\orra{s}_1 = (s_1,\ldots,s_\alpha)$ for $s_1,\ldots,s_\alpha \in T^\omega_0$ and  $\orra{s}_2 = (s'_1,\ldots,s'_\beta)$ for $s'_1,\ldots,s'_\beta \in T^\omega_0$. $\orra{s}_1$ and $\orra{s}_2$ are called {\em essentially distinct} if either $\alpha \neq \beta$ or  for all $\sigma \in \Apar$ $\sigma(\orra{s}_1) \neq \sigma(\orra{s}_2)$.
\end{definition}
\begin{remark}
Note that, for $\alpha = \beta$ in Definition~\ref{def.essdis},  $\sigma(\orra{s}_1) \neq \sigma(\orra{s}_2)$ iff there exists an $i \in \{1,\ldots,\alpha\}$ such that $\sigma(s_i) \neq \sigma(s'_i)$. 
\end{remark}
\begin{example}
$n$ and $s(n)$ are essentially distinct and so are $\bar{0}$ and $s(n)$; $m$ and $s(n)$ are not essentially distinct (just use $\sigma$ with $\sigma(m) = \bar{1}$ and $\sigma(n)= \bar{0}$).
\end{example}

\begin{definition}[s-substitution]
Let $\Theta$ be a finite set of pairs $(X(\orra{s}_{\alpha}),t)$ where $X(\orra{s}_{\alpha}) \in T^\iota_V$, $\orra{s}_\alpha$ a tuple of terms in $T^\omega_0$ and $t \in T^\iota$. Note that the global variables occurring in $\Theta$ need not be of the same type. $\Theta$ is called an {\em s-substitution} if for all $(X(\overrightarrow{s}_{\alpha}),t)$ , $(Y(\overrightarrow{s'}_{\alpha}),t') \in \Theta$ either $X \neq Y$ or the tuples $\orra{s}_{\alpha}$ and $\orra{s'}_{\alpha}$ are essentially distinct. For $\sigma \in \Apar$ we define $\Theta[\sigma] =\{X(\sigma(\overrightarrow{s}_{\alpha})) \ass t\sigma\Eval_\iota \mid (X(\overrightarrow{s}_{\alpha}),t) \in \Theta\}.$\\[1ex]
We define $\dom(\Theta) = \{X(\overrightarrow{s}_{\alpha}) \mid (X(\overrightarrow{s}_{\alpha}),t) \in \Theta\}$ and $\rg(\Theta) = \{t \mid (X(\overrightarrow{s}_{\alpha}),t) \in \Theta\}$.
\end{definition}

\begin{remark}
Note that for pairs $(X(\orra{s}_{\alpha}),t)$ occurring in substitutions we require that the components of $\orra{s}_\alpha$ are terms in $T^\omega_0$. For these term tuples we have $\sigma(\orra{s}) = \sigma(\orra{s})\Eval_\omega$, so there is no need to apply $\Eval_\omega$ after substitution.
\end{remark} 
\begin{proposition}
For all $\sigma \in \Apar$ and every s-substitution $\Theta$, $\Theta[\sigma]$ is a (first-order) substitution.
\end{proposition}
\begin{proof}
It is enough to show that for all $(X(\overrightarrow{s}_{\alpha}),t),(Y(\overrightarrow{s'}_{\alpha'}),t') \in \Theta$ $X(\sigma(\overrightarrow{s}_{\alpha}))$ $\neq Y(\sigma(\overrightarrow{s'}_{\alpha'}))$ for all $\sigma \in \Apar$. If $X \neq Y$ this is obvious; if $X =Y$ then, by definition of $\Theta$, $\overrightarrow{s}_{\alpha}$ and $\overrightarrow{s'}_{\alpha}$ are essentially distinct and so for each $\sigma \in \Apar$ we have $X(\sigma(\overrightarrow{s}_{\alpha})) \neq X(\sigma(\overrightarrow{s'}_{\alpha'}))$. Thus  $\Theta[\sigma]$ is indeed a substitution as for $X(\overrightarrow{s}_{\alpha}) \in T^\iota_V$ $X(\sigma(\overrightarrow{s}_{\alpha})) \in V^\iota$.
\end{proof}
\begin{example}\label{ex.ssubstitution}
\label{ssub}
The following is an s-substitution
$$\Theta = \{(X(n,m), \hat{S}(Y(m),n)), (X(s(n),m), \hat{S}(Y(m),s(n))), (X(0,0), Y(0))\}.$$
\end{example}
The application of an s-substitution $\Theta$ to terms in $T^\iota$ is defined inductively on the complexity of term definitions as usual.
\begin{definition}\label{def.appssub}
Let $\Theta$ be an s-substitution and $\sigma$ a parameter assignment. We define $t\Theta[\sigma]$ for terms $t \in T^\iota$:
\begin{itemize}
\item if $t$ is a constants of type $\iota$, then $t\Theta[\sigma] = t$,
\item if $t= X(\overrightarrow{s}_{\alpha})$ and $ (X(\overrightarrow{s'}_{\alpha}),t')  \in \Theta$ such that $X(\sigma(\orra{s'}_{\alpha})) =  X(\sigma(\orra{s}_{\alpha}))$, then $X(\orra{s}_\alpha)\Theta[\sigma] = \sigma(t')\Eval_\iota$, otherwise  $X(\overrightarrow{s}_{\alpha})\Theta[\sigma] = X(\sigma(\orra{s}_\alpha))$;
\item if $f \in \Fcal$, $f\colon \iota^\alpha  \to \iota$, $s_1,\ldots,s_\alpha \in T^\iota$ then $f(s_1,\ldots,s_\alpha)\Theta[\sigma] = f(s_1\Theta[\sigma],\ldots,$ $s_\alpha\Theta[\sigma])$, 
\item if $\fhat \in \Fcalhat$, $\fhat\colon \tau(\gamma(1),\ldots,\gamma(\alpha_1)) \times \omega^{\beta+1} \to \iota$, then 
$$\fhat(\orra{X}_{\alpha_1}, t_1,\ldots,t_{\beta+1})\Theta[\sigma] = \fhat(\orra{X}_{\alpha_1},\sigma(t_1)\Eval_\omega,\ldots,\sigma(t_{\beta+1})\Eval_\omega)\Eval_\omega \Theta[\sigma].$$
\end{itemize}
\end{definition}
\begin{example}
Let us consider the following defined function symbol:
\begin{eqnarray*}
\ghat(X,n,\bar{0}) &=& X(n,0),\\
\ghat(X,n,s(m)) &=& g(X(n,m),\ghat(X,s(n),m)).
\end{eqnarray*}
and the parameter assignment $\sigma = \left\lbrace n\leftarrow 0 , m\leftarrow s(0) \right\rbrace$. Then the evaluation of the term $\ghat(X,n,m)$ by the s-substitution $\Theta$ from Example~\ref{ssub} proceeds as follows:
\[\begin{array}{lcl}
\ghat(X,n,m)\Theta[\sigma] &=&  \ghat(X,\sigma(n)\Eval_{\omega},\sigma(m)\Eval_{\omega})\Eval_{\omega}\Theta[\sigma] \\
&=& \ghat(X,0,s(0))\Eval_{\omega}\Theta[\sigma] \\
&=& g(X(0,s(0)),\ghat(X,s(0),0)\Eval_{\omega})\Theta[\sigma]\\
&=& g(X(0,s(0)),X(s(0),0))\Theta[\sigma] = g(f(Y(s(0))),s(Y(0)))
\end{array}
\]
Where $\Theta[\sigma]$ is 
$$ \{(X(0,s(0)), Y(s(0))), (X(s(0),s(0)), s(Y(s(0)))), (X(0,0), Y(0))\}.$$

\end{example}
The composition of $s$-substitutions is not trivial as, in general, there is no uniform representation of composition under varying parameter assignments.
\begin{example}\label{ex.ssubstitution2}
Let 
$$\Theta_1 = \{(X_1(n),f(X_1(n))\}\ \Theta_2 = \{(X_1(0),g(a))\}.$$
Then, for $\sigma \in \Apar$ s.t. $\sigma(n)=0$ we get 
$$\Theta_1[\sigma] \circ \Theta_2[\sigma] = \{X_1(0) \ass f(X_1(0))\} \circ \{X_1(0) \ass g(a)\} =  \{X_1(0) \ass f(g(a))\}.$$
On the other hand, for $\sigma' \in \Apar$ with $\sigma'(n)=1$ we obtain
$$\Theta_1[\sigma'] \circ \Theta_2[\sigma'] = \{X_1(1) \ass f(X_1(1))\} \circ \{X_1(0) \ass g(a)\} = \{X_1(1) \ass f(X_1(1)),$$ $$X_1(0) \ass g(a)\}.$$
Or take $\Theta'_1 = \{(X_1(n),X_2(n))\}$ and $\Theta'_2 = \{(X_2(m),X_1(m))\}$.\\[1ex]
Let $\sigma(n)=\sigma(m)=0$ and $\sigma'(n)=0,\sigma'(m)=1$. Then 
\begin{eqnarray*}
\Theta'_1[\sigma] \circ \Theta'_2[\sigma] &=& \{X_2(0) \ass X_1(0)\},\\
\Theta'_1[\sigma'] \circ \Theta'_2[\sigma'] &=& \{X_1(0) \ass X_2(0), X_2(1) \ass X_1(1)\}.
\end{eqnarray*}
\end{example} 
The examples above suggest the following restrictions on s-substitutions with respect to composition. The first definition ensures that domain and range are variable-disjoint.
\begin{definition}
Let $\Theta$ be an s-substitution. $\Theta$ is called {\em normal} if for all $\sigma \in \Apar$ $\dom(\Theta[\sigma]) \intrs V^\iota(\rg(\Theta[\sigma])) = \emptyset$.
\end{definition}

\begin{example}
The s-substitution in Example~\ref{ex.ssubstitution} is normal.  The substitutions $\Theta'_1$ and $\Theta'_2$ in Example~\ref{ex.ssubstitution2} are normal. $\Theta_1$ in Example~\ref{ex.ssubstitution2} is not normal.
\end{example}
\begin{proposition}\label{prop.normalssubdec}
It is decidable whether a given s-substitution is normal.
\end{proposition}
\begin{proof}
Let $\Theta$ be an s-substitution. We search for equal global variables in $\dom(\Theta)$ and in $\rg(\Theta)$; if there are none then $\Theta$ is trivially normal. So let $X \in V^G(\dom(\Theta))$ $\intrs V^G(\rg(\Theta))$. For every $X(\overrightarrow{s}_{\alpha}) \in \dom(\Theta)$ and for every $X(\overrightarrow{t}_{\alpha})$ occurring in $\rg(\Theta)$ we test whether there exists a $\sigma \in \Apar$ such that $\sigma(\orra{s}_\alpha) = \sigma(\orra{t}_\alpha)$. This test uses ordinary first-order unification on terms in $T^\omega_0$. When we find $X(\overrightarrow{s}),X(\overrightarrow{t})$ such that there exists a $\sigma \in \Apar$ for which $\sigma(\orra{s})$ is equal to $\sigma(\orra{t})$ then $\Theta$ is not normal, and normal otherwise.
\end{proof}
Example~\ref{ex.ssubstitution2} shows also that normal s-substitutions cannot always be composed to an s-substitution;  thus we need an additional condition.
\begin{definition}
Let $\Theta_1,\Theta_2$ be normal s-substitutions. $(\Theta_1,\Theta_2)$ is called {\em composable} if for all $\sigma \in \Apar$ 
\begin{enumerate}
\item $\dom(\Theta_1[\sigma]) \intrs \dom(\Theta_2(\sigma)) = \emptyset$,
\item $\dom(\Theta_1[\sigma]) \intrs V^\iota(\rg(\Theta_2[\sigma])) = \emptyset$.
\end{enumerate} 
\end{definition}
\begin{proposition}
It is decidable whether, for two normal s-substitutions $\Theta_1,\Theta_2$, $(\Theta_1,\Theta_2)$ is composable.
\end{proposition}
\begin{proof}
Like in Proposition~\ref{prop.normalssubdec}: by unfication tests on $X(\overrightarrow{s}),X(\overrightarrow{t})$ occurring in the sets under consideration.
\end{proof} 

\begin{definition}
Let $\Theta_1,\Theta_2$ be normal s-substitutions and $(\Theta_1,\Theta_2)$ composable. Assume that 
\begin{eqnarray*}
\Theta_1 &=& \{ (X_1(\orra{s_1}),t_1),\ldots,(X_\alpha(\orra{s_\alpha}),t_\alpha)\},\\
\Theta_2 &=& \{(Y_1(\overrightarrow{w_1}),r_1),\ldots,(Y_\beta(\overrightarrow{w_\beta}),r_\beta)\}.
\end{eqnarray*}
Then the composition $\Theta_1 \star \Theta_2$ is defined as 
$$\{ (X_1(\overrightarrow{s_1}),t_1\Theta_2),\ldots,(X_\alpha(\overrightarrow{s_\alpha}),t_\alpha\Theta_2),(Y_1(\overrightarrow{w_1}),r_1),\ldots,(Y_\beta(\overrightarrow{w_\beta}),r_\beta)\}.$$
\end{definition}
The following proposition shows that $\Theta_1 \star \Theta_2$ really represents composition.
\begin{proposition}\label{prop.ssubcomposition}
Let $\Theta_1,\Theta_2$ be normal s-substitutions and $(\Theta_1,\Theta_2)$ be composable then for all $\sigma \in \Apar$ $(\Theta_1 \star \Theta_2)[\sigma] = \Theta_1[\sigma] \circ \Theta_2[\sigma]$.
\end{proposition}
\begin{proof}
Let 
\begin{eqnarray*}
\Theta_1 &=& \{ (X_1(\orra{s_1}),t_1),\ldots,(X_\alpha(\orra{s_{\alpha}}),t_\alpha)\},\\
\Theta_2 &=& \{(Y_1(\orra{w_1}),r_1), \ldots,(Y_{\beta}(\orra{w_\beta}),r_\beta)\}.
\end{eqnarray*}
Then $\Theta_1 \star \Theta_2$ is defined as 
$$\{ (X_1(\overrightarrow{s_1}),t_1\Theta_2),\ldots,(X_\alpha(\overrightarrow{s_\alpha}),t_\alpha\Theta_2),(Y_1(\overrightarrow{w_1}),r_1),\ldots,(Y_\beta(\overrightarrow{w_\beta}),r_\beta)\}.$$
We write $x_i$ for $X_i(\sigma(\orra{s_i})))$ and $y_j$ for $Y_j(\sigma(\orra{w_j}))$, $\theta_1$ for $\Theta_1[\sigma]$ and $\theta_2$ for $\Theta_2[\sigma]$. Moreover let $t'_i = \sigma(t_i)\Eval_\iota, r'_j = \sigma(r_j)\Eval_\iota$.
Then 
\begin{eqnarray*}
\theta_1 &=& \{ x_1 \ass t'_1,\ldots,x_\alpha  \ass t'_\alpha\},\\
\theta_2 &=& \{(y_1 \ass r'_1,\ldots,y_{\alpha'} \ass r'_\beta\}.
\end{eqnarray*}
As $(\Theta_1,\Theta_2)$ is composable we have
\begin{enumerate}
\item $\{x_1,\ldots,x_\alpha\} \intrs \{y_1,\ldots,y_\beta\} = \emptyset$, and 
\item $\{x_1,\ldots,x_\alpha\} \intrs V^\iota(\{r'_1,\ldots,r'_\beta\}) = \emptyset$.
\end{enumerate} 
So 
\[
\begin{array}{l}
\theta_1\theta_2 =\{ x_1 \ass t'_1,\ldots,x_\alpha  \ass t'_\alpha)\}\theta_2= 
\{x_1 \ass t'_1\theta_2,\ldots,x_\alpha \ass t'_\alpha\theta_2\}  \union \theta_2.
\end{array}
\]
The last substitution is just $(\Theta_1 \star \Theta_2)[\sigma]$.
\end{proof}

\begin{proposition}
Let $\Theta_1,\Theta_2$ be normal s-substitutions and $(\Theta_1,\Theta_2)$ composable. Then $\Theta_1 \star \Theta_2$ is normal.
\end{proposition}
\begin{proof}
Like in the proof of Proposition~\ref{prop.ssubcomposition} let $\Theta_1[\sigma] =\theta_1,\Theta_2[\sigma] = \theta_2$. We have to show that $\dom(\theta_1\theta_2) \intrs V^\iota(\rg(\theta_1\theta_2) = \emptyset$. We have 
$$\theta_1\theta_2 = \{x_1 \ass t'_1\theta_2,\ldots,x_\alpha \ass t'_\alpha\theta_2\}  \union \theta_2.$$
As $\theta_1$ is normal we have $V^\iota(t'_i) \intrs \{x_1,\ldots,x_\alpha\} =\emptyset$ for $i = 1,\ldots,\alpha$. By definition of composability $\rg(\theta_2) \intrs \{x_1,\ldots,x_\alpha\} = \emptyset$, and therefore 
$$V^\iota(\{t'_1\theta_2,\ldots,t'_\alpha\theta_2\}) \intrs \{x_1,\ldots,x_\alpha\} = \emptyset.$$
So  $\{x_1 \ass t'_1\theta_2,\ldots,x_\alpha \ass t'_\alpha\theta_2\}$ is normal. As also $\Theta_2$ is normal we have $\dom(\theta_2) \intrs V^\iota(\rg(\theta_2) = \emptyset$. Hence we obtain $\dom(\theta_1\theta_2) \intrs V^\iota(\rg(\theta_1\theta_2)) = \emptyset.$ 
\end{proof}

\begin{definition}[s-unifier]\label{def.sunifier}
Let $t_1,t_2 \in T^\iota$. An s-substitution $\Theta$ is called an s-unifier of $t_1,t_2$ if for all $\sigma \in \Apar$ 
$(t_1\sigma\Eval_\iota)\Theta[\sigma] = (t_2\sigma\Eval_\iota)\Theta[\sigma]$. We refer to $t_1,t_2$ as s-unifiable if there exists an s-unifier of $t_1,t_2$. s-unifiability can be extended to more than two terms and to formula schemata (to be defined below) in an obvious way. 
\end{definition}
Notice that the  s-substitution of Example~\ref{ssub} is an s-unifier of $X_{3}(0,m)$ and $X_{4}(0,m)$.
\begin{definition}\label{def.sunifnormal}
An s-unifier $\Theta$ of $t_1,t_2$ is called {\em restricted} to $\{t_1,t_2\}$ if $T^\iota_V(\Theta) \IN T^\iota_V(\{t_1,t_2\})$.
\end{definition} 
\begin{remark}
It is easy to see that for any s-sunifier $\Theta$ of $\{t_1,t_2\}$ there exists an s-unifier $\Theta'$ of $\{t_1,t_2\}$ which is {\em restricted} to $\{t_1,t_2\}$.
\end{remark} 
Formula schemata are defined in a way that also the number of variables in formulas can increase with the assignments of parameters. For this reason we use global variables in the definition.

\begin{definition}[formula schemata ($\FS$)]\label{def.formschem}
We define the set $\FS$  inductively:
\begin{itemize}
\item Let $\xi$ be a formula variable in $V^F$ then $\xi \in \FS$.
\item Let $P\colon \iota^\alpha \to o \in \Pcal$ and $t_1,\ldots, t_{\alpha} \in T^\iota$. Then $P(t_1,\ldots,t_\alpha) \in \FS$
\item Let $\Phat \in \Pcal^\tau$ for $\tau\colon \tau(\gamma(1),\ldots,\gamma(\alpha))  \times \omega^{\beta+1} \to o$, for $\alpha \geq 0$, $0< i\leq \alpha$, and $\gamma_i> 0$, $\overrightarrow{X}_{\alpha}$ a tuple of variables in $V^G$  and  $\overrightarrow{t}_{\beta+1}$ a tuple of terms in $T^\omega$. Then $\Phat(\overrightarrow{X}_{\alpha},\overrightarrow{t}_{\beta+1}) \in \FS$.
\item Let $F \in \FS$ then $\neg F \in \FS$.
\item  If $F_1,F_2 \in \FS$ then $F_1 \land F_2 \in \FS$ and $F_1 \lor F_2 \in \FS$.
\end{itemize}
The subset of $\FS$ not containing formula variables is denoted by $\FS_0$. The subset of $\FS$ containing no defined symbols at all and neither parameters nor numerals is denoted by $\Fnull$. $\Fnull$ is a set of ordinary quantifier-free first-order formulas.
\end{definition}

\begin{definition}[defining equations for predicate symbols]\label{def.defeqPS}
For every $\Phat \in \Pcalhat^\tau$ for  $\tau\colon (\omega^{\gamma_1} \to  \iota) \times \ldots \times (\omega^{\gamma_\alpha} \to  \iota) \times \omega^{\beta+1} \to o$, for $\alpha\geq 0$ and $\gamma_i > 0$ for $i \in \{1, \ldots, \alpha\}$, we define a set $D(\Phat)$ of defining equations, where  $\overrightarrow{X}_{\alpha} \in (V^G)^\alpha$, and  $\overrightarrow{t}_{\beta+1} \in (T^\omega)^{\beta+1}$. $D(\Phat)$ consists of the equations
$$ \Phat(\overrightarrow{X}_{\alpha},\overrightarrow{t}_{\beta},\bar{0}) = F_B,\quad \Phat(\overrightarrow{X}_{\alpha},\overrightarrow{t}_{\beta},s(m)) = F_S\{\xi \ass \Phat(\overrightarrow{X}_{\alpha},\overrightarrow{t}_{\beta},m)\},$$ 
where, for a $\orderp$-minimal $\Phat$ $F_B,F_S \in \Fnull$. If $\Phat$ is not $\orderp$-minimal then $F_B,F_S \in \FS$ such that for every 
$\Qhat \in \Pcalhat$ occurring in $F_B,F_S$ we have $\Qhat \orderp \Phat$.  The only global variable and parameters occurring in $F_B$ are in  $\overrightarrow{X}_{\alpha}$, and those occurring in the terms of$\overrightarrow{t}_{\beta}$ respectively.  The only global variables occurring in $F_S$ are $\overrightarrow{X}_{\alpha}$, and besides the parameters occurring in $\overrightarrow{t}_{\beta}$, the parameter $m$ and the formula variable $\xi$ may occur in $F_S$.
Like for $\Fcalhat_\omega$ and $\Fcalhat_\iota$ we define$$D(\Pcalhat) = \Union\{D(\Phat) \mid \Phat \in \Pcalhat\}.$$
\end{definition}
 The evaluation of a formula $F \in \FS$ is denoted by $\Eval_o$ and is defined inductively.
\begin{definition}\label{def.EvalFS}
Let $\sigma \in \Apar$; we define $\sigma(F)\Eval_o$ for $F \in \FS$. 
\begin{itemize}
\item[(1)] Let $\xi$ be a formula variable in $V^F$ then $\sigma(\xi)\Eval_o = \xi$.
\item[(2)] Let $P\colon \iota^\alpha \to o \in \Pcal$ and $t_1,\ldots, t_{\alpha} \in T^\iota$. Then 
$$\sigma(P(t_1,\ldots,t_\alpha))\Eval_o = P(\sigma(t_1)\Eval_\iota,\ldots,\sigma(t_{\alpha})\Eval_\iota).$$
\item[(3)] Let $\Phat \in \Pcal^\tau$ and $F=\Phat(X_1,\ldots,X_\alpha,t_1,\ldots,t_{\beta+1})$. Let $D(\Phat)=$
\begin{eqnarray*}
\Phat(\overrightarrow{X}_{\alpha},\overrightarrow{t}_{\beta},0) &=& F_B, \\
\Phat(\overrightarrow{X}_{\alpha},\overrightarrow{t}_{\beta},m+1) &=& F_S\{\xi \ass \Phat(\overrightarrow{X}_{\alpha},\overrightarrow{t}_{\beta},m)\}.
\end{eqnarray*}
we distinguish two cases:\\
(a) $\sigma(t_{\beta+1})\Eval_\iota = \bar{0}$. Then 
$$\sigma(\Phat(\overrightarrow{X}_{\alpha},\overrightarrow{t}_{\beta},0))\Eval_o = \sigma(F_B)\Eval_o$$
(b) $\sigma(t_{\beta+1})\Eval_\iota = s(\bar{p})$. Then 
$$\sigma(\Phat(\overrightarrow{X}_{\alpha},\overrightarrow{t}_{\beta},s(\bar{p})))\Eval_o = \sigma(F'_S)\Eval_o$$
For 
\begin{eqnarray*}
F'_S &=& F_S\{\xi \ass \Phat(\overrightarrow{X}_{\alpha},\overrightarrow{t}_{\beta},\bar{p})\}\\
\end{eqnarray*}

\item[(4)] $\sigma(\neg F)\Eval_o = \neg \sigma(F)\Eval_o$.
\item[(5)] $\sigma(F_1 \circ F_2)\Eval_o =  \sigma(F_1)\Eval_o \circ \sigma(F_2)\Eval_o$ for $\circ \in \{\land,\lor\}$.
\end{itemize}
\end{definition}
\begin{proposition}
Let $F \in \FS_0$ and $\sigma \in \Apar$. Then $\sigma(F)\Eval_o \in \Fnull$.
\end{proposition}
\begin{proof}
If there are no defined predicate symbols in $F$ then, obviously, $\sigma(F)\Eval_o \in \Fnull$;  indeed, here only the cases (1),(2),(4) and (5) in Definition~\ref{def.EvalFS} apply.

If there are defined predicate symbols we proceed by induction on $\orderp$ and the induction parameter.

Let $\Phat$ be minimal in $\orderp$ and let $F=\Phat(\overrightarrow{X}_{\alpha},\overrightarrow{t}_{\beta+1})$. We show that $\sigma(F)\Eval_o \in F_0$:
\begin{itemize}
\item[(a)] $\sigma(t_{\beta+1})\Eval_\iota = \bar{0}$. Then, by Definition~\ref{def.EvalFS} 
$$\sigma(\Phat(\overrightarrow{X}_{\alpha},\overrightarrow{t}_{\beta+1}))\Eval_o = \sigma(F_B)\Eval_o$$
As $\Phat$ is minimal the formula $F_B$ does not contain defined predicate symbols and so $\sigma(F_B)\Eval_o \in \Fnull$.
\item[(b)] $\sigma(t_{\beta+1})\Eval_\iota = s(\bar{p})$. Here we have 
$$\sigma(\Phat(\overrightarrow{X}_{\alpha},\overrightarrow{t}_{\beta+1}))\Eval_o = \sigma(F'_S)\Eval_o$$
for 
\begin{eqnarray*}
F'_S &=& F_S\{\xi \ass \Phat(\overrightarrow{X}_{\alpha},\overrightarrow{t}_{\beta}, \overline{p})\}.
\end{eqnarray*}
Note that $F_S$ itself does not contain defined predicate symbols; in $F'_S$ we have the symbol $\Phat$ but with $\Phat(\overrightarrow{X}_{\alpha},\overrightarrow{t}_{\beta}, \overline{p})$. Therefore  we proceed by induction on the value of $\sigma(t_{\beta+1})$ and infer that also 
$\sigma(F'_S)\Eval_o \in \Fnull$. 
\end{itemize}
If $\Phat$ is not minimal the base case for $\Phat$ involves only smaller defined predicate symbols. So by induction on $\orderp$ we get the desired result.
\end{proof}
\begin{definition}[unsatisfiable schemata]\label{def.schema-unsat}
Let $F \in \FS$. Then $F$ is called {\em unsatisfiable} if for all $\sigma \in \Apar$ the formula $\sigma(F)\Eval_o$ is unsatisfiable.
\end{definition}
\begin{example}\label{ex.2parameters}
Let $a$ be a constant symbol of type $\iota$, $P \in \Pcal^{\iota \times \iota \to o}$, $\fhat$ as in Example~\ref{ex.Tiota}, $\Phat \in \Pcalhat^\tau$ for $\tau = (\omega\to \iota) \times \omega \to o$, and $\Qhat \in \Pcalhat^{\tau'}$ for $\tau' = (\omega \to \iota)^2 \times \omega^2 \to o$. Concerning the ordering we have $\Phat \orderp \Qhat$. The defining equations for $\Phat$ and $\Qhat$ are:
\begin{eqnarray*}
\Phat(X,\bar{0}) &=& \neg P(X(\bar{0}),\fhat(a,0)),\\
\Phat(X,s(n)) &=& \Phat(X,n) \lor \neg P(X(s(n)),\fhat(a,s(n))).
\end{eqnarray*}
\begin{eqnarray*}
\Qhat(X,Y,n,\bar{0}) &=& P(\fhat(Y(\bar{0}),\bar{0}),Y(\bar{1})) \land \Phat(X,n),\\
\Qhat(X,Y,n,s(m)) &=& P(\fhat(Y(\bar{0}),s(m)),Y(\bar{1})) \land \Phat(X,n).
\end{eqnarray*}
It is easy to see that the schema $\Qhat(X,Y,n,m)$ is unsatisfiable. We compute $\sigma(\Qhat(X,Y,n,m))\Eval_o$ for $\sigma$ with $\sigma(m) = \bar{2}, \sigma(n) = \bar{3}$: 
\[
\begin{array}{l}
\sigma(\Qhat(X,Y,n,m))\Eval_o = P(\fhat(Y(0),2),Y(1)) \land \Phat(X,3)\Eval_o)=\\
P(\fhat(Y(0),2),Y(1)) \land (\Phat(X,2)\Eval_o \lor \neg P(X(3),\fhat(a,3)\Eval_\iota) = \\
P(\fhat(Y(0),2),Y(1)) \land (\Phat(X,1)\Eval_o  \lor \neg P(X(2),\fhat(a,2)\Eval_\iota) \lor \neg P(X(3),\fhat(a,3)\Eval_\iota)) =\\
.... P(g(g(Y(0)),Y(1)) \land ( \neg P(X(0),a) \lor \neg P(X(1),g(a)) \lor\\ 
\neg P(X(2),g(g(a))) \lor \neg P(X(3),g(g(g(a))))).
\end{array}
\]
Note that, for $\sigma(n) = \bar{\alpha}$ the number of different variables in 
$\sigma(\Qhat(X,Y,n,m))\Eval_o$ is $\alpha+2$; so the number of variables increases with the parameter assignments.
\end{example}

\begin{example} \label{1smaform}
Let us now consider the schematic formula representation of the inductive definition extracted from the $m$-successor eventually constant schema presented in Section~\ref{motivation}. This requires us to define four defined predicate symbols ${\hat{F}_1},{\hat{F}_2},{\hat{F}_3},{\hat{F}_4}$, and  ${\hat{F}_5}$ are of type $(\omega \times \omega \to \omega)^2 \times (\omega \to \omega) \times \omega^2  \to o$  where ${\hat{F}_5} \orderp {\hat{F}_4} \orderp {\hat{F}_3} \orderp {\hat{F}_2} \orderp {\hat{F}_1}$. Furthermore, these defined predicate symbols contain  the symbols $=,< \in \omega \times \omega \to o$, $f,\mathit{suc} \in  \omega \to \omega$, $a\colon \omega$, and $n,m\in \Ncal$. We also require a defined function symbol $\hat{S}$ of type $\omega  \times \omega \to \omega$. Note that in this case the type $\iota$ is identical to $\omega$. Using these symbols we can rewrite the inductive definition provided in  Section~\ref{motivation} as follows: 
\begin{align*}
\hat{F_1}(\mathbf{X},n,m) \equiv &  \hat{F_2}(\mathbf{X},n,m)  \wedge\  \hat{F_3}(\mathbf{X},n,m)\\
\hat{F_2}(\mathbf{X},n,0) \equiv&  \eqpred{\hat{S}(X_1(n,0),0)}{n} \ \vee \  \leqpred{X_1(n,0)}{n} \\
\hat{F_2}(\mathbf{X},n,s(m)) \equiv&  ( \eqpred{\hat{S}(X_1(n,s(m)),s(m))}{n} \ \vee \  \leqpred{X_1(n,s(m))}{n} ) \wedge \\ &  \hat{F_2}(\mathbf{X},n,m)\\
\hat{F_3}(\mathbf{X},0,m) \equiv&   \hat{F_5}(\mathbf{X},0,m) \ \wedge\ \nleqpred{a}{0} \\
\hat{F_3}(\mathbf{X},s(n),m) \equiv&  ( \hat{F_5}(\mathbf{X},s(n),m)  \ \wedge  (\hat{F_4}(\mathbf{X},n,m)) \ \wedge \\ & \hat{F_3}(\mathbf{X}, n,m) \\
\hat{F_4}(\mathbf{X},n,0) \equiv&  \nleqpred{\hat{S}(X_2(n,0),0)}{s(n)}  \ \vee \ \eqpred{\hat{S}(X_2(n,0),0)}{n} \ \vee \\ & \leqpred{X_2(n,0)}{n}\\
\hat{F_4}(\mathbf{X},n,s(m)) \equiv&   \nleqpred{\hat{S}(X_2(n,s(m)),s(m))}{s(n)}  \ \vee \\ & \eqpred{\hat{S}(X_2(n,s(m)),s(m))}{n} \vee \leqpred{X_2(n,s(m))}{n} \ \wedge \\ & \hat{F_4}(\mathbf{X},n,m)\\
\hat{F_5}(\mathbf{X},n,0)  \equiv&  \neqpred{\hat{S}(X_3(n),0)}{n}\\ 
\hat{F_5}(\mathbf{X},n,s(m))  \equiv&  \neqpred{\hat{S}(X_3(n),s(m))}{n} \ \vee \ \hat{F_5}(\mathbf{X},n,m)\\
\hat{S}(Z,0)\equiv&  Z\\
\hat{S}(Z,s(n)) \equiv& \mathit{suc}(\hat{S}(Z,n))\\
\end{align*}
\normalsize
\end{example}

\section{The Resolution Calculus}\label{sec:resolution}

The basis of our calculus for refuting formula schemata is a calculus for quantifier-free formulas $\RPLnull$ which combines dynamic normalization rules (a la Andrews, see~\cite{DBLP:journals/jacm/Andrews81}) with the resolution rule. In contrast to~\cite{DBLP:journals/jacm/Andrews81} we do not restrict the resolution rule to atomic formulas. We denote as $\PLnull$ the set of quantifier-free formulas in predicate logic; for simplicity we omit $\impl$ and represent it by $\neg$ and $\lor$ in the usual way.  {\em Sequents} are objects of the form $\Gamma \seq \Delta$ where $\Gamma$ and $\Delta$ are multisets of formulas in $\PLnull$.

\begin{definition}[$\RPLnull$]
The axioms of $\RPLnull$ are sequents $\seq F$ for $F \in \PLnull$.\\[1ex]
The rules are elimination rules for the connectives and the resolution rule.
\[
\infer[\AndrI]{\Gamma \seq \Delta, A}
   {\Gamma \seq \Delta, A \land B} \ \
\infer[\AndrII]{\Gamma \seq \Delta, B}
   {\Gamma \seq \Delta, A \land B} \ \
\infer[\Andl]{A,B,\Gamma \seq \Delta}
  { A \land B,\Gamma \seq \Delta
  }
\]
\[
\infer[\Orr]{\Gamma \seq \Delta,A,B}
 { \Gamma \seq \Delta, A \lor B} \ \
\infer[\OrlI]{A,\Gamma \seq \Delta}
 { A \lor B,\Gamma \seq \Delta} \ \ 
\infer[\OrlII]{B,\Gamma \seq \Delta}
 { A \lor B,\Gamma \seq \Delta} 
\]
\[
\infer[\negr]{A, \Gamma \seq \Delta}
 {\Gamma \seq \Delta, \neg A} \ \  
\infer[\negl]{\Gamma \seq \Delta,A}
 { \neg A, \Gamma \seq \Delta}
\]
The resolution rule where $\vartheta$ is an m.g.u. of $\{A_1,\ldots,A_k,B_1,\ldots,B_l\}$ and\\ 
$V(\{A_1,\ldots,A_k\}) \intrs V(\{B_1,\ldots,B_l\}) = \emptyset$ is 
\[
\infer[\res]{\Gamma\vartheta,\Pi\vartheta \seq \Delta\vartheta,\Lambda\vartheta} 
{\Gamma \seq \Delta,A_1,\ldots,A_k 
  &
  B_1,\ldots,B_m,\Pi \seq \Lambda
}
\]
\end{definition}

\begin{proposition}\label{prop.RPLnull-complete}
$\RPLnull$ is sound and refutationally complete, i.e. 
\begin{itemize}
\item[(1)] all rules in $\RPLnull$ are sound and 
\item[(2)] for any unsatisfiable formula $\forall F$ and $F \in \PLnull$ there exists a $\RPLnull$-derivation of $\seq$ from axioms of the form $\seq F\vartheta$ where $\vartheta$ is a renaming of $V(F)$.
\end{itemize}
\end{proposition}
\begin{proof}
(1) is trivial: if $\m$ is a model of the premis(es) of a rule then $\m$ is also a model of the conclusion.\\[1ex]
For proving (2) we first derive the standard clause set $\Ccal$ of $F$. Therefore, we apply the rules of $\RPLnull$ to $\vdash F$, decomposing $F$ into its subformulas, until we cannot apply any rule other than the resolution rule $res$. The last subformula obtained in this way is atomic and hence a clause.
The standard clause set $\Ccal$ of $F$ is comprised of the clauses obtained in this way.
As $\forall F$ is unsatisfiable, its standard clause set is refutable by resolution. Thus, we apply $res$ to the clauses and obtain $\vdash$. The whole derivation lies in $\RPLnull$. 
\end{proof}
We will extend $\RPLnull$ by rules handling schematic formula definitions. But we have to consider another aspect as well: in inductive proofs the use of lemmas is vital, i.e. an ordinary refutational calculus (which has just a weak capacity of lemma generation) may fail to derive the desired invariant. To this aim we extend the calculus by adding some tautological sequent schemata which enrich $\RPLnull$ (which only decomposes formulas) by the potential to derive more complex formulas. Note that our aim is to use the calculi in an interactive way and not fully automatic, which justifies this process of ``anti-refinement''. \\[1ex]
In extending $\RPLnull$ to  a schematic calculus we have to replace unification by s-unification. Formally we have to define how s-substitutions are extended to formula schemata and sequent schemata.
\begin{definition}\label{def.ssubsFS}
Let $\Theta$ be an s-substitution. We define $F\Theta$ for all $F \in \FS$ which do not contain formula variables.
\begin{itemize}
\item Let $P\colon \iota^\alpha \to o \in \Pcal$ and $t_1,\ldots, t_{\alpha} \in T^\iota$. Then $P(t_1,\ldots,t_\alpha)\Theta = P(t_1\Theta,\ldots,$ $t_\alpha\Theta)$
\item Let $\Phat \in \Pcal^\tau$ for $\tau\colon (\omega^{\gamma_1} \to  \iota) \times \ldots \times (\omega^{\gamma_\alpha} \to  \iota) \times \omega^{\beta+1} \to o$, $X_1,\ldots,X_\alpha \in V^G$, $t_1,\ldots,t_{\beta+1} \in T^\omega$ then 
$$\Phat(X_1,\ldots,X_\alpha,t_1,\ldots,t_{\beta+1})\Theta = \Phat(X_1,\ldots,X_\alpha,t_1\Theta\ldots,t_{\beta+1}\Theta).$$
\item  $(\neg F)\Theta = \neg F\Theta$.
\item  If $F_1,F_2 \in \FS$ then 
\begin{eqnarray*}
(F_1 \land F_2)\Theta &=& F_1\Theta \land F_2\Theta,\\ 
(F_1 \lor F_2)\Theta &=& F_1\Theta \lor F_2\Theta.
\end{eqnarray*}
\end{itemize}
Let $S\colon A_1,\ldots, A_\alpha \seq B_1,\ldots,B_\beta$ be a sequent schema. Then 
$$S\Theta =  A_1\Theta,\ldots, A_\alpha\Theta \seq B_1\Theta,\ldots,B_\beta\Theta.$$
\end{definition}
In the resolution rule we have to take care that the sets of variables in $\{A_1,\ldots,A_k\}$ and $\{B_1,\ldots,B_l\}$ are pairwise disjoint. We need a corresponding concept of disjointness for the schematic case.
\begin{definition}[essentially disjoint]\label{def.essdisjoint}
Let $\Acal,\Bcal$ be finite sets of schematic variables in $T^\iota_V$. $\Acal$ and $\Bcal$ are called {\em essentially disjoint} if for all $\sigma \in \Apar$ $\Acal[\sigma] \intrs \Bcal[\sigma] = \emptyset$.
\end{definition} 
\begin{definition}[$\RPLnull^\Psi$]\label{def.RPLnullpsi}
Let $\Psi$ be a schematic formula definition as in Definitions~\ref{def.formschem} and~\ref{def.defeqPS} where
$$\Phat(\vec{Y},\vec{n},0) =F_B,\quad \Phat(\vec{Y},\vec{n},s(m)) = F_S\{\xi \ass \Phat(\vec{Y},\vec{n},m)\},$$
then $\RPLnull^\Psi$ is the extension of $\RPLnull$ by the rules 
\[
\infer[B\Phat r]{\Gamma \seq \Delta, F_B}
  {\Gamma \seq \Delta, \Phat(\vec{Y},\vec{n},0)
  }
\ \ 
\infer[S\Phat r]{\Gamma \seq \Delta,F_S\{\xi \ass \Phat(\vec{Y},\vec{n},m)\} }
 { \Gamma \seq \Delta,\Phat(\vec{Y},\vec{n},s(m))
 }
\]
\[
\infer[B\Phat l]{F_B, \Gamma \seq \Delta}
  {  \Phat(\vec{Y},\vec{n},0),\Gamma \seq \Delta
  }
\ \ 
\infer[S\Phat l]{F_S\{\xi \ass \Phat(\vec{Y},\vec{n},m)\}, \Gamma \seq \Delta}
 {\Phat(\vec{Y},\vec{n},s(m)) ,\Gamma \seq \Delta
 }
\]
for the elimination of defined symbols. For the introduction of defined symbols we invert the rules above:
\[
\infer[B\Phat r^+]{\Gamma \seq \Delta, \Phat(\vec{Y},\vec{n},0)}
  {\Gamma \seq \Delta,F_B}
\ \ 
\infer[S\Phat r^+] { \Gamma \seq \Delta,\Phat(\vec{Y},\vec{n},s(m))}
 {\Gamma \seq \Delta,F_S\{\xi \ass \Phat(\vec{Y},\vec{n},m)\} }
\]
\[
\infer[B\Phat l^+]{  \Phat(\vec{Y},\vec{n},0),\Gamma \seq \Delta} 
{F_B, \Gamma \seq \Delta} 
\ \ 
\infer[S\Phat l^+] {\Phat(\vec{Y},\vec{n},s(m)) ,\Gamma \seq \Delta} 
{F_S\{\xi \ass \Phat(\vec{Y},\vec{n},m)\}, \Gamma \seq \Delta} 
\]
We also adapt the resolution rule to the schematic case: \\
Let $T^\iota_V(\{A_1,\ldots,A_\alpha\}), T^\iota_V(\{B_1,\ldots,B_\beta\})$ be essentially disjoint sets of schematic variables and $\Theta$ be an s-unifier of $\{A_1,\ldots,A_\alpha,B_1,\ldots,B_\beta\}$. Then the resolution rule is defined as  
\[
\infer[\res]{\Gamma\Theta,\Pi\Theta \seq \Delta\Theta,\Lambda\Theta} 
{\Gamma \seq \Delta,A_1,\ldots,A_\alpha
  &
  B_1,\ldots,B_\beta,\Pi \seq \Lambda
}
\]
Moreover we add the following tautological sequent schemata ($\xi_1,\xi_2$ are formula variables): $\xi_1,\xi_2 \seq \xi_1 \land \xi_2$, $\xi_1 \land \xi_2 \seq \xi_1$, $\xi_1 \land \xi_2 \seq \xi_2$, $\xi_1 \seq \xi_1 \lor \xi_2$, $\xi_2 \seq \xi_1 \lor \xi_2$, $\seq \xi_1,\neg \xi_1$, $\xi_1,\neg \xi_1 \seq$. For comfort we may add arbitrary tautological sequent schemata to increase the flexibility and the practical use of the calculus.

\end{definition}
\begin{remark}
It is easy to see that the added tautology schemata together with the cut rule simulate the logical introduction rules for $\land,\lor,\neg$. As $\res$ generalizes the cut rule this is possible also in $\RPLnull^\Psi$. We could instead have added the introduction rules themselves which is logically equivalent. But note that adding additional tautology schemata (besides these defined above) increases the flexibility of rule specification via ``macros''. 
\end{remark}
The  refutational completeness of $\RPLnull^\Psi$ is not as issue as already $\RPLnull$ is refutationally complete. $\RPLnull^\Psi$ is also sound if the defining equations are considered.
\begin{proposition}\label{prop.RPLnull-sound}
Assume that the sequent $S$ is derivable in $\RPLnull^\Psi$. Then $D(\Fcalhat_\omega)$ $\union D(\Fcalhat_\iota) \union D(\Pcalhat) \models S$.
\end{proposition}
\begin{proof}
The introduction and elimination rules for defined predicate symbols are sound with respect to $D(\Pcalhat)$; the resolution rule (involving s-unification ) is sound with respect to $D(\Fcalhat_\omega) \union D(\Fcalhat_\iota)$.
\end{proof} 
\begin{definition}\label{def.cutderiv}
An $\RPLnull^\Psi$ derivation $\varrho$ is called a {\em cut-derivation} if the s-unifiers of all resolution rules are empty.
\end{definition}
\begin{remark}
A cut-derivation is an $\RPLnull^\Psi$ derivation with only propositional rules. Such a derivation can be obtained by combining all unifiers to a global unifier.
\end{remark}
In computing global unifiers we have to apply s-substitutions to proofs. However, not every s-substitution applied to a $\RPLnull^\Psi$ derivation results in a $\RPLnull^\Psi$ derivation again. Just assume that an s-unifier in a resolution is of the form $(X_1(s),X_2(s'))$; if $\Theta = \{(X_1(s),a),(X_2(s'),b)\}$ for different constant symbols $a,b$ then $X_1(s)\Theta$ and $X_2(s')\Theta$ are no longer unifiable and the resolution is blocked. 
\begin{definition}\label{def.subsadmit}
Let $\rho$ be a derivation in $\RPLnull^\Psi$ which does not contain the resolution rule; then for any s-substitution $\Theta$ $\rho\Theta$ is the derivation in which every sequent occurrence $S$ is replaced by $S\Theta$. We say that $\Theta$ is admissible for $\rho$. Now let $\rho=$
\[
\infer[\res]{\Gamma\Theta',\Pi\Theta' \seq \Delta\Theta',\Lambda\Theta'} 
{\deduce{\Gamma \seq \Delta,A_1,\ldots,A_\alpha}{(\rho_1)}
  &
  \deduce{B_1,\ldots,B_\beta,\Pi \seq \Lambda}{(\rho_2)}
}
\]
where $\Theta'$ is an s-unifier of $\{A_1,\ldots,A_\alpha,B_1,\ldots,B_\beta\}$. 
Let us assume that $\Theta$ is admissible for $\rho_1$ and $\rho_2$. We define that $\Theta$ is admissible for $\rho$ if the set 
$$U\colon \{A_1\Theta,\ldots,A_\alpha\Theta,B_1\Theta,\ldots,B_\beta\Theta\}$$
is s-unifiable. If $\Theta^*$ is an s-unifier of $U$ then we can define $\rho\Theta$ as 
\[
\infer[\res]{\Gamma\Theta\Theta^*,\Pi\Theta\Theta^* \seq \Delta\Theta\Theta^*,\Lambda\Theta\Theta^*} 
{\deduce{\Gamma\Theta \seq \Delta\Theta,A_1\Theta,\ldots,A_\alpha\Theta}{(\rho_1\Theta)}
  &
  \deduce{B_1\Theta,\ldots,B_\beta\Theta,\Pi\Theta \seq \Lambda\Theta}{(\rho_2\Theta)}
}
\]
\end{definition}
\begin{definition}\label{def.globunif}
Let $\varrho$ be an $\RPLnull^\Psi$ derivation and $\Theta$ be an s-substitution which is admissible for $\varrho$. $\Theta$ is called a {\em global unifier} for $\varrho$ if $\varrho\Theta$ is a cut-derivation.
\end{definition}
In order to compute global unifiers we need $\RPLnull^\Psi$ derivations in some 
kind of ``normal form''. Below we define two necessary restrictions on derivations.
\begin{definition}\label{def.normderiv}
An $\RPLnull^\Psi$ derivation $\varrho$ is called {\em normal} if all s-unifiers of resolution rules in $\varrho$ are normal and restricted.
\end{definition}
\begin{remark}
Note that, in case of s-unifiability, we can always find normal and restricted s-unifiers; thus the definition above does not really restrict the derivations, it only requires some renamings.
\end{remark}
\begin{definition}\label{def.regulderiv}
An $\RPLnull^\Psi$ derivation $\varrho$ is called {\em regular} if for all subderivations $\varrho'$ of $\varrho$ of the form 
\[
\infer[res(\Theta)]{\Gamma\Theta,\Pi\Theta \seq \Delta\Theta,\Lambda\Theta}
 { \deduce{\Gamma \seq \Delta,A_1,\ldots, A_\alpha}{(\varrho'_1)}
   &
  \deduce{B_1,\ldots,B_\beta,\Pi \seq \Lambda}{(\varrho'_2)}
 }
\]
we have $V^G(\varrho'_1) \intrs V^G(\varrho'_2) = \emptyset$
\end{definition}
Note that the condition  $V^G(\varrho'_1) \intrs V^G(\varrho'_2) = \emptyset$ in Definition~\ref{def.regulderiv}  guarantees that, for all parameter assigments $\sigma$, $\varrho'_1[\sigma]$ and $\varrho'_2[\sigma]$ are variable-disjoint. \\[1ex]
We write $\varrho' \subsumes \varrho$ if there exists an s-substitution $\Theta$ such that $\varrho'\Theta = \varrho$. 
\begin{proposition}\label{prop.normreg}
Let $\varrho$ be a normal   $\RPLnull^\Psi$ derivation. Then there exists a $\RPLnull^\Psi$ derivation $\varrho'$ such that $\varrho' \subsumes \varrho$ and $\varrho'$ is normal and regular.
\end{proposition}
\begin{proof}
By renaming of variables in subproofs and in s-unfiers. 
\end{proof}
\begin{proposition}\label{prop.golbalunif}
Let $\varrho$ be a normal and regular $\RPLnull^\Psi$ derivation. Then there exists a global s-unifier $\Theta$ for $\varrho$ which is normal and $V^G(\Theta) \IN V^G(\varrho)$.
\end{proposition}
\begin{proof}
By induction on the number of inferences in $\varrho$.\\[1ex]
Induction base: $\varrho$ is an axiom. $\emptyset$ is a global s-unifier which trivially fulfils the properties.\\[1ex]
For the induction step we distinguish two cases.
\begin{itemize}
\item The last rule in $\varrho$ is unary. Then $\varrho$ is of the form 
\[
\infer[\xi]{\Gamma \seq \Delta}
 { \deduce{\Gamma' \seq \Delta'}{(\varrho')}
 } 
\]
By induction hypothesis there exists a global substitution $\Theta'$ which is a global unifier for $\varrho'$ such that $\Theta'$ is normal and 
$V^G(\Theta') \IN V^G(\varrho')$. We define $\Theta = \Theta'$. Then, trivially, $\Theta$ is normal and a global unifier of $\varrho$. Moreover, by definition of the unary rules in $\RPLnull^\Psi$, we have $V^G(\varrho') = V^G(\varrho)$ and so $V^G(\Theta) \IN V^G(\varrho)$.
\item $\varrho$ is of the form
\[
\infer[res(\Theta)]{\Gamma\Theta,\Pi\Theta \seq \Delta\Theta,\Lambda\Theta}
 { \deduce{\Gamma \seq \Delta,A_1,\ldots,A_\alpha}{(\varrho_1)}
   &
  \deduce{B_1,\ldots,B_\beta,\Pi \seq \Lambda}{(\varrho_2)}
} 
\]
As $\varrho$ is a normal $\RPLnull^\Psi$ derivation the unifier $\Theta$ is normal. By regularity of $\varrho$ we have $V^G(\varrho_1) \intrs V^G(\varrho_2) = \emptyset$.\\[1ex]
By induction hypothesis there exist global normal unifiers $\Theta_1,\Theta_2$ for  $\varrho_1$ and $\varrho_2$ such that $V^G(\Theta_1) \IN V^G(\varrho_1)$ and $V^G(\Theta_2) \IN V^G(\varrho_2)$. By  $V^G(\varrho_1) \intrs V^G(\varrho_2) = \emptyset$ we also have $V^G(\Theta_1) \intrs V^G(\Theta_2) = \emptyset$.\\[1ex]
 We show now that $(\Theta_1,\Theta)$ and $(\Theta_2,\Theta)$ are composable. 
As $\Theta_1$ is normal we have for all $\sigma \in \Apar$ 
$$V^\iota(\{A_1,\ldots,A_\alpha\}[\sigma]) \intrs \dom(\Theta_1[\sigma]) = \emptyset.$$
Similarly we obtain 
$$V^\iota(\{B_1,\ldots,B_\beta\}[\sigma]) \intrs \dom(\Theta_2[\sigma]) = \emptyset.$$
As $\Theta$ is normal and restricted we have for all $\sigma \in \Apar$ 
$$V^\iota(\Theta[\sigma]) \IN V^\iota(\sigma\{A_1,\ldots,A_\alpha,B_1,\ldots,B_\beta\}\Eval_o).$$
Therefore $(\Theta_1,\Theta)$ and $(\Theta_2,\Theta)$ are both composable. As $\Theta_1,\Theta_2,\Theta$ are normal so are $\Theta_1 \star \Theta$ and $\Theta_2 \star \Theta$. As $\Theta_1,\Theta_2$ are essentially disjoint we can define 
$$\Theta(\varrho) = \Theta_1 \star \Theta \union \Theta_2 \star \Theta.$$
$\Theta(\varrho)$ is a normal s-substitution and $V^G(\Theta(\varrho)) \IN V^G(\varrho)$.\\[1ex]
$\Theta(\varrho)$ is also a global unifier of $\varrho$. Indeed, $\varrho_1\Theta(\varrho)=$
\[
\deduce{\Gamma\Theta \seq \Delta\Theta,A_1\Theta,\ldots,A_1\Theta}{(\varrho_1\Theta(\varrho))}
\]
and 
$\varrho_2\Theta(\varrho)=$
\[
\deduce{A_1\Theta,\ldots,A_1\Theta,\Pi\Theta \seq \Lambda\Theta}{(\varrho_2\Theta(\varrho))}
\]
So we obtain the derivation 
\[
\infer[\cut]{\Gamma\Theta,\Pi\Theta \seq \Delta\Theta,\Lambda\Theta}
 { \varrho_1\Theta(\varrho)
   &
  \varrho_2\Theta(\varrho)
 }
\]
which is an instance of $\varrho$ and a cut derivation (note that every instance of a cut derivation is a cut derivation as well).

\item $\varrho$ is of the form
\begin{prooftree}
	\AxiomC{$(\varrho_1)$}
	\noLine
	\UnaryInfC{$\Gamma \vdash \Delta$}
	\AxiomC{$(\varrho_2)$}
	\noLine
	\UnaryInfC{$\Pi \vdash \Lambda$}
	\RightLabel{$\chi$}
	\BinaryInfC{$\Gamma', \Pi' \vdash \Delta', \Lambda'$}
\end{prooftree}
where $\chi$ is a binary introduction rule.

As $\varrho$ is a normal $\RPLnull^\Psi$ derivation all occurring $s$-unifiers in $\varrho_1$ and $\varrho_2$ are normal. By regularity of $\varrho$ we have that $V^G(\varrho_1) \intrs V^G(\varrho_2) = \emptyset$.\\[1ex]
By induction hypothesis there exist global normal unifiers $\Theta_1,\Theta_2$ for $\varrho_1$ and $\varrho_2$ such that $V^G(\Theta_1) \IN V^G(\varrho_1)$ and $V^G(\Theta_2) \IN V^G(\varrho_2)$. By  $V^G(\varrho_1) \intrs V^G(\varrho_2) = \emptyset$ we also have $V^G(\Theta_1) \intrs V^G(\Theta_2) = \emptyset$. Moreover, there is no overlap between the domain variables of the unifiers $\Theta_1$ and $\Theta_2$, i.e. $\dom(\Theta_1[\sigma]) \intrs \dom(\Theta_2(\sigma)) = \emptyset$ for all $\sigma \in \Apar$. Therefore, we can define $\Theta = \Theta_1 \union \Theta_2$, which is obviously a global s-unifier of $\varrho$. Furthermore, $V^G(\Theta) = V^G(\Theta_1) \union V^G(\Theta_2)$, therefore $V^G(\Theta) \IN V^G(\varrho_1) \union V^G(\varrho_2)$ and by definition of binary introduction rules in $\RPLnull^\Psi$, we have $V^G(\Theta) \IN V^G(\varrho)$.
\end{itemize}
\end{proof}

\begin{example}
\label{basecaseProof}
We provide a simple $\RPLnull^\Psi$ refutation using the schematic formula constructed in Example \ref{1smaform}. We will only cover the $\RPLnull^\Psi$ derivation of the base case and wait for the introduction of proof schemata to provide a full refutation. We abbreviate $X_1,\cdots,X_5$ by $\mathbf{X}$. 
\begin{tiny}
\begin{prooftree}
\AxiomC{$  \vdash   \hat{F_1}(\mathbf{X},0,0)$}
\RightLabel{$S\hat{F_1}r$}
\UnaryInfC{$  \vdash   \hat{F_2}(\mathbf{X},0,0) \wedge \hat{F_3}(\mathbf{X},0,0)$}
\RightLabel{$\wedge:r$}
\UnaryInfC{$  \vdash  \hat{F_2}(\mathbf{X},0,0)$}
\RightLabel{$B\hat{F_2}r$}
\UnaryInfC{$ \vdash \leqpred{X_4(0,0)}{0}\vee \eqpred{X_4(0,0)}{0} $}
\RightLabel{$\vee:r$}
\UnaryInfC{$ \vdash \leqpred{X_4(0,0)}{0}, \eqpred{X_4(0,0)}{0} $}
\RightLabel{$B\hat{S}r$}
\UnaryInfC{$ \vdash \leqpred{X_4(0,0)}{0}, \eqpred{X_4(0,0)}{0} $}
\noLine
\UnaryInfC{(2)}
\end{prooftree}
\begin{prooftree}
\AxiomC{(2)}
\AxiomC{$  \vdash   \hat{F_1}(\mathbf{X},0,0)$}
\RightLabel{$S\hat{F_1}r$}
\UnaryInfC{$  \vdash   \hat{F_2}(\mathbf{X},0,0) \wedge \hat{F_3}(\mathbf{X},0,0) $}
\RightLabel{$\wedge:r$}
\UnaryInfC{$  \vdash   \hat{F_3}(\mathbf{X},0,0)$}
\RightLabel{$S\hat{F_3}r$}
\UnaryInfC{$  \vdash   \hat{F_5}(X_{2}(0,0),\mathbf{X},0,0) \wedge \nleqpred{\mathbf{0}}{0}$}
\UnaryInfC{$  \vdash   \hat{F_5}(X_2(0,0),\mathbf{X},0,0)$}
\RightLabel{$B\hat{F_5}r$}
\UnaryInfC{$  \vdash   \neqpred{\hat{S}(X_2(0,0),0)}{0}$}
\RightLabel{$B\hat{S}r$}
\UnaryInfC{$  \vdash   \neqpred{X_2(0,0)}{0}$}
\RightLabel{$\neg:r$}
\UnaryInfC{$  \eqpred{X_2(0,0)}{0} \vdash $}
\RightLabel{Res$\left(\left\lbrace\begin{array}{c}  X_4(0,0)\leftarrow X_2(0,0) \end{array}\right\rbrace  \right) $}
\BinaryInfC{$\vdash \leqpred{X_4(0,0)}{0}$}
\noLine
\UnaryInfC{$(1)$}
\end{prooftree}

\begin{prooftree}
\AxiomC{$(1)$}
\noLine
\UnaryInfC{$ \vdash \leqpred{X_4(0,0)}{0}$}
\AxiomC{$  \vdash   \hat{F_1}(\mathbf{X},0,0)$}
\RightLabel{$S\hat{F_1}r$}
\UnaryInfC{$  \vdash   \hat{F_2}(\mathbf{X},0,0) \wedge \hat{F_3}(\mathbf{X},0,0) $}
\RightLabel{$\wedge:r$}
\UnaryInfC{$  \vdash  \hat{F_3}(\mathbf{X},0,0)$}
\RightLabel{$B\hat{F_3}r$}
\UnaryInfC{$ \vdash \hat{F_5}(X_1(0,0),\mathbf{X},0,0)\wedge \nleqpred{a}{0}$}
\RightLabel{$\wedge:r_2$}
\UnaryInfC{$  \vdash \nleqpred{a}{0}$}
\RightLabel{$\neg:r$}
\UnaryInfC{$  \leqpred{a}{0} \vdash$}
\RightLabel{Res$\left(\left\lbrace X_4(0,0)\leftarrow a \right\rbrace  \right) $}
\BinaryInfC{$\vdash$}
\noLine
\UnaryInfC{$(\delta_0, 0,0)$}

\end{prooftree}
\end{tiny}
\end{example}
Derivations in $\RPLnull^\Psi$ are defined like for $\RPLnull$. But derivations in $\RPLnull^\Psi$ with ordinary axioms do not suffice to describe schematic derivations. We need the additional concept of {\em call graphs} which can be decorated by $\RPLnull^\Psi$ derivations in such a way that together they provide a sound derivation in  $\RPLnull^\Psi$. 

\section{A Scaffolding for Schematic Derivations}
\label{Scaffolding}
In this section we define a scaffolding which supports the construction of schematic derivations. These are sequences of proofs pieced together from derivations whose leaves may be non-axiomatic initial sequents. These non-axiomatic initial sequents match the end sequent of other derivations and when pieced together form a derivation containing axiomatic initial sequents only. Usually schematic derivations have a starting derivation which contains {\em free parameters}. By substituting these free parameters by natural numbers and propagating the substitutions we can construct a resolution derivation. 

Each one of these derivations can be thought of as a control mechanism which directs a {\em flow} of substitutions through a network of {\em junctions}. The non-axiomatic initial sequents can be thought of as the {\em junctions} at which two flows connect. As one changes the substitutions one changes how the substitutions flow through the network of junctions and thus the final output proof. This analogy to network flow is precisely how our scaffolding should be interpreted. Together, a collection of coalescing flow controls describe what we refer to as a {\em call graph} which  defines, finitistically, how the flow is controlled within a network of junctions.  We will show how this flow control mechanism can be used to provide a semantic foundation for schematic $\RPLnull^\Psi$ derivations discussed later in this paper.

\subsection{Flows over a Junction Network}
Rather than defining our framework over individual terms of $T^\omega$, we instead consider $m$-tuples of terms which we refer to as {\em points of length $m$}.  Points are divided into sets $\left\lbrace (t_1,\cdots, t_m)\ \middle\vert\ \right.  $ $\left.  t_1,\cdots, t_m\in T^\omega\right\rbrace$, denoted by $\Tvs{m}$, containing all points of the same length. We refer to these sets as {\em point spaces} of length $m$. Additionally, we distinguish the subset   $\left\lbrace (t_1,\cdots, t_m)\ \middle\vert\ t_1,\cdots, t_m\in \mathit{Num}\right\rbrace$, denoted by  $\Tvs{m}_{0}$. We refer to this subset of $\Tvs{m}$ as the {\em concrete point space} of length $m$. Associated with $\Tvs{m}_{0}$ is a well founded total order $\prec_m$ which orders points of length $m$. 

We extend the ordering of $\prec_m$ to $\Tvs{m}$ by extending the evaluation procedure for $T^\omega$  discussed in  Section~\ref{sec.schematicproofs}. Given a point $\mathbf{v} = (t_1,\cdots, t_m) \in \Tvs{m}$ and $\sigma\in \mathcal{S}$, then  $\sigma(\mathbf{v})\Eval_{\omega} = ( \sigma(t_1)\Eval_{\omega},\cdots,  \sigma(t_m)\Eval_{\omega}) \in \Tvs{m}_0$. Now let $S\subseteq \mathcal{S}$ and $\mathbf{v},\mathbf{v}'\in \Tvs{m}$. Then  $\mathbf{v}\prec_{m}^S \mathbf{v}'$ iff for every $\sigma \in S$, $\sigma(\mathbf{v})\Eval_{\omega}\prec_m \sigma(\mathbf{v}')\Eval_{\omega}$. 

Points are paired with a symbol from a countably infinite set $\Delta$ to form {\em junctions}. However, not every pairing of symbol and point is desired, in particular each symbol should only be paired with points of a specific length. To enforce this restriction we introduce a so called {\em arity function} $\Af:\Delta \rightarrow \mathbb{N}$. Thus, if for some symbol $\delta \in \Delta$, $\Af(\delta) = m$ a well-formed junction would pair $\delta$ with a point of length $m$. We refer to junctions which are well-formed with respect to an arity function $\Af$ as  {\em $\Af$-junctions}. The set of all $\Af$-junctions will be referred to as an {\em  $\Af$-junction network $\mathcal{J}^{\star}(\Af)$}. 

Note that $\prec_{m}^S$  may be easily extended to junctions whose points come from the same point space,  however, this leaves many junctions incomparable. Thus, we extend $\prec_{m}^S$ to $\prec_{\mathcal{J}^{\star}(\Af)}^{S}$ as follows: 

\begin{definition}[network order] 
\label{networkorder}
Let  $(\delta, p), (\delta', q)\in \mathcal{J}^{\star}(\Af)$ and $S\subseteq \mathcal{S}$. Then $(\delta, p)\prec_{\mathcal{J}^{\star}(\Af)}^{S} (\delta', q)$ if either $p \prec_{\Af(\delta)}^S q $, or  $\Af(\delta') < \Af(\delta)$. When possible, we use the simplified notation $\prec^{S}$.
 \end{definition}
 Notice that $\prec^{S}$, for certain choices of arity function, is not well founded. While this is in general problematic for termination, we will only consider finite subsets of $\Delta$ and thus infinite descending chains can be avoided.

\begin{definition}
\label{curves} 
Let $\mathcal{S}^{*}$ be a partitioning of $\mathcal{S}$ into mutually disjoint sets and $\Pf{\mathcal{S}^{*}}{\Af}:\mathcal{S}^{*}\rightarrow \mathcal{J}^{\star}(\Af)$ an injective mapping from $\mathcal{S}^{*}$ to finite subsets of $\mathcal{J}^{\star}(\Af)$. We refer to   $\Pf{\mathcal{S}^{*}}{\Af}$ as a {\em flow} if the following conditions are met:
\begin{itemize}
\item[1)] $\exists! j\in \mathcal{J}^{\star}(\Af)$ $\forall S\in\mathcal{S}^{*}$ , s.t. $j \in \Pf{\mathcal{S}^{*}}{\Af}(S)$. This junction is  referred to as the {\em source} of the flow and will be denoted by  $\left[\Pf{\mathcal{S}^{*}}{\Af}\right]$.

\item[2)] $\forall S\in\mathcal{S}^{*}$ and $q\in \Pf{\mathcal{S}^{*}}{\Af}(S)$, $q \prec^{S} \left[ \Pf{\mathcal{S}^{*}}{\Af}\right]$ or $q = \left[ \Pf{\mathcal{S}^{*}}{\Af}\right] $ .
\end{itemize}
\end{definition}

\begin{example}
\label{ex:simpleFlow}
In order to illustrate the power of the flow formalism let us consider the flow:
$$ \Pf{\mathcal{S}^{*}}{\Af}= \left\lbrace \begin{array}{c} \left( \lbrace (\delta, n) ,(\delta, \mathit{p}(n)) \rbrace,S_1\right)\ ,\ \left( \lbrace(\delta, n) \rbrace, S_2\right) \end{array} \right\rbrace,$$
where $S_1 = \lbrace \sigma \in \mathcal{S} \ \&\  \sigma(n)\Eval_{\omega} > 0\rbrace$ and $S_2 = \lbrace \sigma \in \mathcal{S} \ \&\  \sigma(n)\Eval_{\omega} = 0\rbrace$. Note that $S_1\cup S_2 \equiv \mathcal{S}$. This flow is sourced from $(\delta, n)$, i.e. $ \left[ \Pf{\mathcal{S}^{*}}{\Af}\right] = (\delta, n)$. Evaluation at $\lbrace n\leftarrow s^{3}(0) \rbrace$ implies finding the partition this assignment belongs to, in this case $S_1$. Thus, $\Pf{\mathcal{S}^{*}}{\Af}(\lbrace n\leftarrow s^{3}(0) \rbrace) = \Pf{\mathcal{S}^{*}}{\Af}(S_1) = \lbrace (\delta, n) ,(\delta, \mathit{p}(n)) \rbrace$.  Notice that $(\delta, \mathit{p}(n)) \prec^{S_1} (\delta, n)$ when $p(\cdot)$ denotes the predecessor function. 

This flow formalizes primitive recursion, however, without defining normalization this is not entirely obvious. In the following section we introduce {\em call graphs} which provide a normalization procedure as well as flow composition. 
\end{example}

\begin{definition}
A flow $\Pf{\mathcal{S}^{*}}{\Af}$ is {\em regular} if $\left\lbrace j \ \middle\vert\ j \in \Pf{\mathcal{S}^{*}}{\Af}(S) \ , \ S\in \mathcal{S}^* \right\rbrace $
is finite.  The set of all regular flows over a $\Af$-junction network is denoted by $\regFlow{\Af}$.
\end{definition}
Frequently we may write $P_{\mathcal{S}^*}$ if the arity is not important.

While non-regular flows may be interesting in their own right, for this work we need only to consider regular flows being that they directly correspond to {\em well-formed proof schemata}.

\subsection{Call Graphs}
 When multiple flows are defined within the same junction network their intersections provide a graph-like structure. A {\em call graph} is a special case of flow intersection where each junction occuring in a flow is the source of a flow.
\begin{figure}
\vspace*{-2em}
\begin{center}
\scalebox{.75}{\begin{tikzpicture}
\draw (.5,.7) circle (.5cm);
\draw (4,.77) circle (.5cm);
\draw (8,2.13) circle (.5cm);
\draw[dotted] (8.28,2.4) circle (.3cm);
\draw[dotted] (4.19,.52) circle (.6cm);
\draw[dotted] (.2,.4) circle (.4cm);

\draw [black, xshift=4cm] plot [smooth, tension=1] coordinates { (-5,1) (1,1)  (7,4)};
\draw [black, xshift=4cm] plot [smooth, tension=1] coordinates { (-5,3) (-1.75,-.7) (3.2,3.5)  (4.6,-1)};

\node[rotate=0] (C1) at (.5,1.4)  {$P_1(\sigma_1)$};
\node[rotate=0] (C2) at (4,1.47)  {$P_1(\sigma_2)$};
\node[rotate=0] (C3) at (7.1,2.4)  {$P_1(\sigma_3)$};

\node[rotate=0] (C11) at (-.5,0)  {$P_2(\theta_1)$};
\node[rotate=0] (C22) at (5.2,0)  {$P_2(\theta_2)$};
\node[rotate=0] (C33) at (8.75,2.9)  {$P_2(\theta_3)$};

\node[rotate=0] (C3) at (9,-.3)  {$P_2$};
\node[rotate=0] (C3) at (10,3)  {$P_1$};

\node[] (A) at (8.28,2.4)  {$\bullet$};
\node[] (B) at (4.19,.52)  {$\bullet$};
\node[] (C) at (.2,.4)  {$\bullet$};

\node[] (D) at (8,2.13)  {$\bullet$};
\node[] (E) at (.5,.7)  {$\bullet$};
\node[] (F) at (4,.77)  {$\bullet$};

\end{tikzpicture}}
\end{center}
\vspace*{-2em}
\caption{The intersection of two flows within a junction network.}
\label{cellfig}
\end{figure}
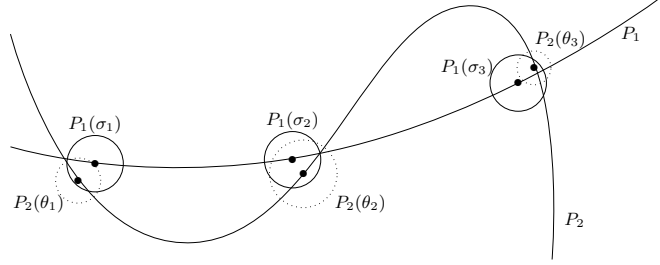

\begin{definition}
\label{callgraph}
A finite set of flows $\mathcal{G}$ over $\mathcal{J}^{\star}(\Af)$ is referred to as a {\em call graph} if for every $P_{\mathcal{S}^{*}_1}\in \mathcal{G}$, $S\in \mathcal{S}^{*}$, $\sigma\in S$, $j \in P_{\mathcal{S}^{*}_1}$ there exists a unique $P_{\mathcal{S}^{*}_2}\in \mathcal{G}$ and $\theta \in \mathcal{S}$ s.t. $\theta([P_{\mathcal{S}^{*}_2}])\Eval_{\omega} = \sigma(j)\Eval_{\omega}$. We write $\mathit{flow}(j,\sigma) = P_{\mathcal{S}^{*}_2}$ and $\mathit{subst}(j,\sigma) = \theta$.
We refer to $\mathcal{G}$ as {\em finite} when $|\mathcal{G}| \in \mathbb{N}$. The set of all finite call graphs definable over a junction network $\mathcal{J}^{\star}(\Af)$ is denoted by $\mathcal{G}^{\star}(\Af)$.
\end{definition}
When $P_{\mathcal{S}^{*}}(S) = \left\lbrace [P_{\mathcal{S}^{*}}] \right\rbrace $ for some $S\in\mathcal{S}^{*}$ we refer to $S$ as a {\em sink} of $P_{\mathcal{S}^{*}}$. The sinks represent end points of the flow.

\begin{example}
\label{callgraphex}
Let $\mathcal{G} = \lbrace P_1, P_2 \rbrace$ be a call graph over $\mathcal{J}^{\star}(\Af)$ , where
$$ P_1 = \lbrace (\lbrace (\delta, n) ,(\delta', n, \mathit{p}(n),n, 0) \rbrace , \mathcal{S})\rbrace $$ 
$$ P_2 = \left\lbrace \begin{array}{c} \left(  \lbrace (\delta', n,m,k,w) ,(\delta', n,m,\mathit{p}(k),s(w)) \rbrace , S_1\right)  \ ,\ \\
\left( \lbrace(\delta', n,m,k,w), (\delta', n,\mathit{p}(m),n,w) \rbrace , S_2 \right) \ ,\ \\
\left( \lbrace(\delta', n,m,k,w)\rbrace , S_3\right) 
\end{array} \right\rbrace$$
and $S_1 = \left\lbrace \sigma \vert \sigma \in \mathcal{S}, \sigma(k)\Eval_{\omega} > 0  \right\rbrace$, $S_2 = \left\lbrace \sigma \vert \sigma \in \mathcal{S},\sigma(k)\Eval_{\omega} = 0 \ \&\  \sigma(m)\Eval_{\omega} > 0 \right\rbrace$, and $S_3 = \left\lbrace \sigma \vert \sigma \in \mathcal{S}, \sigma(k)\Eval_{\omega} = 0 \ \&\  \sigma(m)\Eval_{\omega} = 0 \right\rbrace$. If we assume that $\prec_{4}$ is the lexicographical order then $P_1$ and $P_2$ respect the network order and are flows. Furthermore, the component $\left( \lbrace(\delta', n,m,k,w)\rbrace , S_3\right)$ of  $P_2$ denotes the {\em sinks} of $P_2$ being that it matches every substitution of  $S_3$ with a singleton set of junctions. 

Let us now consider the relationship between the input parameter assignment $\sigma$ and the assignment connecting the flows of the call graph. For clarity reasons we associated with each position within the points a numeric term $\alpha_i$. For example, applying the parameter assignment  $\sigma = \left\lbrace n\leftarrow \alpha_1,m \leftarrow \alpha_2, k \leftarrow \alpha_3 ,w\leftarrow \alpha_4\right\rbrace$ to the junction $(\delta', n,m,k,w)$ may be written as $(\delta', n,m,k,w)\sigma$. We can write  $P_1(S)$, such that $S\in S^{*}$ and $\sigma\in S$, as follows:
$$\lbrace (\lbrace (\delta,n)\sigma,(\delta', n, m,k,w)\mathit{subst}((\delta', n, m,k,w),\sigma) \rbrace , \mathcal{S})\rbrace $$ 
A call graph $G$ essentially defines a set of processes which generate sequences of parameter assignments of the form $$\left[ \sigma, \mathit{subst}(j_1,\sigma)\right],\left[\theta_1,\mathit{subst}(j_2,\theta_1)\right],\left[ \theta_2 , \mathit{subst}(j_3,\theta_2)\right],\ldots $$ 
where $\theta_{1} = \mathit{subst}(j_1,\sigma)$ and $\theta_2 = \mathit{subst}(j_2,\theta_1)$. Furthermore, if we generate the above sequence at the flow $P_{S^{*}_1}\in G$ then $j_1\in \mathit{flow}(\left[ P_{S^{*}_1} \right],\sigma)(S_1)$, where $S_1\in S^{*}_1$ and $\sigma \in S_1$ , $j_2\in \mathit{flow}(j_1,\theta_1)(S_2)$, where $\mathit{flow}(j_1,\theta_1)$ is defined over the partitioning $S_{2}^{*}$, $S_{2} \in S_{2}^{*}$ and $\theta_1\in S_2$ , and $j_3\in \mathit{flow}(j_2,\theta_2)(S_3)$, where $\mathit{flow}(j_2,\theta_2)$ is defined over the partitioning $S_{3}^{*}$, $S_{3} \in S_{3}^{*}$ and $\theta_1\in S_3$.

Assuming we start traversing the above call graph from $P_1$ and that $\alpha_1>0$  the following sequence of parameter assignment pairs will emerge. 

$$\left[ \overbrace{\left\lbrace n \leftarrow \alpha_1 \right\rbrace}^{\sigma}, \overbrace{\left\lbrace \begin{array}{c} n\leftarrow \alpha_1\ , \ m \leftarrow p(\alpha_1) ,\\ k \leftarrow \alpha_1\ , \ w \leftarrow 0\end{array} \right\rbrace }^{\theta_1} \right] $$
where  $j_1 = (\delta', n, \mathit{p}(n),n, 0)$, $\mathit{flow}(j_1,\sigma) = P_2$, and $\mathit{Subst}(j_1,\sigma) = \theta_1$. 
$$\left[ \overbrace{\left\lbrace \begin{array}{c} n\leftarrow \alpha_1\ ,\ m \leftarrow p(\alpha_1) ,\\ k \leftarrow \alpha_1,w \leftarrow 0 \end{array}\right\rbrace }^{\theta_1}, \overbrace{\left\lbrace\begin{array}{c}  n\leftarrow \alpha_1\ ,\ m \leftarrow p(\alpha_1) ,\\ k \leftarrow p(\alpha_1),w \leftarrow s(0)\end{array} \right\rbrace }^{\theta_2} \right] $$
where  $j_2 = (\delta', n,m,\mathit{p}(k),s(w))$, $\mathit{flow}(j_2,\theta_1) = P_2$, and $\mathit{Subst}(j_2,\theta_1) = \theta_2$. 

$$\left[ \overbrace{\left\lbrace \begin{array}{c} n\leftarrow \alpha_1\ ,\ m \leftarrow p(\alpha_1) ,\\ k \leftarrow p(\alpha_1)\ ,\ w \leftarrow s(0)\end{array} \right\rbrace}^{\theta_2} , \overbrace{\left\lbrace \begin{array}{c} n\leftarrow \alpha_1\ ,\ m \leftarrow p(\alpha_1) ,\\ k \leftarrow p(p(\alpha_1))\ ,\ w \leftarrow s(s(0))\end{array} \right\rbrace}^{\theta_3}  \right] $$
where  $j_3 = (\delta', n,m,\mathit{p}(k),s(w))$, $\mathit{flow}(j_3,\theta_2) = P_2$, and $\mathit{Subst}(j_3,\theta_2) = \theta_3$.
$$\vdots$$

$$\left[ \overbrace{\left\lbrace \begin{array}{c} n\leftarrow \alpha_1\ ,\ m \leftarrow p(\alpha_1) ,\\ k \leftarrow p^{\alpha_1}(\alpha_1)\ ,\ w \leftarrow s^{\alpha_1}(0)\end{array} \right\rbrace}^{\theta_3} , \overbrace{\left\lbrace \begin{array}{c} n\leftarrow \alpha_1\ ,\ m \leftarrow p(p(\alpha_1)) ,\\ k \leftarrow \alpha_1\ ,\ w \leftarrow s^{\alpha_1}(0)\end{array}\right\rbrace}^{\theta_{4} } \right] $$
where  $j_4 = (\delta', n,\mathit{p}(m),n,w)$, $\mathit{flow}(j_4,\theta_4) = P_2$, and $\mathit{Subst}(j_4,\theta_3) = \theta_4$. 

$$\vdots$$

$$\left[  \overbrace{\left\lbrace \begin{array}{c} n\leftarrow \alpha_1\ ,\ m \leftarrow p^\alpha_1(\alpha_1) ,\\ k \leftarrow 0\ ,\ w \leftarrow s^{(\alpha_1)^2}(0)\end{array} \right\rbrace }^{\theta_4} \right] $$
Being that we have reached a sink at this point $\mathit{flow}$ and $\mathit{Subst}$ are only defined for the source of $P_2$. 
\end{example}
\subsection{Call Graph Traces}
Given a finite call graph $\mathcal{G}$, a flow $P\in \mathcal{G}$ and a assignment $\sigma\in \mathcal{S}$ we may consider the $\sigma$-$\theta$ transition defined by the flow through the call graph from source to sinks. We refer to the tree of junctions, which is constructed as an assignment passes through  the call graph towards the sinks, as {\em the trace of $\sigma$ at $P$ in $\mathcal{G}$}. The main result of this section is that the trace of an assignment is always a finite tree.  

However, before we define call graph traces let us consider the call graph for primitive recursion $\mathcal{G} = \left\lbrace P\right\rbrace$ using the flow from Example~\ref{ex:simpleFlow} repeated below:
$$P= \left\lbrace \begin{array}{c} \left( \lbrace (\delta, n) ,(\delta, \mathit{p}(n)) \rbrace,S_1\right)\ ,\ \left( \lbrace(\delta, n) \rbrace, S_2\right) \end{array} \right\rbrace,$$
where $S_1 = \lbrace \sigma \in \mathcal{S} \ \&\  \sigma(n)\Eval_{\omega} > 0\rbrace$ and $S_2 = \lbrace \sigma \in \mathcal{S} \ \&\  \sigma(n)\Eval_{\omega} = 0\rbrace$. Starting from any parameter assignment $\sigma = \left\lbrace n\leftarrow s^{\alpha}(0)\right\rbrace$ for $\alpha\in \mathbb{N}$, the trace of $\sigma$ at $P$ in $\mathcal{G}$ is the following sequence of junctions: 
$$(\delta,  s^{\alpha}(0)), (\delta,  s^{\alpha-1}(0)), \cdots, (\delta,  s^{1}(0)), (\delta,  0).$$
Note that if we take any two adjacent junctions in the above sequence they form the set of junctions associated with  the partition  $S_1$ after evaluation by $\left\lbrace n\leftarrow s^{\beta}(0)\right\rbrace$, for $\alpha\geq \beta >0$. Looking back at Definition~\ref{callgraph}, and considering a pair of junctions from the above trace, say $(\delta,  s^{\beta }(0))$ and  $(\delta,  s^{\beta -1}(0))$, the above mentioned association with $S_1$ holds only when the assignment is $\left\lbrace n\leftarrow s^{\beta }(0)\right\rbrace$, however Definition~\ref{callgraph} also requires us to find an assignment and a flow such that $(\delta,  s^{\beta -1}(0))$ is the source of a flow under that assignment. This assignment would be $\left\lbrace n\leftarrow s^{\beta -1}(0)\right\rbrace$ thus transitioning us from the pair $(\delta,  s^{\beta }(0))$ and  $(\delta,  s^{\beta -1}(0))$ to the pair $(\delta,  s^{\beta -1}(0))$ and  $(\delta,  s^{\beta -2}(0))$. 

Before formally defining call graph traces, let us consider a more complex example which has a tree shaped trace structure. Consider the call graph $\mathcal{G} = \left\lbrace P_1,P_2 \right\rbrace$ where the flows are defined as follows:
$$P_1= \left\lbrace \begin{array}{c} \left( \lbrace (\delta, n) ,(\delta, \mathit{p}(n)), (\delta', n,n) \rbrace,S_1\right)\ ,\ \left( \lbrace(\delta, n) \rbrace, S_2\right) \end{array} \right\rbrace$$
$$P_2= \left\lbrace \begin{array}{c} \left( \lbrace (\delta', n,m) ,(\delta',n,\mathit{p}(m))\rbrace,S_1'\right)\ ,\ \left( \lbrace(\delta', n,m) \rbrace, S_2'\right) \end{array} \right\rbrace$$
where $S_1 = \lbrace \sigma \in \mathcal{S} \ \&\  \sigma(n)\Eval_{\omega} > 0\rbrace$ and $S_2 = \lbrace \sigma \in \mathcal{S} \ \&\  \sigma(n)\Eval_{\omega} = 0\rbrace$, $S_1' = \lbrace \sigma \in \mathcal{S} \ \&\  \sigma(m)\Eval_{\omega} > 0\rbrace$ and $S_2' = \lbrace \sigma \in \mathcal{S} \ \&\  \sigma(m)\Eval_{\omega} = 0\rbrace$.  This call graph illustrates nested primitive recursion and the trace of $\left\lbrace n\leftarrow s^{\alpha}(0)\right\rbrace$ at $P_1$ in $\mathcal{G}$ is 
$$\Tree[.\mbox{$(\delta, s^{\alpha}(0))$} [.\mbox{$(\delta', s^{\alpha}(0), s^{\alpha}(0))$} [.\mbox{$(\delta', s^{\alpha}(0), s^{\alpha-1}(0))$} [.\mbox{$\vdots$} [.\mbox{$(\delta', s^{\alpha}(0), 0)$} ] ]]
          [.\mbox{$(\delta, s^{\alpha-1}(0))$} [.\mbox{$(\delta', s^{\alpha-1}(0), s^{\alpha-1}(0))$} [.\mbox{$(\delta', s^{\alpha-1}(0), s^{\alpha-2}(0))$} [.\mbox{$\vdots$} [.\mbox{$(\delta', s^{\alpha-1}(0), 0)$} ] ]]]
                [.\mbox{$\vdots$}  [.\mbox{$\ (\delta, s^{1}(0))$} [.\mbox{$(\delta', s^{1}(0), s^{1}(0))$} [.\mbox{$(\delta', s^{1}(0), 0)$} ]]
                           [.\mbox{$\ (\delta, s^{1}(0))$} ]]]]]]$$
Notice that the trace of $\left\lbrace n\leftarrow s^{\alpha}(0), m\leftarrow s^{\alpha}(0)\right\rbrace$ at $P_2$ in $\mathcal{G}$ is similar to the the traces of the call graph $G= \left\lbrace P \right\rbrace$. We now formally define call graph traces and the computation of a trace of  $\sigma$ at a flow $P$ in a call graph $\mathcal{G}$. 

\begin{definition}[$\Af$-trace]
\label{TraceDef}
An {\em $\Af$-trace} is a pairing of an $\Af$-junction with a set of $\Af$-traces built using  the following inductive definition: 
\begin{itemize}
\item if $j$ is an $\Af$-junction, then $\left[j, \emptyset \right]$ is an $\Af$-trace.
\item if $j$ is an $\Af$-junction and $\left[j_1, T_1 \right],\cdots, \left[j_m, T_m \right]$ are $\Af$-traces such that for some $S\subseteq \mathcal{S}$, $j_1 \prec^{S} j,\cdots j_m \prec^{S} j$, then $\left[j, \bigcup_{i=1}^{m} \left[j_1, T_1 \right] \right]$ is an $\Af$-trace. 
\end{itemize}
We will refer to the outermost junction in an $\Af$-trace as the {\em root of the trace}. 
\end{definition}
\begin{definition}
\label{normal}
Let $\mathcal{G}\in \mathcal{G}^{\star}(\Af)$, $P_{S^{*}}\in \mathcal{G}$, $ S\in \mathcal{S}^*$ and $\sigma \in S$. The {\em trace} of $\sigma$ at $P_{S^{*}}$ in $\mathcal{G}$, denoted by  $\mathit{T}(\mathcal{G},P_{S^{*}},\sigma)$ is 
$$  \left[ \sigma([P_{S^{*}}])\Eval_{\omega} , \bigcup_{j\in P_{S^{*}}(S)\setminus \left\lbrace [P_{S^{*}}] \right\rbrace } \mathit{T}(\mathcal{G},\mbox{flow}(j,\sigma),\mbox{subst}(j,\sigma))  \right]$$  
where $\sigma(j)\Eval_{\omega} = \theta([P^j])\Eval_{\omega}$, where $\theta = \mbox{subst}(j,\sigma)$.
\end{definition}
\begin{example}
Consider the call graph defined in Example \ref{callgraphex}. The trace of $\lbrace n\leftarrow 2\rbrace$ at $P_1$ in $\mathcal{G}$ or $T(\mathcal{G}, P_1,\lbrace n\leftarrow 2\rbrace)$ results in the following computation:
\begin{small}
\begin{align*}
 T(\mathcal{G}, P_1,\lbrace n\leftarrow 2\rbrace)  & =  \left[ (\delta ,2), T(\mathcal{G}, P_2,\lbrace n\leftarrow 2,m\leftarrow 1,k\leftarrow 2, w\leftarrow 0 \rbrace) \right]\\
  T\left( \mathcal{G}, P_2,\left\lbrace \begin{array}{c} n\leftarrow 2,m\leftarrow 1,\\ k\leftarrow 2, w\leftarrow 0\end{array}\right\rbrace \right)  & =  \left[ (\delta' ,2,1,2,0), T\left( \mathcal{G}, P_2,\left\lbrace \begin{array}{c} n\leftarrow 2,m\leftarrow 1,\\ k\leftarrow 1, w\leftarrow 1\end{array}\right\rbrace \right) \right]\\
  T\left( \mathcal{G}, P_2,\left\lbrace \begin{array}{c} n\leftarrow 2,m\leftarrow 1,\\ k\leftarrow 1, w\leftarrow 1\end{array}\right\rbrace \right)  & =  \left[ (\delta' ,2,1,1,1), T\left( \mathcal{G}, P_2,\left\lbrace \begin{array}{c} n\leftarrow 2,m\leftarrow 1,\\ k\leftarrow 0, w\leftarrow 2\end{array}\right\rbrace \right) \right]\\
    T\left( \mathcal{G}, P_2,\left\lbrace \begin{array}{c} n\leftarrow 2,m\leftarrow 1,\\ k\leftarrow 0, w\leftarrow 2\end{array}\right\rbrace \right)  & =  \left[ (\delta' ,2,1,0,2), T\left( \mathcal{G}, P_2,\left\lbrace \begin{array}{c} n\leftarrow 2,m\leftarrow 0,\\ k\leftarrow 2, w\leftarrow 2\end{array}\right\rbrace \right) \right]\\
     T\left( \mathcal{G}, P_2,\left\lbrace \begin{array}{c} n\leftarrow 2,m\leftarrow 0,\\ k\leftarrow 2, w\leftarrow 2\end{array}\right\rbrace \right)  & =  \left[ (\delta' ,2,0,2,2), T\left( \mathcal{G}, P_2,\left\lbrace \begin{array}{c} n\leftarrow 2,m\leftarrow 0,\\ k\leftarrow 1, w\leftarrow 3\end{array}\right\rbrace \right) \right]\\
     T\left( \mathcal{G}, P_2,\left\lbrace \begin{array}{c} n\leftarrow 2,m\leftarrow 0,\\ k\leftarrow 1, w\leftarrow 3\end{array}\right\rbrace \right)  & =  \left[ (\delta' ,2,0,1,3), T\left( \mathcal{G}, P_2,\left\lbrace \begin{array}{c} n\leftarrow 2,m\leftarrow 0,\\ k\leftarrow 0, w\leftarrow 4\end{array}\right\rbrace \right) \right]\\
      T\left( \mathcal{G}, P_2,\left\lbrace \begin{array}{c} n\leftarrow 2,m\leftarrow 0,\\ k\leftarrow 0, w\leftarrow 4\end{array}\right\rbrace \right)  & =  \left[ (\delta' ,2,0,0,4), \emptyset \right]
\end{align*}
\end{small}
Thus the resulting trace is as follows: $\left[ (\delta ,2),\left[ (\delta' ,2,1,2,0),\left[ (\delta' ,2,1,1,1), \right. \right. \right.$\\
$\left. \left.\left[ (\delta' ,2,1,0,2), \left. \left[ (\delta' ,2,0,2,2), \left[ (\delta' ,2,0,1,3), \left[ (\delta' ,2,0,0,4), \emptyset \right] \right] \right] \right] \right] \right] \right]$.
\end{example}
\begin{theorem}
\label{terminationCallGraph}
Let $\mathcal{G}\in \mathcal{G}^{\star}(\Af,\mathcal{S})$, $P_{S^*}\in \mathcal{G}$ and $\sigma \in \mathcal{S}$. Then $\mathit{T}(\mathcal{G},P_{S^*},\sigma)$ always produces a finite trace.
\end{theorem}
\begin{proof}
By the definition of $\mathcal{G}^{\star}(\Af,\mathcal{S})$ only a finite number of symbols may occur in the junctions $\mathit{T}(\mathcal{G},P_{S^*},\sigma)$ which we denote by $\Delta_{0}\subseteq \Delta$. Let $$A = \left\lbrace \alpha \ \vert\ \alpha \in \mathbb{N}\ \& \ \exists \delta\in \Delta_{0} ( \Af(\delta) = \alpha ) \right\rbrace,$$ $$\Delta(B) = \left\lbrace \delta \ \vert \ \delta\in  \Delta_{0} \ \& \ \Af(\delta)\in B \right\rbrace, \mbox{ and}$$
$$\Delta_{max} = \Delta\left( \left\lbrace {\displaystyle \max_{\alpha \in A}} \ \alpha  \right\rbrace\right).$$
Now for the base case consider any sub-trace $T'$ of  $\mathit{T}(\mathcal{G},P_{S^*},\sigma)$ rooted at a junction $j$ whose symbol is in $\Delta_{max}$, we show that $T'$  must be finite. By Definition~\ref{TraceDef} all junctions occurring in $T'$ below $j$ must be smaller than $j$ with respect to the network order. By Definition~\ref{networkorder}, there is a total well-ordering of the junctions occurring in $T'$ because the arity function maps all junctions occurring in $T'$ to the same value. Thus, none of the sub-traces of  $\mathit{T}(\mathcal{G},P_{S^*},\sigma)$ rooted at a junction $j$ whose symbol is in $\Delta_{max}$ can be infinite.

Now for the induction hypothesis, let us consider the set $B$ which contains the $m$ of largest values of $A$. We assume that any sub-trace $T'$ of $\mathit{T}(\mathcal{G},P_{S^*},\sigma)$ rooted at a junction $j$ whose symbol is in $\Delta(B)$ is finite. For the step case, we show that when $B$ contains the $m+1$ largest values of $A$ any sub-trace $T'$ of $\mathit{T}(\mathcal{G},P_{S^*},\sigma)$ rooted at a junction $j$ whose symbol is in $\Delta(B)$ is finite. 

If every junction occurring in $T'$ has a symbol which the arity function maps to the same value as the symbol of $j$ then, instead of the set $B$ we can consider the subset of $B$, $\left\lbrace \Af(j) \right\rbrace$. The base case handles the situation when $B$ only contains one element and thus $T'$ is finite. 

If $T'$ contains a junction $j'$ whose symbol $\delta'$ is mapped to a different value then the symbol of $j$ then we may split $B$ into two sets, namely $$B_{low} = \left\lbrace \beta \ \vert \ \beta\in B \ \& \ \beta < \Af(\delta')\right\rbrace \mbox{ and } B_{high} = \left\lbrace \beta \ \vert \ \beta\in B \ \& \ \Af(\delta')\leq \beta\right\rbrace.$$
Notice that the sub-trace $T''$ starting from $j'$ may only contain symbols from $\Delta(B_{high})$ which contains less than $m+1$ symbols and thus by the induction hypothesis $T''$ must be finite and finishes the proof of the step case. Notice that we have covered all possible sub-traces of $\mathit{T}(\mathcal{G},P_{S^*},\sigma)$ and have shown that they must be finite. This implies that $\mathit{T}(\mathcal{G},P_{S^*},\sigma)$ itself must be finite.
\end{proof}
It is not obvious which functions beyond primitive recursion are representable by finite saturated call graphs, however it is quite obvious that limiting flows such that the all points occurring in the flow are primitive recursively computable from the source would result in call graphs which evaluate to primitive recursively bounded traces. It is shown in~\cite{Schnorr} (see Theorem 7.1.5. on page 120) that a primitive recursive bound is sufficient to imply the existence of a primitive recursive description and thus call graphs with the above flow restriction  are limited to primitive recursion. To show equivalence to primitive recursion one just has to note that call graphs implement composition, and flows implement primitive recursion, the projections, and the basic functions. The difference between our formalism and primitive recursion is an increase in flexibility necessary for describing recursive refutations.
 
\section{Schematic $\RPLnull^\Psi$ Derivations}\label{SchRPL0}
To construct schematic $\RPLnull^\Psi$ derivations we need a countably infinite set of {\em proof symbols} which are used to label the individual proofs of a  {\em proof schema}. A particular proof schema uses a finite set of proof symbols  $\Delta^*\subset \Delta$. Also, we need a concept of {\em proof labels} which are used to normalize $\RPLnull^\Psi$ derivations. 
\begin{definition}[proof label]
Let $\delta\in \Delta$ be a proof symbol and $\vartheta$ a parameter substitution. We refer to the object $(\delta,\vartheta)$ as a {\em proof label}.  
\end{definition}
\begin{example}
\label{labelex}
Let us consider the $\RPLnull^\Psi$ derivations of Example~\ref{TheMainExample}. The labels
$$(\delta_3, \left\lbrace r\leftarrow p(r),q \leftarrow s(q)\right\rbrace ) $$
$$(\delta_1, \left\lbrace w\leftarrow p(w),k\leftarrow s(k),r\leftarrow s(m),q\leftarrow 0\right\rbrace )$$
$$(\delta_2, \left\lbrace  r\leftarrow p(r),q \leftarrow s(q) \right\rbrace ) $$
are used to label sequents of derivation $\rho_1(\delta_2,\vec{X}, n,m,w,k,r,q)$
\end{example}
We will need to locate particular sequents within a given $\RPLnull^{\Psi}$ derivation. Thus we assume that each sequent in an $\RPLnull^{\Psi}$ derivation $\pi$ is given a unique position from a set of positions $\Lambda_{\pi}$. By $|\pi|_{\lambda}$ we denote the sequent $S$ occurring in $\pi$ at position $\lambda\in \Lambda_{\pi}$. Furthermore, let $\pi^{*}$ be an $\RPLnull^{\Psi}$ derivation whose end-sequent is $|\pi|_{\lambda}$, then  by $\pi[\pi^{*}]_{\lambda}$ we denote the $\RPLnull^{\Psi}$ derivation where the derivation starting at $\lambda$ in $\pi$ is replaced by $\pi^{*}$. Also, by $\mathit{leaf}(\pi)$ for an $\RPLnull^{\Psi}$ derivation $\pi$, we denote the set of labels  associated with the leaves of $\pi$. 
\begin{definition}[labelled sequents and derivations]
Let $S$ be a sequent and  $(\delta,\vartheta)$ a proof label, then $(\delta,\vartheta)\colon S$ is a {\em  sequent}.
A {\em labeled $\RPLnull^\Psi$ derivation} $\pi$ is an $\RPLnull^\Psi$ derivation where $|\pi|_{\lambda}$, for $\lambda\in \mathit{leaf}(\pi)$, may be a labeled sequent. By $\mathit{LAB}(\pi)$, where $\pi$ is a labeled $\RPLnull^\Psi$ derivation, we denote the set of all positions $\lambda$ such that $|\pi|_{\lambda}$ is a labeled sequent. By $\mathit{Ax}(\pi)$ we denote the set of positions $\mathit{leaf}(\pi)\setminus \mathit{LAB}(\pi)$.
\end{definition}
 \begin{example}
 The following labeled sequents use the labels discussed in Example~\ref{labelex}. 
 $$(\delta_3, \left\lbrace r\leftarrow p(r),q \leftarrow s(q)\right\rbrace )\colon \vdash  \hat{F_4}(\mathbf{X},k,s(q)) $$
$$(\delta_1, \left\lbrace w\leftarrow p(w),k\leftarrow s(k),r\leftarrow s(m),q\leftarrow 0\right\rbrace ) \colon $$ $$\vdash \leqpred{\hat{S}(Y(p(w),s(k)),s(q))}{s(k)}$$
$$(\delta_2, \left\lbrace r\leftarrow p(r),q \leftarrow s(q)\right\rbrace )\colon\vdash  \leqpred{Y(w,k)}{k}, \hat{F_5}(\mathbf{X},Y,k,s(q)) $$
 \end{example}
 sequents and labeled derivations can be used to link derivations together. We can group our $\RPLnull^\Psi$ derivations together based on a partitioning of the set of all parameter assignments and link these groups together by the proof labels.

\begin{definition}[$\RPLnull^\Psi$ schema]\label{def.proofschema}
Let $\Psi$ be as in Definition~\ref{def.formschem}, $\Delta^* \subset \Delta$ a finite set of proof symbols, $\delta_0 \in \Delta^*$ the {\em main symbol}, $ \mathcal{N}_0\subset \mathcal{N}$, and for each $\delta\in \Delta^*$ we associated a partitioning $\mathcal{S}_{\delta}$ of $\mathcal{S}$ which partitions $\mathcal{S}$ into  $k_{\delta}$ mutually disjoint sets, for $k\geq 0$. 
To every $\delta \in \Delta^*$ we assign 
$$\Dcal(\delta)\colon ((\rho_1(\delta,\vec{X}, \vec{n}),S_{1})\oplus \ldots \oplus(\rho_{k_{\delta}}(\delta,\vec{X},\vec{n} ),S_{k_{\delta}}), S_\delta(\vec{X},\vec{n})),$$
where
\begin{itemize}
	\item for $1\leq i \leq k_{\delta}$, $S_{i}\in \mathcal{S}_{\delta}$,
	\item $\vec{X}$ is a tuple of global variables
	\item $\vec{n}$ is a tuple of parameters contained in $\mathcal{N}_0$,
	 \item for each $\delta \in \Delta^*$ and $1\leq j \leq k_{\delta}$, $\rho_j(\delta,\vec{X},\vec{n})$ is a labelled $\RPLnull^\Psi$ derivation of $S_\delta(\vec{X},\vec{n})$ where 
	\begin{itemize}
	 \item for each $\lambda\in \mathit{Ax}(\rho_j(\delta,\vec{X},\vec{n}))$,  $|\rho_j(\delta,\vec{X},\vec{n})|_{\lambda} = \ \seq \Phat(\vec{X},\vec{r},\vec{s})$, i.e. the top symbol of the recursive definition,  and
	 \item for each $\lambda\in \mathit{Lab}(\rho_j(\delta,\vec{X},\vec{n}))$, $|\rho_j(\delta,\vec{X},\vec{n})|_{\lambda} = (\delta',\vartheta)\colon S_{\delta'}(\vec{X},\vec{s})$ such that $\delta' \in \Delta^*$ and $S_{\delta'}(\vec{X},\vec{s})$ = $S_{\delta'}(\vec{X},\vec{h})\vartheta$ where
	 $$\Dcal(\delta')\colon ((\rho_1(\delta',\vec{X}, \vec{h}),S_{1})\oplus \ldots \oplus(\rho_{k_{\delta'}}(\delta',\vec{X},\vec{h} ),S_{k_{\delta'}}), S_{\delta'}( \vec{X},\vec{h})).$$
We will refer to the parameter substitution $\vartheta$ as $\mathit{subst}(\Dcal, \delta,j,\lambda)$. 
\end{itemize}
\end{itemize} 
An $\RPLnull^\Psi$ schema $\Dcal$ is defined as the set 
$$\Union\{\Dcal(\delta) \mid \delta \in \Delta^*\}.$$
Furthermore, an $\RPLnull^\Psi$ schema is referred to as a {\em schematic $\RPLnull^\Psi$ derivation} if there exists a top symbol which we denote by $\delta_0$.  If the proofs of $\Dcal(\delta_0)$  end with $\seq$ we refer to $\Dcal$ as an {\em $\RPLnull^\Psi$ refutation schema}. 
\end{definition}
We consider only normal and regular $\RPLnull^\Psi$ derivations in the construction of an $\RPLnull^\Psi$ schema, and thus only consider normal and regular $\RPLnull^\Psi$ schema.
\begin{definition} 
Let $\Dcal$ be a $\RPLnull^\Psi$ schema. If all derivations $\rho_j(\delta,\vec{X},\vec{t})$ occurring for any proof symbol $\delta$ in $\Dcal$ are normal and regular $\RPLnull^\Psi$ derivations, we say that $\Dcal$ is a normal and regular $\RPLnull^\Psi$  schema.
\end{definition}
\begin{example}\label{ex.2parametersCont}
Below is the complete refutation schema for Example~\ref{ex.2parameters} using two symbols $\delta_0$ and $\delta_1$ where 
$$\Dcal(\delta_0) \equiv ((\rho_0(\delta_0,X,Y, n,m),S_1) \oplus (\rho_1(\delta_0,X,Y, n,m),S_2), \vdash),$$
$$\Dcal(\delta_1) \equiv ((\rho_0(\delta_1,X,Y,k, n,m),S_1) \oplus (\rho_1(\delta_1,X,Y,k, n,m),S_2), \vdash \hat{Q}(X,Y,n,m)),$$
$S_1 = \left\lbrace \sigma\ \middle\vert\ n\sigma = 0\right\rbrace $, and $S_2 = \left\lbrace \sigma\ \middle\vert\ n\sigma \not = 0\right\rbrace $. By $\rho_0(\delta_0,X,Y, n,m)$ we denote the following derivation:
\small
\[
\infer[res(\sigma_1)]{\begin{array}{c}\seq \end{array} }{
\infer[\wedge_{r_1}]{\seq P(\hat{f}(Y(0),m),Y(1))}{
\infer[B\hat{Q} r]{\seq P(\hat{f}(Y(0),m),Y(1))\wedge \hat{P}(X,0)}{
\seq \hat{Q}(X,Y,0,m)
}}
&
\infer[\neg:l]{ P(X(0),\hat{f}(a,0))\seq }{
\infer[B\hat{P} r]{\seq \neg P(X(0),\hat{f}(a,0))}{
\infer[\wedge_{r_2}]{\seq \hat{P}(X,0)}{
\infer[B\hat{Q} r]{\seq P(\hat{f}(Y(0),m),Y(1))\wedge \hat{P}(X,0)}{
    \seq \hat{Q}(X,Y,0,m)
 } }}}
}
\]
\normalsize
where $\sigma_1 = \{X(0) \leftarrow \fhat(Y(0),m), Y(1) \leftarrow \fhat(a,0)\}$ and by $\rho_1(\delta_0,X,Y, n,m)$ we denote the following derivation:
\small

\[
\infer[res \sigma_1]{\begin{array}{c}\seq \end{array} }{
\infer[\wedge_{r_1}]{\seq P(\hat{f}(Y(0),m),Y(1))}{
\infer[B\hat{Q} r]{\seq P(\hat{f}(Y(0),m),Y(1))\wedge \hat{P}(X,0)}{
\deduce{\seq \hat{Q}(X,Y,0,m)}{ (\delta_1,\left\lbrace k\leftarrow 0,n\leftarrow p(n)\right\rbrace ) }
}}
&
\infer[\neg:l]{ P(X(0),\hat{f}(a,0))\seq }{
\infer[B\hat{P} r]{\seq \neg P(X(0),\hat{f}(a,0))}{
\infer[\wedge_{r_2}]{\seq \hat{P}(X,0)}{
\infer[B\hat{Q} r]{\seq P(\hat{f}(Y(0),m),Y(1))\wedge \hat{P}(X,0)}{
    \deduce{\seq \hat{Q}(X,Y,0,m)}{(\delta_1,\left\lbrace k\leftarrow 0,n\leftarrow p(n)\right\rbrace ) }
 } }}}
}
\]
\normalsize
Concerning the proof symbol $\delta_1$, by $\rho_0(\delta_1,X,Y,k, n,m)$ we denote the following derivation: 
\begin{small}
\[
\infer[\land_{r_1}]{\begin{array}{c}\seq P(\hat{f}(Y(0),m),Y(1))\\ (2) \end{array}}{
\infer[S\hat{Q}r]{\seq P(\hat{f}(Y(0),m),Y(1))\wedge \hat{P}(X,s(k))}{
 \deduce{\seq\hat{Q}(X,Y,s(k),m)}{}
}}
\]

\[
\infer[res(\sigma_2)]{\begin{array}{c}\seq \hat{P}(X,k)\\ (1) \end{array}}{
\deduce{\seq P(\hat{f}(Y(0),m),Y(1))}{(2)}
&
\infer[\neg_r]{ P(X(k),\hat{f}(a,k)) \seq  \hat{P}(X,k)}{
\infer[\vee_r]{  \seq  \neg P(X(k),\hat{f}(a,k)),\hat{P}(X,k)}{
\infer[S\hat{P}r]{  \seq  \neg P(X(k),\hat{f}(a,k))\vee \hat{P}(X,k)}{
\infer[\wedge_{r_2}]{  \seq \hat{P}(X,s(k))}{
\infer[S\hat{Q}r]{  \seq P(\hat{f}(Y(0),m),Y(1))\wedge \hat{P}(X,s(k))}{
 \deduce{\seq\hat{Q}(X,Y,s(k),m)}{}
}
}
}}}}
\]
\end{small}
where $\sigma_2 = \{X(k) \leftarrow \fhat(Y(0),m), Y(1) \leftarrow \fhat(a,k)\}$.
\begin{small}

\[
\infer[S\hat{Q}r^+]{\begin{array}{c}\vdash \hat{Q}(X,Y,k,m) \end{array} }{
\infer[res\{X \leftarrow Z\}\footnote{Note that this renaming is no s-unifier but a (higher-order) renaming to ensure regularity.}]{P(\hat{f}(Y(0),m),Y(1))\wedge \hat{P}(X,k)}{
\deduce{\seq \hat{P}(Z,k)}{(1)}
&
\infer[res]{\hat{P}(X,k) \seq P(\hat{f}(Y(0),m),Y(1))\wedge \hat{P}(X,k)}{
\deduce{\seq P(\hat{f}(Y(0),m),Y(1))}{(2)}
&
\wedge\mathrm{\mbox{-}Axiom}}
}}
\]
\end{small}
and  by $\rho_1(\delta_1,X,Y,k,n,m)$ we denote the following derivation:
\begin{small}
\[
\infer[\land_{r_1}]{\begin{array}{c}\seq P(\hat{f}(Y(0),m),Y(1))\\ (2) \end{array}}{
\infer[S\hat{Q}r]{\seq P(\hat{f}(Y(0),m),Y(1))\wedge \hat{P}(X,s(k))}{
 \deduce{\seq\hat{Q}(X,Y,s(k),m)}{(\delta_1,\left\lbrace k\leftarrow s(k),n\leftarrow p(n)\right\rbrace ) }
}}
\]

\[
\infer[res(\sigma_2)]{\begin{array}{c}\seq \hat{P}(X,k)\\ (1) \end{array}}{
\deduce{\seq P(\hat{f}(Y(0),m),Y(1))}{(2)}
&
\infer[\neg_r]{ P(X(k),\hat{f}(a,k)) \seq  \hat{P}(X,k)}{
\infer[\vee_r]{  \seq  \neg P(X(k),\hat{f}(a,k)),\hat{P}(X,k)}{
\infer[S\hat{P}r]{  \seq  \neg P(X(k),\hat{f}(a,k))\vee \hat{P}(X,k)}{
\infer[\wedge_{r_2}]{  \seq \hat{P}(X,s(k))}{
\infer[S\hat{Q}r]{  \seq P(\hat{f}(Y(0),m),Y(1))\wedge \hat{P}(X,s(k))}{
 \deduce{\seq\hat{Q}(X,Y,s(k),m)}{(\delta_1,\left\lbrace k\leftarrow s(k),n\leftarrow p(n)\right\rbrace ) }
}
}
}}}}
\]
\end{small}
where $\sigma_2 = \{X(k) \leftarrow \fhat(Y(0),m), Y(1) \leftarrow \fhat(a,k)\}$.
\begin{small}

\[
\infer[S\hat{Q}r^+]{\begin{array}{c}\vdash \hat{Q}(X,Y,k,m) \end{array} }{
\infer[res\{Z \leftarrow X\}]{P(\hat{f}(Y(0),m),Y(1))\wedge \hat{P}(X,k)}{
\deduce{\seq \hat{P}(X,k)}{(1)}
&
\infer[res]{\hat{P}(Z,k) \seq P(\hat{f}(Y(0),m),Y(1))\wedge \hat{P}(X,k)}{
\deduce{\seq P(\hat{f}(Y(0),m),Y(1))}{(2)}
&
\wedge\mathrm{\mbox{-}Axiom}}
}}
\]
\end{small}
\end{example}
Note that the definition of a $\RPLnull^\Psi$ refutation schema is very general and without further conditions, a given schema may not evaluate to an $\RPLnull$ deduction under (all) parameter assignments. Some proof symbols may never be reached from $\delta_0$  or even more disconcerting, a call sequence may go on forever. Let us now define the evaluation of $\RPLnull^\Psi$ refutation schema to make this point clearer. 

\begin{definition}[proof join operator]
Let $\varphi_0$ and $\varphi_1$ be labelled $\RPLnull^\Psi$ derivations such that the end sequent of $\varphi_1$ is $S(\vec{n})$ were $\vec{n}$ is a tuple of parameters. We define the join of $\varphi_0$ and $\varphi_1$, denoted by $\varphi_0 \bowtie \varphi_1$, as 
$$\varphi_0\bowtie \varphi_1 = \varphi_0[\varphi_{1}
\vartheta_{\lambda_1}]_{\lambda_1}[\varphi_{1}
\vartheta_{\lambda_2}]_{\lambda_2}\cdots [\varphi_{1}
\vartheta_{\lambda_k}]_{\lambda_k}$$
where $\left\lbrace \lambda_1, \cdots ,\lambda_k\right\rbrace =$ $\left\lbrace \lambda \ \vert \ \lambda\in \mathit{LAB}(\varphi_0)\ \& \ \vert\varphi_0\vert_{\lambda} = (\delta,\vartheta_{\lambda})\colon S_\lambda \ \& \ S(\vec{n})\vartheta_{\lambda} = S_{\lambda} \right\rbrace $.
\end{definition}

\begin{example}
Consider the proofs $\rho_1(\delta_0,X,Y,1,2)$ and $\rho_0(\delta_1,X,Y,0,0,2)$ of Example~\ref{ex.2parametersCont}. Note that $$\rho_1(\delta_0,X,Y,1,2) \bowtie \rho_0(\delta_1,X,Y,0,0,2) \not = \rho_1(\delta_0,X,Y,1,2)$$
because there exists two positions of $\rho_1(\delta_0,X,Y,1,2)$ in $\mathit{LAB}(\rho_1(\delta_0,X,Y,1,2))$ whose sequent matches the end-sequent of $\rho_0(\delta_1,X,Y,0,0,2)$ after substitution. However if we instead consider  the proof $\rho_0(\delta_1,X,Y,1,0,2)$ we get 
$$\rho_1(\delta_0,X,Y,1,2) \bowtie \rho_0(\delta_1,X,Y,1,0,2)  = \rho_1(\delta_0,X,Y,1,2).$$
Note that this example illustrates a trivial join.

Joins can be nested as follows: 
$$\left( \rho_1(\delta_0,X,Y,2,2) \bowtie \rho_1(\delta_1,X,Y,0,1,2) \right) \bowtie \rho_0(\delta_1,X,Y,1,0,2).$$
Nesting join applications is essential for defining evaluation of $\RPLnull^\Psi$ schemata.
\end{example}
\begin{definition}[evaluation of $\RPLnull^\Psi$ schemata]
\label{evaldef}
Let $\Dcal$ be a normal and regular $\RPLnull^\Psi$ schema,  $\sigma\in \mathcal{S}$, and $\delta$ a proof symbol of $\Dcal$. Then the evaluation of $\Dcal$ at $\delta$ by $\sigma$, denoted by $N(\Dcal,\delta,\sigma)$, is as follows: 
Let $\Dcal(\delta)$ be $((\rho_1(\delta,\vec{X}, \vec{n}),S_{1})\oplus \ldots \oplus(\rho_{k_{\delta}}(\delta,\vec{X},\vec{n} ),S_{k_{\delta}}), S(\delta, \vec{X},\vec{n})),$
and $\sigma \in S_{i}$, for $1\leq i\leq k_{\delta}$. Then one of the following cases must hold:  
\begin{itemize}
\item If $\mathit{LAB}(\rho_i(\delta,\vec{X}, \vec{n})) = \emptyset$, then $N(\Dcal,\delta,\sigma) = \sigma(\rho_i(\delta,\vec{X}, \vec{n}))\Eval_{\omega}$
\item  $\mathit{LAB}(\rho_i(\delta,\vec{X}, \vec{n})) = \left\lbrace \lambda_1,\cdots, \lambda_l\right\rbrace $, then 
$$N(\Dcal,\delta,\sigma) = \left(\cdots \left( \sigma(\rho_i(\delta,\vec{X}, \vec{n}))\Eval_{\omega} \bowtie N(\Dcal, \delta_1, \sigma_1)\right)  \cdots \right)  \bowtie N(\Dcal, \delta_l, \sigma_l),$$
where $\sigma_i= \sigma[\mathit{subst}(\Dcal,\delta,i,\lambda)]\downarrow_{\omega}$.
\end{itemize}
\end{definition}
Note that Definition~\ref{evaldef} does not ensure termination of the the proof evaluation procedure, nor does it guaranteed an $\RPLnull$ derivation as a result. To ensure termination one has to provide a call schematics, i.e. the call graphs constructed in Section~\ref{Scaffolding}, which guarantee termination.  This so called call semantics is a mapping from the junctions of the flows occurring in a call graph to the labels occurring in the derivations of an $\RPLnull^\Psi$ schema.

\begin{definition}
Let $\Dcal(\delta)$ be as defined in Definition \ref{def.proofschema} and $C \in \regFlow{\Af}$. We say that $C \models \Dcal(\delta)$ if for every $\sigma \in \mathcal{S}$, there exists a unique $i\in\lbrace1,\cdots , k_{\delta}\rbrace$ s.t.  $(\rho_{i}(\delta,\vec{X},\vec{n}),S_{i})\in \Dcal(\delta)$, $\sigma \in S_{i}$ and the following holds: 
\begin{itemize}
\item $[C]= (\delta, \vec{n})$
\item For each $\lambda\in \mathit{LAB}(\rho_{i}(\delta,\vec{X},\vec{n}))$ there exists a junction $(\delta',\vec{t})\in C(\sigma)$, $(\delta',\vec{t})\neq \sigma([C])\vert_{\omega}$, s.t. $\mathit{subst}((\delta',\vec{t}),\sigma) = \sigma[\vartheta]\vert_{\omega}$, where $\vert(\sigma(\rho_{i}(\delta,\vec{X},\vec{n}))\vert_{\omega})\vert_{\lambda} = (\delta',\vartheta):S'$.
\end{itemize}
\end{definition}
Now we can connect call graphs to $\RPLnull^\Psi$ schemata.
\begin{definition}
Let $\Dcal$ be a $\RPLnull^\Psi$ refutation schema defined over the symbols $\Delta^*\subset \Delta$ and $G \in \mathcal{G}^{\star}(\Af)$. We say that $G \models \Dcal$ if for each $C\in G$ there exists a unique $\delta\in \Delta^*$ s.t. $C\models \Dcal(\delta)$ and $\vert \Delta^*\vert = \vert G \vert$. $\Dcal$ is said to be {\em well-formed}.
\end{definition} 
Well-formed $\RPLnull^\Psi$ schemata can be normalized to $\RPLnull^\Psi$ derivations. 

\begin{theorem}
Let $\Dcal$ be a well-formed $\RPLnull^\Psi$  schema, Then for any $\delta$ of $\Dcal$,  $N(\Dcal,\delta,\sigma)$ is a finite $\RPLnull$ derivation.
\end{theorem}
\begin{proof}
Follows directly from Theorem~\ref{terminationCallGraph}.
\end{proof}

\begin{example}
The associated call graph for the refutation schema of Example \ref{ex.2parametersCont} is $\mathcal{G} = \lbrace C_1, C_2\rbrace$, where
\begin{align*}
C_1 =& \left\lbrace \begin{array}{c} 
\left( \left\lbrace   (\delta_0, n,m) ,(\delta_1,0,\mathit{p}(n), m) \right\rbrace , S_1\right)   \\ 
\left( \lbrace (\delta_0, n,m)  \rbrace , S_2 \right)  \\ 
  \end{array} \right\rbrace\\
  C_2 =& \left\lbrace \begin{array}{c} 
\left( \left\lbrace   (\delta_1,k, n,m) ,(\delta_1, s(k),\mathit{p}(n),m) \right\rbrace , S_1\right)  \\ 
\left( \lbrace (\delta_1,k, n,m)  \rbrace , S_2 \right)  \\ 
  \end{array} \right\rbrace\\
  \end{align*}
and $S_1 = \lbrace \sigma \in \mathcal{S} \ \&\  \sigma(n)\Eval_{\omega} > 0\rbrace$ and $S_2 = \lbrace \sigma \in \mathcal{S} \ \&\  \sigma(n)\Eval_{\omega} = 0\rbrace$. Furthermore, 
$C_1 \models \Dcal(\delta_0)$,  $C_2 \models \Dcal(\delta_1)$,  and thus $\mathcal{G} \models  \Dcal(\delta_0) \cup \Dcal(\delta_1)$.
\end{example}
For normal, regular $\RPLnull^\Psi$ schemata we can define a unification schema albeit with additional structure. While  $\RPLnull^\Psi$ derivations can easily be amended with the addition of labels which aid evaluation, this cannot be easily done to s-unifiers. Instead we need to add {\em placeholder} unifiers which are replaced by s-unifiers during evaluation.   unification schemata do not easily allow for the structure necessary for aiding evaluation and composition of recursively defined unifiers. Thus, 

\begin{definition}[unification schema]
Let $\Dcal = \bigcup_{i=1}^{k} \Dcal(\delta_i)$ be a normal, regular $\RPLnull^\Psi$ schema. We define the unification schema of $\Dcal$, $\Theta(\Dcal) = \bigcup_{i=1}^{k} \Theta(\Dcal(\delta_i))$, as follows: For each $1\leq i \leq k$ let $\Dcal(\delta_i)$ be
$$((\rho_1(\delta_i,\vec{X}, \vec{n}),S_{1})\oplus \ldots \oplus(\rho_{m}(\delta_i,\vec{X},\vec{n} ),S_{m}), S_{\delta_i}(\vec{X},\vec{n})).$$
For each  $\rho_j(\delta_i,\vec{X},\vec{n})$, where $1\leq j\leq m$, we construct a {\em unification triple}  $\left(\theta_j, \mathbf{L}_j, S_j\right)$, where
\begin{itemize}
\item $\theta_j$ is  the global s-unifier of  $\rho_j(\delta_i,\vec{X},\vec{n})$,
\item $\mathbf{L}_j = \left\lbrace  (\delta',\vartheta)\ \middle\vert \ \lambda\in \mathit{LAB}(\rho_j(\delta_i,\vec{X},\vec{n})) \ \& \ \vert\rho_j(\delta_i,\vec{X},\vec{n})\vert_{\lambda} = (\delta',\vartheta)\colon S \right\rbrace,$ i.e. all labels occurring in $\rho_j(\delta_i,\vec{X},\vec{n})$,
\item $S_j$ is the same partition used in $\Dcal(\delta_i)$.
\end{itemize} 
We define  $\Theta(\Dcal(\delta_i)) = (\left(\theta_1, \mathbf{L}_1, S_1\right)\oplus \ldots \oplus \left(\theta_m, \mathbf{L}_m, S_m\right) )$, the  collection of unification triples constructed from $\Dcal(\delta_i)$.
\end{definition} 

\begin{definition}[evaluation of unification schemata]
\label{evaldef}
Let $\Theta(\Dcal)$ unification schema,  $\sigma\in \mathcal{S}$, and $\delta$ a proof symbol of $\Dcal$. Then the evaluation of $\Theta(\Dcal)$ at $\delta$ by $\sigma$, denoted by $N(\Theta(\Dcal),\delta,\sigma)$, is as follows: 
Let  $\Theta(\Dcal(\delta))$ be $(\left(\theta_1, \mathbf{L}_1, S_1\right)\oplus \ldots \oplus \left(\theta_m, \mathbf{L}_m, S_m\right) )$
and $\sigma \in S_{i}$, for $1\leq i\leq m$. Then one of the following cases must hold:
\begin{itemize}
\item If $\mathbf{L}_i = \emptyset$, then $N(\Theta(\Dcal),\delta,\sigma) = \sigma(\theta_i)\vert_{\omega}$
\item  $\mathbf{L}_i = \left\lbrace (\delta_1,\nu_1),\cdots, (\delta_l,\nu_l)\right\rbrace$, then 
$$N(\Theta(\Dcal),\delta,\sigma) = \left( \bigcup_{r=1}^{l} N(\Theta(\Dcal), \delta_r, \sigma_r)\right)  \sigma(\theta_i)\vert_{\omega},$$
where $\sigma_r= \sigma[\nu_r]\downarrow_{\omega}$. 
\end{itemize}
\end{definition}

\begin{example}
\label{TheMainExample}
Below is the complete refutation schema for the schematic formula provided in Example~\ref{1smaform}. Note that we abbreviate $X_1,\cdots,X_5$ by $\mathbf{X}$. Furthermore the definition of $\Fhat_5$ requires 5 global variables with two arguments and one global variable $Y$ with a single argument. However, in the derivations below we need $Y$ to be two place. We allow replacement of the variable $Y$ by a two place variable. This is essentially a variable renaming. Let
\begin{align*}
\Dcal(\delta_0)\colon  &((\rho_1(\delta_0,\vec{X},Y, n,m),S_{1})\oplus (\rho_2(\delta_0,\vec{X},Y, n,m),S_{2})\oplus \\
&\ (\rho_3(\delta_0,\vec{X},Y, n,m),S_{3}) \oplus(\rho_{4}(\delta_0,\vec{X},Y,n,m),S_{4}), \vdash)
\end{align*}
where 
\begin{align*}
  S_1 \equiv & \left\lbrace \sigma \middle\vert \sigma \in \mathcal{S}\ ,\ n\sigma>0 \ \&\  m\sigma>0\right\rbrace\\
  S_2 \equiv & \left\lbrace \sigma \middle\vert \sigma \in \mathcal{S}\ ,\ n\sigma=0 \ \&\  m\sigma>0\right\rbrace\\
  S_3 \equiv & \left\lbrace \sigma \middle\vert \sigma \in \mathcal{S}\ ,\ n\sigma>0 \ \&\  m\sigma=0\right\rbrace\\
 S_4 \equiv & \left\lbrace \sigma \middle\vert \sigma \in \mathcal{S}\ ,\ n\sigma=0 \ \&\  m\sigma=0\right\rbrace
\end{align*}
The refutation provided in Example~\ref{basecaseProof} is $\rho_{4}(\delta_0,\vec{X},Y,n,m)$. The other three derivations are as follows:
\begin{itemize}
\item[1)] $\rho_1(\delta_0,\vec{X},Y, n,m)$ denotes: 
\begin{small}
\begin{prooftree}
\AxiomC{$\left( \delta_1,\left\lbrace \begin{array}{cc}
 w\leftarrow n  & k\leftarrow 0  \\ r\leftarrow s(m) &
q\leftarrow 0 \end{array}\right\rbrace \right)$}
\noLine
\UnaryInfC{$ \vdash \leqpred{Y(n,0)}{0}$}
\AxiomC{$\left( \delta_5,\left\lbrace \begin{array}{cc}
 w\leftarrow n  & k\leftarrow 0  \\ r\leftarrow 0 &
q\leftarrow m \end{array}\right\rbrace \right)$}
\noLine
\UnaryInfC{$\vdash \hat{F_3}(\mathbf{X},0,m)$}
\RightLabel{$B\hat{F_3}r$}
\UnaryInfC{$ \vdash \hat{F_5}(\mathbf{X},0,m)\wedge \nleqpred{a}{0}$}
\RightLabel{$\wedge:r_2$}
\UnaryInfC{$  \vdash \nleqpred{a}{0}$}
\RightLabel{$\neg:r$}
\UnaryInfC{$  \leqpred{a}{0} \vdash$}
\RightLabel{Res$\left(\mu \right) $}
\BinaryInfC{$\vdash$}
\noLine
\end{prooftree}
\end{small}
where   $\mu = \left\lbrace Y(n,0)\leftarrow a \right\rbrace$. 
\vspace{1em}
\item[2)] $\rho_2(\delta_0,\vec{X},Y, n,m)$ denotes: 
\begin{small}
\begin{prooftree}
\AxiomC{$\left( \delta_1,\left\lbrace \begin{array}{cc}
n\leftarrow 0 & w\leftarrow 0 \\ k\leftarrow 0 & r\leftarrow s(m) \\
q\leftarrow 0 \end{array}\right\rbrace \right) $}
\noLine
\UnaryInfC{$ \vdash \leqpred{Y(0,0)}{0}$}
\AxiomC{$  \vdash   \hat{F_1}(\mathbf{X},0,m)$}
\RightLabel{$S\hat{F_1}r$}
\UnaryInfC{$  \vdash   \hat{F_2}(\mathbf{X},0,m) \wedge \hat{F_3}(\mathbf{X},0,m)$}
\RightLabel{$\wedge:r$}
\UnaryInfC{$\vdash \hat{F_3}(\mathbf{X},0,m)$}
\RightLabel{$B\hat{F_3}r$}
\UnaryInfC{$ \vdash \hat{F_5}(,\mathbf{X},Y(0),0,m)\wedge \nleqpred{a}{0}$}
\RightLabel{$\wedge:r_2$}
\UnaryInfC{$  \vdash \nleqpred{a}{0}$}
\RightLabel{$\neg:r$}
\UnaryInfC{$  \leqpred{a}{0} \vdash$}
\RightLabel{Res$\left(\mu \right) $}
\BinaryInfC{$\vdash$}
\noLine
\end{prooftree}
\end{small}
where $\mu = \left\lbrace Y(0,0)\leftarrow a \right\rbrace$. 
\vspace{1em}

\item[3)] $(\rho_3(\delta_0,\vec{X},Y, n,m),S_{3})$ where : 
\begin{small}
\begin{prooftree}
\AxiomC{$\left( \delta_6,\left\lbrace \begin{array}{cc}
m\leftarrow 0 & w\leftarrow n \\ k\leftarrow 0 & r\leftarrow s(0) \\
q\leftarrow 0 \end{array}\right\rbrace \right) $}
\noLine
\UnaryInfC{$ \vdash \leqpred{Y(n,0)}{0}$}
\AxiomC{$\left( \delta_5,\left\lbrace \begin{array}{cc}
m\leftarrow 0 & w\leftarrow n \\ k\leftarrow 0 & r\leftarrow 0 \\
q\leftarrow 0 \end{array}\right\rbrace \right) $}
\noLine
\UnaryInfC{$  \vdash  \hat{F_3}(\mathbf{X},0,0)$}
\RightLabel{$B\hat{F_3}r$}
\UnaryInfC{$ \vdash \hat{F_5}(\mathbf{X},0,0)\wedge \nleqpred{a}{0}$}
\RightLabel{$\wedge:r_2$}
\UnaryInfC{$  \vdash \nleqpred{a}{0}$}
\RightLabel{$\neg:r$}
\UnaryInfC{$  \leqpred{a}{0} \vdash$}
\RightLabel{Res$\left(\mu\right) $}
\BinaryInfC{$\vdash$}
\noLine
\end{prooftree}
\end{small}
where   $\mu = \left\lbrace Y(n,0)\leftarrow a \right\rbrace$. 
\end{itemize} 
Now let us consider 
\begin{align*}
\Dcal(\delta_1)\colon & ((\rho_1(\delta_1,\vec{X},Y, n,m,w,k,r,0),S_{5})\oplus \\ & \ (\rho_2(\delta_1,\vec{X},Y, n,m,w,k,r,0),S_{6}) , \vdash  \leqpred{\mathbf{\alpha}}{k})
\end{align*}
where 
\begin{align*}
S_5 \equiv & \left\lbrace \sigma \middle\vert \sigma \in \mathcal{S}\ ,\ w\sigma>0 \right\rbrace\\
S_6 \equiv & \left\lbrace \sigma \middle\vert \sigma \in \mathcal{S}\ ,\ w\sigma=0 \right\rbrace
\end{align*}
\begin{itemize}
\item[1)] $\rho_1(\delta_1,\vec{X},Y,n,m,w,k,r,q)$  denotes: 
\begin{small}
\begin{prooftree}
\AxiomC{$\left( \delta_3,\left\lbrace \begin{array}{cc}
 r\leftarrow p(r) & q\leftarrow 0 \end{array}\right\rbrace \right)$}
\noLine
\UnaryInfC{$  \vdash  \hat{F_4}(\mathbf{X},k,0)$}
\RightLabel{$B\hat{F_4}r$}
\UnaryInfC{$  \vdash \nleqpred{(X_2(k,0)}{s(k)}\vee \leqpred{X_2(k,0)}{k}\vee  \eqpred{X_2(k,0)}{k} $}
\RightLabel{$\vee:r$}
\UnaryInfC{$  \vdash \nleqpred{X_2(k,0)}{s(k)} ,  \leqpred{X_2(k,0)}{k}\vee  \eqpred{X_2(k,0)}{k} $}
\RightLabel{$\vee:r$}
\UnaryInfC{$  \vdash \nleqpred{X_2(k,0)}{s(k)} ,  \leqpred{X_2(k,0)}{k},  \eqpred{X_2(k,0)}{k} $}
\RightLabel{$\neg:r$}
\UnaryInfC{$   \leqpred{X_2(k,0)}{s(k)} \vdash   \leqpred{X_2(k,0)}{k},  \eqpred{X_2(k,0)}{k} $}
\noLine
\UnaryInfC{$(2)$}
\end{prooftree}
\end{small}
\begin{small}
\begin{prooftree}
\AxiomC{$\left( \delta_1,\left\lbrace \begin{array}{cc}
 w\leftarrow p(w) & k\leftarrow s(k) \\ \multicolumn{2}{c}{q\leftarrow 0} \end{array}\right\rbrace \right) $}
\noLine
\UnaryInfC{$ \vdash \leqpred{Y(p(w),s(k))}{s(k)}$}
\AxiomC{$(2)$}
\RightLabel{Res$\left(\mu_1\right) $}
\BinaryInfC{$ \vdash  \leqpred{Y(p(w),s(k))}{k}, \eqpred{Y(p(w),s(k))}{k} $}
\noLine
\UnaryInfC{$(1)$}
\end{prooftree}
\end{small}
\begin{small}
\begin{prooftree}
\AxiomC{$(1)$}
\AxiomC{$\left( \delta_2,\left\lbrace \begin{array}{cc}
 r\leftarrow p(p(r)) & q\leftarrow 0 \end{array}\right\rbrace \right)$}
\noLine
\UnaryInfC{$  \vdash   \leqpred{Y(w,k)}{k}, \hat{F_5}(\mathbf{X},k,0)$}
\RightLabel{$B\hat{F_5}r$}
\UnaryInfC{$  \vdash   \leqpred{Y(w,k)}{k}, \neqpred{\hat{S}(X_3(k),0)}{k}$}
\RightLabel{$B\hat{S}r$}
\UnaryInfC{$  \vdash   \leqpred{Y(w,k)}{k},  \neqpred{X_3(k)}{k}$}
\RightLabel{$\neg:r$}
\UnaryInfC{$  \eqpred{X_3(k)}{k} \vdash   \leqpred{Y(w,k)}{k}$}
\RightLabel{Res$\left(\mu_2\right) $}
\BinaryInfC{$\vdash  \leqpred{Y(w,k)}{k}$}
\noLine
\end{prooftree}
\end{small}
where $\mu_1 = \left\lbrace\begin{array}{c} X_2(k,0)\leftarrow Y(p(w),s(k))\end{array}\right\rbrace$ and $\mu_2 = \left\lbrace\begin{array}{c} X_3(k) \leftarrow  Y(p(w),s(k))\end{array}\right\rbrace $.
\vspace{1em}
\item[2)] $\rho_2(\delta_1,\vec{X},Y, n,m,w,k,r,q)$ denotes: 
\begin{small}
\begin{prooftree}
\AxiomC{$\left( \delta_4,\left\lbrace \begin{array}{cc}
 w\leftarrow 0 & r\leftarrow p(r) \\ \multicolumn{2}{c}{q\leftarrow 0} \end{array}\right\rbrace \right)$}
\noLine
\UnaryInfC{$  \vdash  \hat{F_2}(\mathbf{X},k,0)$}
\RightLabel{$B\hat{F_2}r$}
\UnaryInfC{$ \vdash \leqpred{X_1(k,0)}{k}\vee \eqpred{X_1(k,0)}{k} $}
\RightLabel{$\vee:r$}
\UnaryInfC{$ \vdash \leqpred{X_1(k,0)}{k}, \eqpred{X_1(k,0)}{k} $}
\RightLabel{$B\hat{S}r$}
\UnaryInfC{$ \vdash \leqpred{X_1(k,0)}{k}, \eqpred{X_1(k,0)}{k} $}
\noLine
\UnaryInfC{(1)}
\end{prooftree}
\begin{prooftree}
\AxiomC{(1)}
\AxiomC{$\left( \delta_2,\left\lbrace \begin{array}{cc}
 w\leftarrow 0 & r\leftarrow p(p(r)) \\ \multicolumn{2}{c}{q\leftarrow 0} \end{array}\right\rbrace \right)$}
\noLine
\UnaryInfC{$  \vdash  \leqpred{Y(0,k)}{k}, \hat{F_5}(\mathbf{X},k,0)$}
\RightLabel{$B\hat{F_5}r$}
\UnaryInfC{$  \vdash  \leqpred{Y(0,k)}{k},  \neqpred{\hat{S}(X_3(k),0)}{k}$}
\RightLabel{$B\hat{S}r$}
\UnaryInfC{$  \vdash  \leqpred{Y(0,k)}{k},  \neqpred{X_3(k)}{k}$}
\RightLabel{$\neg:r$}
\UnaryInfC{$  \eqpred{X_3(k)}{k} \vdash  \leqpred{Y(0,k)}{k}$}
\RightLabel{Res$\left(\mu\right) $}
\BinaryInfC{$\vdash \leqpred{Y(0,k)}{k}$}
\end{prooftree}
\end{small}
where $\mu = \left\lbrace\begin{array}{c}  X_1(k,0)\leftarrow Y(0,k)\ , \  X_3(k)\leftarrow Y(0,k) \end{array}\right\rbrace$.
\end{itemize}
Now let us consider 
\begin{align*}
\Dcal(\delta_2)\colon & ((\rho_1(\delta_2,\vec{X},Y n,m,w,k,r,q,X^{\iota}),S_{7})\oplus\\
& \ (\rho_2(\delta_2,\vec{X},Y, n,m,w,k,r,q),S_{8}) , \vdash  \leqpred{Y(w,k)}{k}, \hat{F_5}(\mathbf{X},k,q))
\end{align*}
where 
\begin{align*}
S_7 \equiv & \left\lbrace \sigma \middle\vert \sigma \in \mathcal{S}\ ,\ r\sigma>0 \right\rbrace\\ 
S_8 \equiv & \left\lbrace \sigma \middle\vert \sigma \in \mathcal{S}\ ,\ r\sigma=0 \right\rbrace
\end{align*}
\begin{itemize}
\item[1)] $\rho_1(\delta_2,\vec{X},Y, n,m,w,k,r,q)$ denotes: 
\begin{tiny}
\begin{prooftree}
\AxiomC{$\left( \delta_3,\left\lbrace \begin{array}{cc}
r\leftarrow p(r) & q\leftarrow s(q) \end{array}\right\rbrace \right)$}
\noLine
\UnaryInfC{$  \vdash  \hat{F_4}(\mathbf{X},k,s(q))$}
\RightLabel{$S\hat{F_4}r$}
\UnaryInfC{$ \vdash (\nleqpred{\hat{S}(X_2(k,s(q)),s(q))}{s(k)}\vee \leqpred{X_2(k,s(q))}{k}\vee  \eqpred{\hat{S}(X_2(k,s(q)),s(q))}{k}) \wedge \hat{F_4}(\mathbf{X},k,q)$}
\RightLabel{$\wedge:r_2$}
\UnaryInfC{$  \vdash \nleqpred{\hat{S}(X_2(k,s(q)),s(q))}{s(k)}\vee \leqpred{X_2(k,s(q))}{k}\vee  \eqpred{\hat{S}(X_2(k,s(q)),s(q))}{k} $}
\RightLabel{$\vee:r$}
\UnaryInfC{$  \vdash \nleqpred{\hat{S}(X_2(k,s(q)),s(q))}{s(k)} , \leqpred{X_2(k,s(q))}{k}\vee  \eqpred{\hat{S}(X_2(k,s(q)),s(q))}{k} $}
\RightLabel{$\vee:r$}
\UnaryInfC{$  \vdash \nleqpred{\hat{S}(X_2(k,s(q)),s(q))}{s(k)}, \leqpred{X_2(k,s(q))}{k},  \eqpred{\hat{S}(X_2(k,s(q)),s(q))}{k} $}
\RightLabel{$\neg:r$}
\UnaryInfC{$  \leqpred{\hat{S}(X_2(k,s(q)),s(q))}{s(k)} \vdash  \leqpred{X_2(k,s(q))}{k},  \eqpred{\hat{S}(X_2(k,s(q)),s(q))}{k} $}
\noLine
\UnaryInfC{$(2)$}

\end{prooftree}
\end{tiny}
\begin{small}
\begin{prooftree}
\AxiomC{$\left( \delta_1,\left\lbrace \begin{array}{cc}
w\leftarrow p(w) & k\leftarrow s(k)\\ r\leftarrow s(m) & q\leftarrow 0 \end{array}\right\rbrace \right) $}
\noLine
\UnaryInfC{$ \vdash \leqpred{Y(p(w),s(k))}{s(k)}$}
\AxiomC{$(2)$}
\RightLabel{Res$\left(\mu_1\right) $}
\BinaryInfC{$ \vdash \leqpred{Y(w,k)}{k}, \eqpred{\hat{S}(Y(w,k),s(q))}{k} $}
\noLine
\UnaryInfC{$(1)$}
\end{prooftree}
\end{small}
\begin{small}
\begin{prooftree}
\AxiomC{$(1)$}
\AxiomC{$\left( \delta_2,\left\lbrace \begin{array}{cc} r\leftarrow p(r) & q\leftarrow s(q) \end{array}\right\rbrace \right) $}
\noLine
\UnaryInfC{$  \vdash  \leqpred{Y_1(w,k)}{k}, \hat{F_5}(\mathbf{X},k,s(q))$}
\RightLabel{$S\hat{F_5}r$}
\UnaryInfC{$  \vdash  \leqpred{Y_1(w,k)}{k},  \neqpred{X_3(k)}{k}\vee \hat{F_5}(\mathbf{X},k,q)$}
\RightLabel{$\vee:r$}
\UnaryInfC{$  \vdash  \leqpred{Y_1(w,k)}{k},  \neqpred{X_3(k)}{k}, \hat{F_5}(\mathbf{X},k,q)$}
\RightLabel{$\neg:r$}
\UnaryInfC{$  \eqpred{X_3(k)}{k} \vdash  \leqpred{Y_1(w,k)}{k}, \hat{F_5}(\mathbf{X},k,q)$}
\RightLabel{Res$\left( \mu_2\right) $}
\BinaryInfC{$\vdash  \leqpred{Y(w,k)}{k}, \hat{F_5}(\mathbf{X},k,q)$}
\end{prooftree}
\end{small}
where $\mu_1 = \left\lbrace\begin{array}{c} X_2(k,s(q))\leftarrow Y(w,k),\ Y(p(w),s(k)) \leftarrow \hat{S}(Y(w,k),s(q)) \end{array}\right\rbrace$, and $\mu_2 = \left\lbrace\begin{array}{c} Y_1(w,k)\leftarrow Y(w,k),\ X_3(k) \leftarrow \hat{S}(Y(w,k),s(q)) \end{array}\right\rbrace$.\vspace{1em}  
\item[2)] $\rho_2(\delta_2,\vec{X},Y, n,m,w,k,r,q)$ denotes: 
\begin{tiny}

\begin{prooftree}
\AxiomC{$\left( \delta_3,\left\lbrace \begin{array}{cc}
r\leftarrow 0 & q\leftarrow s(q) \end{array}\right\rbrace \right) $}
\noLine
\UnaryInfC{$  \vdash  \hat{F_4}(\mathbf{X},k,s(q))$}
\RightLabel{$S\hat{F_4}r$}
\UnaryInfC{$ \vdash ( \nleqpred{\hat{S}(X_2(k,s(q)),s(q))}{s(k)}\vee \leqpred{X_2(k,s(q))}{k}\vee  \eqpred{\hat{S}(X_2(k,s(q)),s(q))}{k}) \wedge \hat{F_4}(\mathbf{X},k,q)$}
\RightLabel{$\wedge:r_2$}
\UnaryInfC{$  \vdash \nleqpred{\hat{S}(X_2(k,s(q)),s(q))}{s(k)}\vee \leqpred{X_2(k,s(q))}{k} \vee  \eqpred{\hat{S}(X_2(k,s(q)),s(q))}{k} $}
\RightLabel{$\vee:r$}
\UnaryInfC{$  \vdash \nleqpred{\hat{S}(X_2(k,s(q)),s(q))}{s(k)}, \leqpred{X_2(k,s(q))}{k}\vee  \eqpred{\hat{S}(X_2(k,s(q)),s(q))}{k} $}
\RightLabel{$\vee:r$}
\UnaryInfC{$  \vdash \nleqpred{\hat{S}(X_2(k,s(q)),s(q))}{s(k)}, \leqpred{X_2(k,s(q))}{k},  \eqpred{\hat{S}(X_2(k,s(q)),s(q))}{k} $}
\RightLabel{$\neg:r$}
\UnaryInfC{$  \leqpred{\hat{S}(X_2(k,s(q)),s(q))}{s(k)} \vdash  \leqpred{X_2(k,s(q))}{k},  \eqpred{\hat{S}(X_2(k,s(q)),s(q))}{k} $}
\noLine
\UnaryInfC{$(2)$}
 \end{prooftree}
 \end{tiny}
\begin{small}
\begin{prooftree}
\AxiomC{$\left( \delta_1,\left\lbrace \begin{array}{cc}
w\leftarrow p(w) & k\leftarrow s(k)\\ r\leftarrow s(m) & q\leftarrow 0 \end{array}\right\rbrace \right)$}
\noLine
\UnaryInfC{$ \vdash \leqpred{Y(p(w),s(k))}{s(k)}$}
\AxiomC{$(2)$}
\RightLabel{Res$\left(\mu_1\right) $}
\BinaryInfC{$ \vdash \leqpred{Y(w,k)}{k}, \eqpred{\hat{S}(Y(w,k),s(q))}{k} $}
\noLine
\UnaryInfC{$(1)$}
\end{prooftree}
\end{small}
\begin{small}
\begin{prooftree}
\AxiomC{$(1)$}
\AxiomC{$\left( \delta_5,\left\lbrace \begin{array}{cc}
 r\leftarrow 0 & q\leftarrow s(q) \end{array}\right\rbrace \right) $}
\noLine
\UnaryInfC{$  \vdash   \hat{F_3}(\mathbf{X},k,s(q))$}
\RightLabel{$S\hat{F_3}r$}
\UnaryInfC{$  \vdash   \hat{F_5}(\mathbf{X},k,s(q)) \wedge   \hat{F_4}(\mathbf{X},k,s(q))  \wedge \hat{F_3}(\mathbf{X},p(k),s(q))$}
\RightLabel{$\wedge:r$}
\UnaryInfC{$  \vdash   \hat{F_5}(\mathbf{X},k,s(q)) \wedge   \hat{F_4}(\mathbf{X},p(k),s(q)) $}
\RightLabel{$\wedge:r$}
\UnaryInfC{$  \vdash   \hat{F_5}(\mathbf{X},k,s(q))$}
\RightLabel{$S\hat{F_5}r$}
\UnaryInfC{$  \vdash    \neqpred{X_3(k)}{k}\vee \hat{F_5}(\mathbf{X},k,q)$}
\RightLabel{$\vee:r$}
\UnaryInfC{$  \vdash    \neqpred{X_3(k)}{k}, \hat{F_5}(\mathbf{X},k,q)$}
\RightLabel{$\neg:r$}
\UnaryInfC{$  \eqpred{X_3(k)}{k} \vdash   \hat{F_5}(\mathbf{X},k,q)$}
\RightLabel{Res$\left(\mu_2 \right) $}
\BinaryInfC{$\vdash  \leqpred{Y(w,k)}{k}, \hat{F_5}(\mathbf{X},k,q)$}
\end{prooftree}
\end{small}
\end{itemize}
where $\mu_1 = \left\lbrace\begin{array}{c} X_2(k,s(q))\leftarrow Y(w,k),\ Y(p(w),s(k))\leftarrow \hat{S}(Y(w,k),s(q))\end{array}\right\rbrace  $, and $\mu_2 = \left\lbrace\begin{array}{c} X_3(k)\leftarrow \hat{S}(Y(w,k),s(q))\end{array}\right\rbrace  $
\vspace{1em}
Now let us consider  
\begin{align*}
\Dcal(\delta_3)\colon & ((\rho_1(\delta_3,\vec{X},Y,n,m,w,k,r,q),S_{7})\oplus \\
& \ (\rho_2(\delta_3,\vec{X},Y, n,m,w,k,r,q),S_{8}) ,  \vdash  \hat{F_4}(\mathbf{X},k,q))
\end{align*}

\begin{itemize}
\item[1)] $\rho_1(\delta_3,\vec{X},Y, n,m,w,k,r,q)$ denotes:
\begin{small}
\begin{prooftree}

\AxiomC{$\left( \delta_3,\left\lbrace \begin{array}{cc}
 r\leftarrow p(r) & q\leftarrow s(q) \end{array}\right\rbrace \right)$}
\noLine
\UnaryInfC{$  \vdash   \hat{F_4}(\mathbf{X},k,s(q))$}
\RightLabel{$S\hat{F_4}r$}
\UnaryInfC{$  \vdash M \wedge \hat{F_4}(\mathbf{X},k,q) $}
\RightLabel{$\wedge:r$}
\UnaryInfC{$  \vdash  \hat{F_4}(\mathbf{X},k,q)$}
\end{prooftree}
\end{small}
\noindent where $M$ denotes $$(\neg  \leqpred{\hat{S}(X_2(k,s(q)),s(q))}{s(k)}\vee \leqpred{X_2(k,s(q))}{k} \vee$$  $$\eqpred{\hat{S}(X_2(k,s(q)),s(q))}{k} ).$$
\item[2)] $\rho_2(\delta_3,\vec{X},Y,n,m,w,k,r,q)$ denotes
\begin{small}
\begin{prooftree}

\AxiomC{$\left( \delta_5,\left\lbrace \begin{array}{cc}
w\leftarrow p(w) & k\leftarrow s(k) \\ \multicolumn{2}{c}{r\leftarrow 0} \end{array}\right\rbrace \right) $}
\noLine
\UnaryInfC{$  \vdash  \hat{F_3}(\mathbf{X},s(k),q)$}
\RightLabel{$S\hat{F_3}r$}
\UnaryInfC{$  \vdash   \hat{F_5}(\mathbf{X},s(k),q) \wedge \hat{F_4}(\mathbf{X},k,q) \wedge \hat{F_3}(\mathbf{X},k,q)$}
\RightLabel{$\wedge:r$}
\UnaryInfC{$  \vdash   \hat{F_5}(\mathbf{X},s(k),q) \wedge \hat{F_4}(\mathbf{X},k,q)$}
\RightLabel{$\wedge:r$}
\UnaryInfC{$  \vdash  \hat{F_4}(\mathbf{X},k,q)$}
\end{prooftree}
\end{small}
\end{itemize}

Now let us consider 
\begin{align*}
\Dcal(\delta_4)\colon & ((\rho_1(\delta_4,\vec{X},Y, n,m,w,k,r,q),S_{7})\oplus\\ & \ (\rho_2(\delta_4,\vec{X},Y, n,m,w,k,r,q),S_{8}) ,  \vdash  \hat{F_2}(\mathbf{X},k,q))
\end{align*}
\begin{itemize}
\item[1)] $\rho_1(\delta_4,\vec{X}, n,m,w,k,r,q)$ denotes
\begin{small}
\begin{prooftree}
\AxiomC{$\left( \delta_4,\left\lbrace \begin{array}{cc}
r\leftarrow p(r) & q\leftarrow s(q) \end{array}\right\rbrace \right)$}
\noLine
\UnaryInfC{$  \vdash   \hat{F_2}(\mathbf{X},k,s(q))$}
\RightLabel{$S\hat{F_2}r$}
\UnaryInfC{$  \vdash ( \eqpred{\hat{S}(X_1(k,s(q)),s(q))}{k} \ \vee \  \leqpred{X_1(k,s(q))}{k} ) \wedge \hat{F_2}(\mathbf{X},k,q) $}
\RightLabel{$\wedge:r$}
\UnaryInfC{$  \vdash  \hat{F_2}(\mathbf{X},k,q)$}
\end{prooftree}
\end{small}
\item[2)] $\rho_2(\delta_4,\vec{X},Y, n,m,w,k,r,q)$ denotes 
\begin{small}
\begin{prooftree}
\AxiomC{$  \vdash   \hat{F_1}(\mathbf{X},k,q)$}
\RightLabel{$S\hat{F_1}r$}
\UnaryInfC{$  \vdash   \hat{F_2}(\mathbf{X},k,q) \wedge \hat{F_3}(\mathbf{X},k,q)$}
\RightLabel{$\wedge:r$}
\UnaryInfC{$  \vdash  \hat{F_2}(\mathbf{X},k,q)$}
\end{prooftree}
\end{small}
\end{itemize}

Now let us consider 
\begin{align*}
\Dcal(\delta_5)\colon & ((\rho_1(\delta_5,\vec{X}, Y n,m,w,k,r,q),S_{5})\oplus \\ &\ (\rho_2(\delta_5,\vec{X}, Y n,m,w,k,r,q),S_{6}) ,  \vdash  \hat{F_3}(\mathbf{X},k,q))
\end{align*}

\begin{itemize}
\item[1)] $\rho_1(\delta_5,\vec{X},Y,n,m,w,k,r,q)$ denotes
\begin{small}
\begin{prooftree}
\AxiomC{$\left( \delta_5,\left\lbrace \begin{array}{cc}
w\leftarrow p(w) & k\leftarrow s(k) \end{array}\right\rbrace \right)$}
\noLine
\UnaryInfC{$  \vdash  \hat{F_3}(\mathbf{X},s(k),q)$}
\RightLabel{$S\hat{F_3}r$}
\UnaryInfC{$  \vdash \hat{F_5}(\mathbf{X},s(k),q) \wedge \hat{F_4}(\mathbf{X},k,q) \wedge \hat{F_3}(\mathbf{X},k,q) $}
\RightLabel{$\wedge:r$}
\UnaryInfC{$  \vdash  \hat{F_3}(\mathbf{X},k,q)$}
\end{prooftree}
\end{small}
\item[2)] $\rho_2(\delta_5,\vec{X},Y,n,m,w,k,r,q)$ denotes
\begin{small}
\begin{prooftree}
\AxiomC{$  \vdash   \hat{F_1}(\mathbf{X},k,q)$}
\RightLabel{$S\hat{F_1}r$}
\UnaryInfC{$  \vdash   \hat{F_2}(\mathbf{X},k,q) \wedge \hat{F_3}(\mathbf{X},k,q) $}
\RightLabel{$\wedge:r$}
\UnaryInfC{$  \vdash  \hat{F_3}(\mathbf{X},k,q)$}
\end{prooftree}
\end{small}
\end{itemize}
Finally, let us consider 
\begin{align*}
\Dcal(\delta_6)\colon & ((\rho_1(\delta_6,\vec{X},Y, n,m,w,k,r,q),S_{5})\oplus\\ & (\rho_2(\delta_6,\vec{X},Y, n,m,w,k,r,q),S_{6}) ,  \vdash  \leqpred{Y(w,k)}{k})
\end{align*} 
\begin{itemize}
\item[1)] $\rho_1(\delta_6,\vec{X},Y, n,m,w,k,r,q)$ denotes
\begin{small}
\begin{prooftree}
\AxiomC{$\left( \delta_3,\left\lbrace \begin{array}{cc}
r\leftarrow p(r) & q\leftarrow 0 \end{array}\right\rbrace \right)$}
\noLine
\UnaryInfC{$  \vdash  \hat{F_4}(\mathbf{X},k,0)$}
\RightLabel{$B\hat{F_4}r$}
\UnaryInfC{$  \vdash \nleqpred{X_2(k,0}{s(k)}\vee \leqpred{X_2(k,0)}{k}\vee  \eqpred{X_2(k,0)}{k} $}
\RightLabel{$\vee:r$}
\UnaryInfC{$  \vdash \nleqpred{X_2(k,0)}{s(k)} ,  \leqpred{X_2(k,0)}{k}\vee  \eqpred{X_2(k,0)}{k} $}
\RightLabel{$\vee:r$}
\UnaryInfC{$  \vdash \nleqpred{X_2(k,0)}{s(k)} ,  \leqpred{X_2(k,0)}{k},  \eqpred{X_2(k,0)}{k} $}
\RightLabel{$\neg:r$}
\UnaryInfC{$   \leqpred{X_2(k,0)}{s(k)} \vdash   \leqpred{X_2(k,0)}{k},  \eqpred{X_2(k,0)}{k} $}
\noLine
\UnaryInfC{(2)}
\end{prooftree}
\end{small}
\begin{small}
\begin{prooftree}
\AxiomC{$\left( \delta_1,\left\lbrace \begin{array}{cc}
w\leftarrow p(w) & k\leftarrow s(k)\\ \multicolumn{2}{c}{q\leftarrow 0} \end{array}\right\rbrace \right)$}
\noLine
\UnaryInfC{$ \vdash \leqpred{Y(p(w),s(k))}{s(k)}$}
\AxiomC{$(2)$}
\RightLabel{Res$\left(\mu_1\right) $}
\BinaryInfC{$ \vdash  \leqpred{Y(p(w),s(k))}{k}, \eqpred{Y(p(w),s(k))}{k} $}
\noLine
\UnaryInfC{$(1)$}
\end{prooftree}
\end{small}
\begin{small}
\begin{prooftree}
\AxiomC{$(1)$}
\AxiomC{$\left( \delta_5,\left\lbrace \begin{array}{cc}
r\leftarrow p(r) & q\leftarrow 0 \end{array}\right\rbrace \right)$}
\noLine
\UnaryInfC{$\vdash \hat{F_3}(\mathbf{X},0,m)$}
\RightLabel{$B\hat{F_3}r$}
\UnaryInfC{$ \vdash \hat{F_5}(\mathbf{X},k,0)\wedge \nleqpred{a}{0}$}
\RightLabel{$\wedge:r_2$}
\UnaryInfC{$  \vdash   \hat{F_5}(\mathbf{X},k,0)$}
\RightLabel{$B\hat{F_5}r$}
\UnaryInfC{$  \vdash   \neqpred{\hat{S}(X_3(k),0)}{k}$}
\RightLabel{$B\hat{S}r$}
\UnaryInfC{$  \vdash   \neqpred{X_3(k)}{k}$}
\RightLabel{$\neg:r$}
\UnaryInfC{$  \eqpred{X_3(k)}{k} \vdash  $}
\RightLabel{Res$\left(\mu_2  \right) $}
\BinaryInfC{$\vdash  \leqpred{Y(w,k)}{k}$}
\end{prooftree}
\end{small}
where $\mu_1 = \left\lbrace\begin{array}{c} X_2(k,0)\leftarrow Y(p(w),s(k))\end{array}\right\rbrace  $ and\\ $\mu_2 = \left\lbrace\begin{array}{c}  Y(p(w),s(k))\leftarrow Y(w,k)\ , \ X_3(k)\leftarrow Y(w,k)\end{array}\right\rbrace  $. 
\item[2)] $\rho_2(\delta_6,\vec{X},Y n,m,w,k,r,q)$ denotes:
\begin{small}
\begin{prooftree}

\AxiomC{$\left( \delta_4,\left\lbrace \begin{array}{cc}
r\leftarrow p(r) & q\leftarrow 0 \end{array}\right\rbrace \right)$}
\noLine
\UnaryInfC{$  \vdash  \hat{F_2}(\mathbf{X},k,0)$}
\RightLabel{$B\hat{F_2}r$}
\UnaryInfC{$ \vdash \leqpred{X_1(k,0)}{k}\vee \eqpred{X_1(k,0)}{k} $}
\RightLabel{$\vee:r$}
\UnaryInfC{$ \vdash \leqpred{X_1(k,0)}{k}, \eqpred{X_1(k,0)}{k} $}
\RightLabel{$B\hat{S}r$}
\UnaryInfC{$ \vdash \leqpred{X_1(k,0)}{k}, \eqpred{X_1(k,0)}{k} $}
\AxiomC{$\left( \delta_5,\left\lbrace \begin{array}{cc}
r\leftarrow p(r) & q\leftarrow 0 \end{array}\right\rbrace \right)$}
\noLine
\UnaryInfC{$\vdash \hat{F_3}(\mathbf{X},k,0)$}
\RightLabel{$B\hat{F_3}r$}
\UnaryInfC{$ \vdash \hat{F_5}(\mathbf{X},k,0)\wedge \nleqpred{a}{0}$}
\RightLabel{$\wedge:r_2$}
\UnaryInfC{$  \vdash   \hat{F_5}(\mathbf{X},k,0)$}
\RightLabel{$B\hat{F_5}r$}
\UnaryInfC{$  \vdash   \neqpred{\hat{S}(X_3(k),0)}{k}$}
\RightLabel{$B\hat{S}r$}
\UnaryInfC{$  \vdash   \neqpred{X_3(k)}{k}$}
\RightLabel{$\neg:r$}
\UnaryInfC{$  \eqpred{X_3(k)}{k} \vdash  $}
\RightLabel{Res$\left(\mu \right) $}
\BinaryInfC{$\vdash \leqpred{Y(w,k)}{k}$}
\end{prooftree}
\end{small}
where $\mu = \left\lbrace\begin{array}{c}  X_1(k,0)\leftarrow Y(w,k)\ , \ X_3(k)\leftarrow Y(w,k) \end{array}\right\rbrace $.
\end{itemize}
\end{example}
In Example~\ref{TheMainExample} we constructed an $\RPLnull^\Psi$ refutation. In the following example we construct a call graph defining the call semantics of the constructed derivation. Note that we will use the partitionings defined in Example~\ref{TheMainExample} in Example~\ref{examplecallgraph}. 
\begin{example}
\label{examplecallgraph}
Let us consider the call graph $\mathcal{G} = \lbrace C_1, C_2,C_3 ,C_4,C_5 ,C_6,C_7\rbrace$, where
\scalebox{.8}{\begin{minipage}{.45\textwidth}
\small
\begin{align*}
C_1 =& \left\lbrace \begin{array}{c} 
\left( \left\lbrace  \begin{array}{c} 
    (\delta_0, n,m) ,\\ 
    (\delta_1,  n,m,n,0,s(m),0),\\ 
    (\delta_5, n,m,n,0,0,m) \end{array} 
\right\rbrace , S_1\right)   \\ 
\left( \left\lbrace  \begin{array}{c} 
    (\delta_0, n,m), \\
    (\delta_1,  0,m,0,0,s(m),0) \end{array}
    \right\rbrace , S_2\right)   \\ 
\left( \left\lbrace  \begin{array}{c} 
    (\delta_0, n,m), \\
    (\delta_5,  n,0,n,0,0,0),\\ 
    (\delta_6, n,0,n,0,s(0),0) \end{array} 
\right\rbrace, S_3 \right)  \\
\left( \left\lbrace  \begin{array}{c} 
    (\delta_0, n,m)\end{array} 
\right\rbrace, S_4 \right)  
\end{array}\right\rbrace
\end{align*}  
\end{minipage}}
\hspace{2em}
\scalebox{.8}{\begin{minipage}{.5\textwidth}
\small
\begin{align*}
C_2=& \left\lbrace \begin{array}{c} 
\left( \left\lbrace  \begin{array}{c} 
    (\delta_1, n,m,w,k,r,q) ,\\ 
    (\delta_1,  n,m,p(w),s(k),r,0),\\ 
    (\delta_2, n,m,w,k,p(p(r)),0),\\
    (\delta_3, n,m,w,k,p(r),0) \end{array} 
\right\rbrace , S_5\right)   \\ 
\left( \left\lbrace  \begin{array}{c} 
    (\delta_1, n,m,w,k,r,q) ,\\ 
    (\delta_4,  n,m,0,k,p(r),0),\\
     (\delta_2, n,m,0,k,p(p(r)),0)\end{array}
    \right\rbrace , S_6\right)  
\end{array}\right\rbrace
\end{align*} 
\end{minipage}}
 
\scalebox{.8}{\begin{minipage}{.45\textwidth}
\small
\begin{align*}
C_3=& \left\lbrace \begin{array}{c} 
\left( \left\lbrace  \begin{array}{c} 
    (\delta_2, n,m,w,k,r,q) ,\\ 
    (\delta_2,  n,m,w,k,p(r),s(q)),\\ 
    (\delta_1, n,m,p(w),s(k),s(m),0),\\
    (\delta_3, n,m,w,k,p(r),s(q)) \end{array} 
\right\rbrace , S_7\right)   \\ 
\left( \left\lbrace  \begin{array}{c} 
    (\delta_2, n,m,w,k,r,q) ,\\ 
    (\delta_1, n,m,p(w),s(k),s(m),0) ,\\ 
    (\delta_3, n,m,w,k,0,s(q)), \\
     (\delta_5, n,m,w,k,0,s(q))\end{array}
    \right\rbrace , S_8\right)  
\end{array}\right\rbrace
\end{align*}  
\end{minipage}}
\hspace{.5em}
\scalebox{.8}{\begin{minipage}{.45\textwidth}
\small
\begin{align*}
C_4=& \left\lbrace \begin{array}{c} 
\left( \left\lbrace  \begin{array}{c} 
    (\delta_3, n,m,w,k,r,q) ,\\ 
    (\delta_3, n,m,w,k,p(r),s(q)) ,\\ 
	\end{array}
    \right\rbrace , S_{7}\right)\\
\left( \left\lbrace  \begin{array}{c} 
    (\delta_3, n,m,w,k,r,q) ,\\ 
    (\delta_5, n,m,p(w),s(k),0,q)
     \end{array} 
\right\rbrace , S_{8}\right)     
\end{array}\right\rbrace
\end{align*}  
\end{minipage}}

\scalebox{.8}{\begin{minipage}{.45\textwidth}
\small
\begin{align*}
C_5=& \left\lbrace \begin{array}{c} 
\left( \left\lbrace  \begin{array}{c} 
    (\delta_4, n,m,w,k,r,q) ,\\ 
    (\delta_4, n,m,w,k,p(r),s(q)) ,\\ 
	\end{array}
    \right\rbrace , S_{7}\right)\\
\left( \left\lbrace  \begin{array}{c} 
    (\delta_4, n,m,w,k,r,q)
     \end{array} 
\right\rbrace , S_{8}\right)     
\end{array}\right\rbrace
\end{align*}  
\end{minipage}}
\hspace{1em}
\scalebox{.8}{\begin{minipage}{.45\textwidth}
\small
\begin{align*}
C_6=& \left\lbrace \begin{array}{c} 
\left( \left\lbrace  \begin{array}{c} 
    (\delta_5, n,m,w,k,r,q) ,\\ 
    (\delta_5, n,m,p(w),s(k),r,q) ,\\ 
	\end{array}
    \right\rbrace , S_{5}\right)\\
\left( \left\lbrace  \begin{array}{c} 
    (\delta_5, n,m,w,k,r,q)
     \end{array} 
\right\rbrace , S_{6}\right)     
\end{array}\right\rbrace
\end{align*}  
\end{minipage}}

\small
\begin{align*}
C_7=& \left\lbrace \begin{array}{c} 
\left( \left\lbrace  \begin{array}{c} 
    (\delta_6, n,m,w,k,r,q) ,\\ 
    (\delta_6,  n,m,p(w),s(k),r,0),\\ 
    (\delta_5, n,m,w,k,p(r),0),\\
    (\delta_3, n,m,w,k,p(r),0) \end{array} 
\right\rbrace , S_5\right)   \\ 
\left( \left\lbrace  \begin{array}{c} 
    (\delta_6, n,m,w,k,r,q) ,\\ 
    (\delta_4,  n,m,0,k,p(r),0),\\
     (\delta_5, n,m,0,k,p(r),0)\end{array}
    \right\rbrace , S_6\right)  
\end{array}\right\rbrace
\end{align*}  
\normalsize
Note that , 
$C_1 \models \Dcal(\delta_0)$,  $C_2 \models \Dcal(\delta_1)$,   $C_3 \models \Dcal(\delta_2)$, and  $C_4 \models \Dcal(\delta_3)$, $C_5 \models \Dcal(\delta_4)$, $C_6 \models \Dcal(\delta_5)$, and $C_7 \models \Dcal(\delta_6)$ and thus $\mathcal{G} \models  \Dcal$ where $\Dcal = \cup_{i} \Dcal(\delta_i)$.

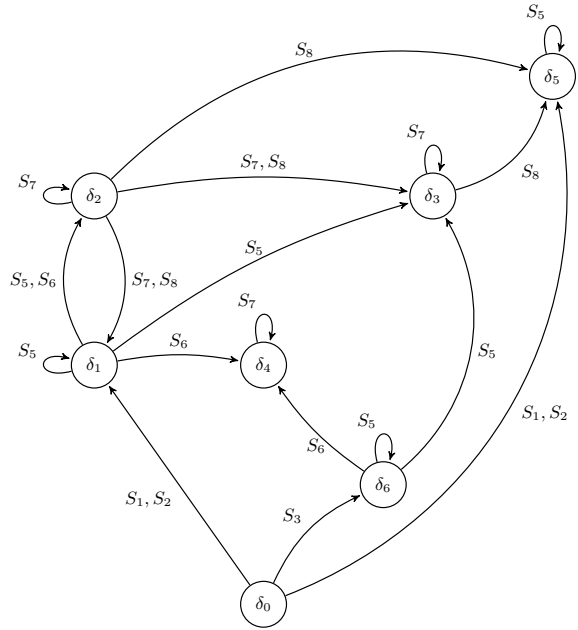
\begin{figure}
\begin{center}
\scalebox{.75}{\begin{tikzpicture}[->,>=stealth',shorten >=1pt,auto,node distance=3cm,
                    semithick]
  \tikzstyle{every state}=[fill=white,draw=black,text=black]

  \node[state]         (A) {$\delta_{0}$};
  \node[state]         (G) [above right of=A]{$\delta_{6}$};
  \node[state]         (E) [above left of=G]{$\delta_{4}$};
  \node[state]         (B) [left of=E]{$\delta_{1}$};
  \node[state]         (C) [above of=B]{$\delta_{2}$};
  \node[state]         (D) [ right of=E, right of=C]{$\delta_{3}$};
  \node[state]         (F) [above right of=D]{$\delta_{5}$};

  \path   (A) edge [bend left=-40] node[ below right ]{$S_1,S_2$} (F);
  \path   (A) edge [bend right=-20] node[ above left ]{$S_3$} (G);
  \path   (A) edge  node[ below left ]{$S_1,S_2$} (B);
  \path   (B) edge [loop left] node[ above left ]{$S_5$} (B);
  \path   (B) edge [bend right=-30] node[ left  ]{$S_5,S_6$} (C);
  \path   (B) edge [bend right=-10] node[ above  ]{$S_5$} (D);
  \path   (B) edge [bend right=-10] node[ above  ]{$S_6$} (E);
  \path   (C) edge [loop left] node[ above left ]{$S_7$} (C);
  \path   (C) edge [bend right=-30] node[ right  ]{$S_7,S_8$} (B);
  \path   (C) edge [bend right=-10] node[ above  ]{$S_7,S_8$} (D);
  \path   (C) edge [bend right=-30] node[ above  ]{$S_8$} (F);
  \path   (D) edge [loop above] node[ above left ]{$S_7$} (D);
  \path   (D) edge [bend left=-30] node[ below right  ]{$S_8$} (F);
  \path   (E) edge [loop above] node[ above left ]{$S_7$} (E);
  \path   (F) edge [loop above] node[ above left ]{$S_5$} (F);
  \path   (G) edge [loop above] node[ above left ]{$S_5$} (G);
  \path   (G) edge [bend right=-10] node[ below  ]{$S_6$} (E);
  \path   (G) edge [bend left=-40] node[ right  ]{$S_5$} (D);

\end{tikzpicture}}
\end{center}
\caption{Call graph from Example~\ref{examplecallgraph} as an automaton.}
\end{figure}
\end{example}

\section{Implementation and Experiments with GradedStrictMonotoneSequenceSchema}

In this section we will discuss our implementations and some experiments using the example $\lstinline[columns=fixed]{GradedStrictMonotoneSequenceSchema}$.
The underlying system of our implementations is Gapt\footnote[1]{ http://www.logic.at/gapt/ } (General Architecture for Proof Theory) \cite{ebner2016system}, which is a framework for implementing proof transformations and provides numerous algorithms for the analysis, transformation, and construction of proofs in various formal calculi. Gapt is implemented in Scala and licensed under the GNU General Public License. The software is available under \url{https://logic.at/gapt}. Gapt initially started as an implementation of the {\CERES} method. The system which provided the foundational architecture for the current version of Gapt was developed for the analysis of F\"{u}rstenberg's proof of the infinitude of primes~\cite{DBLP:journals/tcs/BaazHLRS08}. Gapt also provides an interface for importing proofs from most major theorem provers and exporting proofs and other structures in TPTP format.  

For information on how to install and use the system Gapt we refer to the Gapt User Manual\footnote[2]{http://www.logic.at/gapt/downloads/gapt-user-manual.pdf}. Gapt opens in a Scala interactive shell (scala$>$) which can be used to run all the commands provided by the system. 

In the examples directory of Gapt one can find several example proof schemata, as the proof schema discussed in Example \ref{TheMainExample}. In Gapt we refer to this proof schema as $\lstinline[columns=fixed]{GradedStrictMonotoneSequenceSchema.scala}$. From the Scala interactive shell one can load examples by importing the objects. For instance, the $\lstinline[columns=fixed]{GradedStrictMonotoneSequenceSchema}$ can be imported with the following command:
\begin{verbatim}
scala> import examples.GradedStrictMonotoneSequenceSchema
\end{verbatim}
For many applications we will also need to import the context of the proof schema:
\begin{verbatim}
scala> import examples.GradedStrictMonotoneSequenceSchema.ctx
\end{verbatim}
To display the base case proof and the step case proof of the proof schema $\lstinline[columns=fixed]{GradedStrictMonotoneSequenceSchema}$, we have to access $\lstinline[columns=fixed]{omegaBc}$ and $\lstinline[columns=fixed]{omegaSc}$ and output them in prooftool, which is a viewer for proofs and other elements also implemented in Gapt \cite{dunchev2013prooftool}. The sequence of commands
\begin{verbatim}
scala> val base = GradedStrictMonotoneSequenceSchema.omegaBc
scala> prooftool( base )
\end{verbatim}
stores the base case proof $\lstinline[columns=fixed]{omegaBc}$ in $\lstinline[columns=fixed]{base}$ and outputs the derivation in prooftool, see Figure~\ref{fig:base}.
\begin{figure}
\begin{center}
\fbox{\includegraphics[scale=.5]{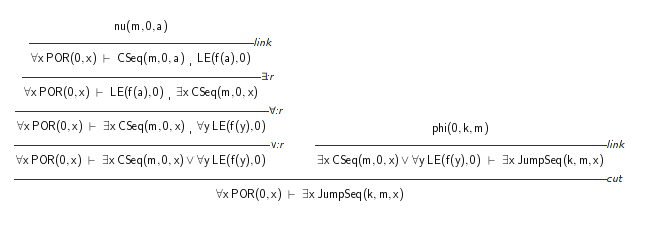}}
\end{center}
\caption{Prooftool output of the base case proof.}
\label{fig:base}
\end{figure}
Similarly, with the sequence of commands below we can output the step case proof in prooftool, see Figure \ref{fig:stepcase}.
\begin{verbatim}
scala> val step = GradedStrictMonotoneSequenceSchema.omegaSc
scala> prooftool( step )
\end{verbatim}
\begin{figure}
\begin{center}
\fbox{\includegraphics[scale=.5]{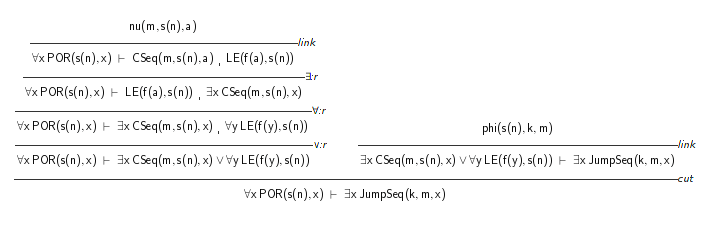}}
\end{center}
\caption{Prooftool output of the step case proof.}
\label{fig:stepcase}
\end{figure}

Note that to obtain the proof schema from Example \ref{TheMainExample} we have to instantiate the parameter $k$ with $0$. In fact, we can instantiate the proof schema $\lstinline[columns=fixed]{GradedStrictMonotoneSequenceSchema}$ with an arbitrary value for parameters $n$ and $m$ (having $k$ fixed). The sequence of commands below instantiates the proof schema with the value $n = 2$ and $m = 1$.
\begin{verbatim}
scala> val proof1 = instantiateProof.Instantiate( le"omega (s (s 0))
                    0 (s 0) " )
scala> prooftool( proof1 )
\end{verbatim}
Note that for this command it is important to import the context $\lstinline[columns=fixed]{ctx}$ of $\lstinline[columns=fixed]{GradedStrictMonotoneSequenceSchema}$ as well! We will not show the output of prooftool here, as the proof is too large to be displayed. Instead, we will output the so-called {\em sunburst view} of the proof. The sunburst view is accessible via prooftool and was introduced to obtain a means of displaying very large proofs. It can be interpreted as a structure which can be unrolled to a proof in tree-like structure. Indeed, the point in the middle of the sunburst corresponds to the end-sequent of a proof in tree-like structure. The different colors represent different types of inferences. Cut rules are displayed in green, structural rules and axioms are displayed in gray, the orange parts correspond to unary logical rules, the yellow ones to binary logical rules, strong quantifier rules are displayed in red and weak quantifier rules in blue. The sunburst of $\lstinline[columns=fixed]{proof1}$ is illustrated in Figure \ref{fig:instance21Sunburst}.
\begin{figure}
\begin{center}
\fbox{\includegraphics[scale=.4]{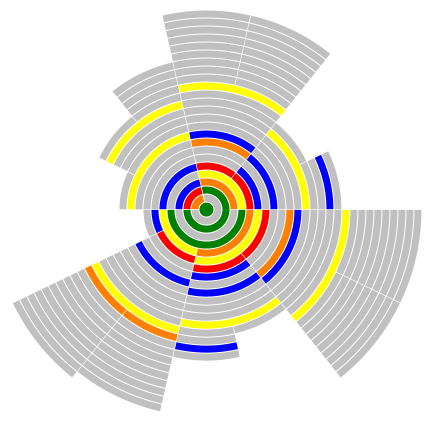}}
\end{center}
\caption{Sunburst output of GradedStrictMonotoneSequenceSchema for parameter $n = 2$ and $m = 1$ ($k=0$).}
\label{fig:instance21Sunburst}
\end{figure}
The sunburst view of the instantiated proof schema for $n=4$ and $m=3$ is displayed in Figure \ref{fig:instance43Sunburst}, we obtain it with the sequence of commands below.
\begin{verbatim}
scala> val proof2 = instantiateProof.Instantiate( le"omega 
                    (s (s (s (s 0)))) 0 (s (s (s 0))) " )
scala> prooftool( proof2 )
\end{verbatim}
\begin{figure}
\begin{center}
\fbox{\includegraphics[scale=.4]{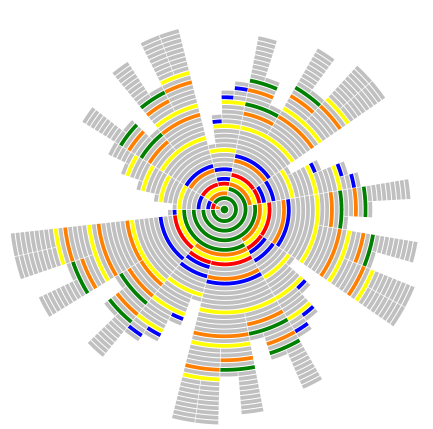}}
\end{center}
\caption{Sunburst output of GradedStrictMonotoneSequenceSchema for parameter $n = 4$ and $m = 3$.}
\label{fig:instance43Sunburst}
\end{figure}

Instantiated proof schemata are essentially {\LK}-proofs extended with an equational theory and thus, any of the proof analytic tools and methods of Gapt can be applied. The command
\begin{verbatim}
scala> val cs = CharacteristicClauseSet(StructCreators.extract(
                proof2))
\end{verbatim}
constructs an instance of the clause representation of the characteristic NNF formula (the running example in this work). It can be also displayed in prooftool. It is also possible to construct a resolution refutation from $\lstinline[columns=fixed]{cs}$ using the command
\begin{verbatim}
scala> val res = SPASS.extendToManySortedViaErasure.getResolution
                 Proof( cs )
scala> prooftool( res )
\end{verbatim}
The output is illustrated in Figure \ref{fig:instance43resolution}. 
\begin{figure}
\begin{center}
\fbox{\includegraphics[scale=.17]{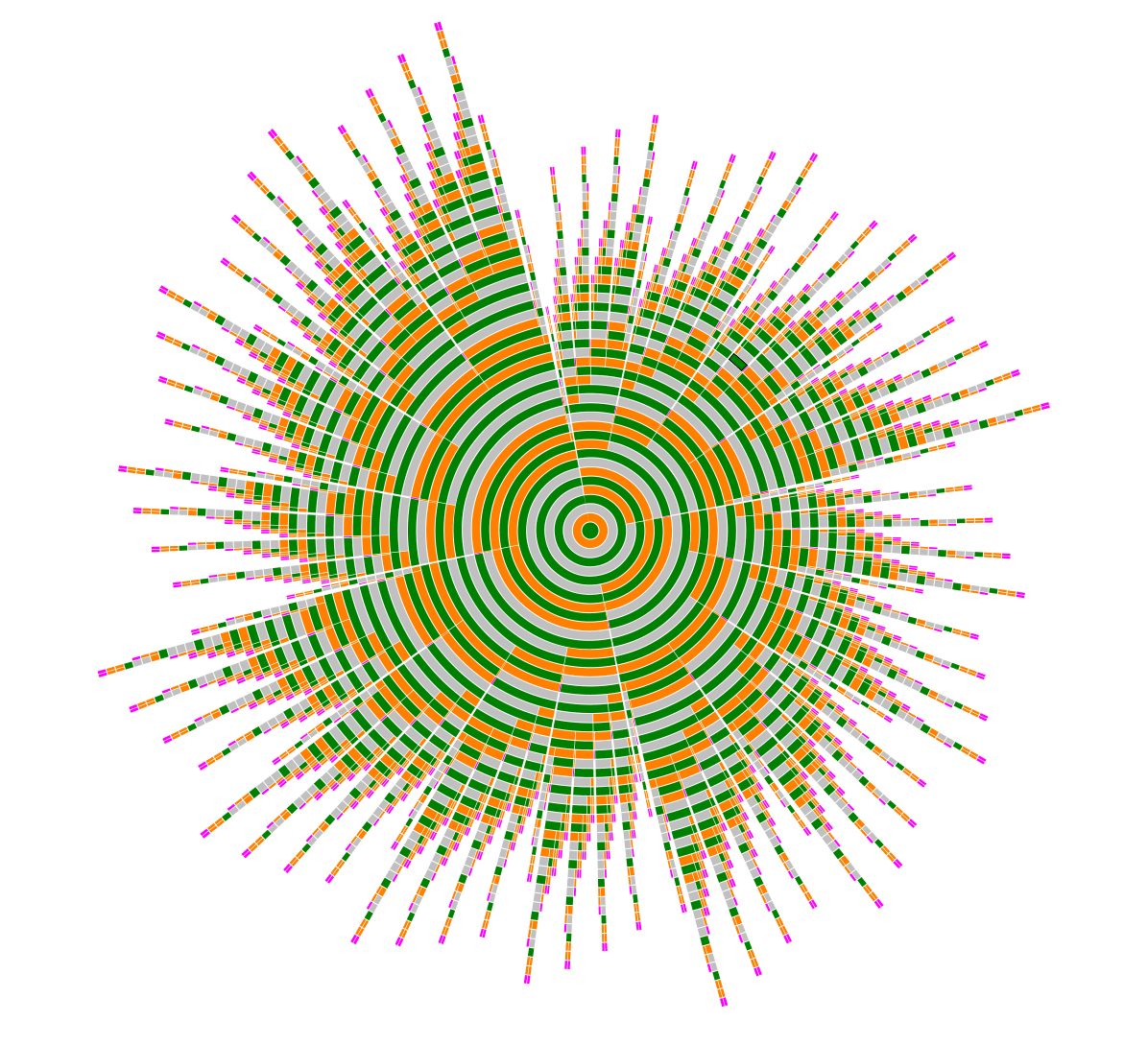}}
\end{center}
\caption{The resolution refutation of the characteristic formula schema of the GradedStrictMonotoneSequenceSchema for parameter $n = 4$ and $m = 3$.}
\label{fig:instance43resolution}
\end{figure}

Moreover, we can extract an expansion proof from the resolution refutation and output it in prooftool with the sequence of commands below.
\begin{verbatim}
scala> val ep = ResolutionToExpansionProof( res.get )
scala> prooftool( ep ) 
\end{verbatim}
The output in prooftool is illustrated in Figure \ref{fig:instance43expansion}.
\begin{figure}
\begin{center}
\fbox{\includegraphics[scale=.4]{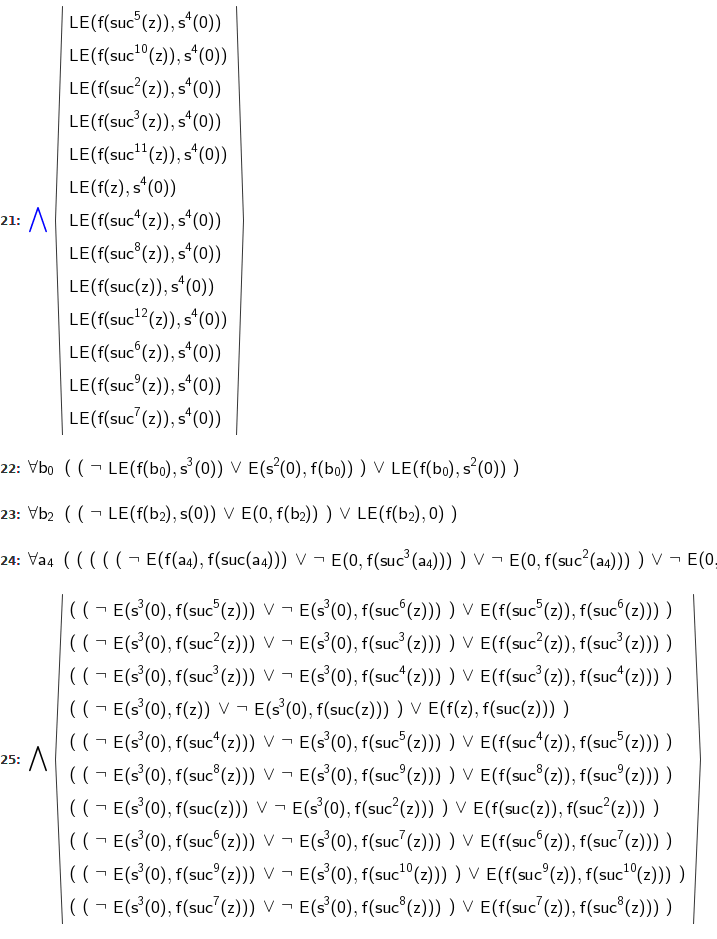}}
\end{center}
\caption{Expansion proof of the resolution refutation of GradedStrictMonotoneSequenceSchema for parameter $n = 4$ and $m = 3$.}
\label{fig:instance43expansion}
\end{figure}
There are also tools specifically designed for uninstantiated proof schemata. One of the most interesting benefits of the interactive features of Gapt is the representation of the cut-structure of a proof schema as an inductive definition. The sequence of commands
\begin{verbatim}
scala> val SCS = SchematicStruct( "omega" ).getOrElse( Map() )
\end{verbatim}
stores the $\lstinline[columns=fixed]{SchematicStruct}$ in $\lstinline[columns=fixed]{SCS}$. Given its type it is not possible to display it in prooftool. With the code
\begin{verbatim}
  val SCS: Map[CLS, ( Struct, Set[Var] )] = SchematicStruct(
           "phi" ).getOrElse( Map() )
  
  val CFPRP = CharFormPRP( SCS )
 
  CharFormPRP.PR( CFPRN )
\end{verbatim}
the schematic characteristic formula is stored in $\lstinline[columns=fixed]{CFPRP}$ and the primitive recursive definitions of the formula is constructed with $\lstinline[columns=fixed]{CharFormPRP.PR( CFPRP )}$.

\section{Future Work and Applications}
The initial intention of this research was to develop a schematic resolution calculus and thus allowing interactive proof analysis using CERES-like methods~\cite{MBaaz2008a} in the presence of induction. More precisely, the resolution calculus introduced in this work will provide the basis for a schematic CERES method more expressive than the methods proposed in \cite{CDunchev2014,ALeitsch2017}. As already indicated, the key to proof analysis using CERES lies in the fact that it provides a bridge between automated deduction and proof theory. In the schematic setting a bridge has been provided~\cite{CDunchev2014,ALeitsch2017}, and the formalism presented here provides a setting to study automated theorem proving for schematic first-order logic. 

Our recursive semantics (Section~\ref{Scaffolding}) separates local resolution derivations from the global ``shape'' of the refutation, an essential characteristic of induction. While constructing a recursive resolution refutation for a recursive unsatisfiable formula is incomplete, it is not clear whether the problem remains incomplete when the call graph is fixed. In other words, we may instead ask: {\em ``Is providing a recursive resolution refutation, with respect to a given call graph, for recursive formulas complete?''} The answer to this question is not so clear in that it depends on the resolution calculus itself as well as the associated unification problem. Both concepts are developed in this paper. 

Concerning the resolution calculus presented in Section \ref{sec:resolution}, both the Andrew's calculus-like sequent rules and the introduction of global variables provide the necessary extensions to resolution accommodating the recursive nature of our formula. The unification problem discussed in Section~\ref{sec.schematicproofs} has not been addressed so far, and furthermore it may have interesting decidable fragments impacting schematic proof analysis as well as other fields. 

Overall, the avenues we leave for future investigations provide ample opportunities for studying schematic theorem proving.

\bibliographystyle{splncs04}
\bibliography{references}

\begin{thebibliography}{10}
\providecommand{\url}[1]{\texttt{#1}}
\providecommand{\urlprefix}{URL }
\providecommand{\doi}[1]{https://doi.org/#1}

\bibitem{DBLP:journals/jacm/Andrews81}
Andrews, P.B.: Resolution in type theory. J. Symb. Log.  \textbf{36}(3),
  414--432 (1971)

\bibitem{VAravantinos2013}
Aravantinos, V., Echenim, M., Peltier, N.: A resolution calculus for
  first-order schemata. Fundamenta Informaticae  (2013)

\bibitem{MBaaz2008a}
Baaz, M., Hetzl, S., Leitsch, A., Richter, C., Spohr, H.: Ceres: An analysis of
  {F}\"{u}rstenberg's proof of the infinity of primes. Theoretical Computer
  Science  \textbf{403}(2-3),  160--175 (Aug 2008).
  \doi{10.1016/j.tcs.2008.02.043},
  \url{http://dx.doi.org/10.1016/j.tcs.2008.02.043}

\bibitem{DBLP:journals/tcs/BaazHLRS08}
Baaz, M., Hetzl, S., Leitsch, A., Richter, C., Spohr, H.: {CERES:} an analysis
  of f{\"{u}}rstenberg's proof of the infinity of primes. Theor. Comput. Sci.
  \textbf{403}(2-3),  160--175 (2008)

\bibitem{MBaaz2000}
Baaz, M., Leitsch, A.: Cut-elimination and redundancy-elimination by
  resolution. Journal of Symbolic Computation  \textbf{29},  149--176 (2000)

\bibitem{JBrotherston2005}
Brotherston, J.: Cyclic proofs for first-order logic with inductive
  definitions. In: Tableaux'05, Lecture Notes in Comp. Sci., vol.~3702, pp.
  78--92 (2005)

\bibitem{JBrotherston2010}
Brotherston, J., Simpson, A.: Sequent calculi for induction and infinite
  descent. Journal of Logic and Computation  \textbf{21}(6),  1177--1216 (2010)

\bibitem{DCerna2016}
Cerna, D.M., Leitsch, A.: Schematic cut elimination and the ordered pigeonhole
  principle. In: Automated Reasoning - 8th International Joint Conference,
  {IJCAR} 2016, Coimbra, Portugal, June 27 - July 2, 2016, Proceedings. pp.
  241--256 (2016), peer Reviewed

\bibitem{dunchev2013prooftool}
Dunchev, C., Leitsch, A., Libal, T., Riener, M., Rukhaia, M., Weller, D.,
  Woltzenlogel-Paleo, B.: Prooftool: a gui for the gapt framework. arXiv
  preprint arXiv:1307.1942  (2013)

\bibitem{CDunchev2014}
Dunchev, C., Leitsch, A., Rukhaia, M., Weller, D.: Cut-elimination and proof
  schemata. In: TbiLLC. Lecture Notes in Computer Science, vol.~8984, pp.
  117--136. Springer (2013)

\bibitem{ebner2016system}
Ebner, G., Hetzl, S., Reis, G., Riener, M., Wolfsteiner, S., Zivota, S.: System
  description: Gapt 2.0. In: International Joint Conference on Automated
  Reasoning. pp. 293--301. Springer (2016)

\bibitem{SHetzl2008}
Hetzl, S., Leitsch, A., Weller, D., Paleo, B.W.: Herbrand sequent extraction.
  In: Intelligent Computer Mathematics. pp. 462--477. Springer (2008)

\bibitem{DBLP:phd/hal/Kersani14}
Kersani, A.: Preuves par induction dans le calcul de superposition. (Induction
  proof in superposition calculus). Ph.D. thesis, Grenoble Alpes University,
  France (2014), \url{https://tel.archives-ouvertes.fr/tel-01551801}

\bibitem{ALeitsch2017}
Leitsch, A., Peltier, N., Weller, D.: {CERES} for first-order schemata. J. Log.
  Comput.  \textbf{27}(7),  1897--1954 (2017). \doi{10.1093/logcom/exx003},
  \url{https://doi.org/10.1093/logcom/exx003}

\bibitem{RMcdowell1997}
Mcdowell, R., Miller, D.: Cut-elimination for a logic with definitions and
  induction. Theoretical Computer Science  \textbf{232} (1997)

\bibitem{Schnorr}
Schnorr, C.: Rekursive Funktionen und ihre Komplexit\"{a}t. B. G. Teubner
  Stuttgart (1974)

\bibitem{GTakeuti1975}
Takeuti, G.: Proof Theory, Studies in logic and the foundations of mathematics,
  vol.~81. American Elsevier Pub. (1975)

\end{thebibliography}

\end{document}